\documentclass[11pt]{article}
\usepackage[left=1in, right=1in, top=1in, bottom=1in, margin=1in]{geometry}

\usepackage{tcolorbox}
\usepackage{ninecolors}
\usepackage{xcolor}

\usepackage{amsmath,amssymb,xspace,graphicx,relsize,bm,breqn,multirow}
\usepackage{multicol}

\usepackage{braket}
\usepackage{qcircuit}
\usepackage{adjustbox}

\usepackage{float}
\usepackage[linesnumbered,ruled,vlined]{algorithm2e}
\SetKwInput{KwInput}{Input}
\SetKwInput{KwOutput}{Output}
\SetKwInOut{Promise}{Promise}
\SetKwInput{Goal}{Goal}
\SetKwProg{Fn}{function}{}{}
\SetKwFor{RepTimes}{repeat}{times}{}
\SetKwFunction{Ver}{Verify}
\SetKwFunction{Prep}{PrepareState}
\SetKwComment{Comment}{/* }{ */}
\usepackage{algorithmic}

\newtcolorbox{myalgorithm}[1][]{
    colback=gray!10, 
    colframe=black, 
    arc=5pt, 
    boxrule=0.5pt, 
    left=0pt, right=0pt, top=0pt, bottom=0pt 
}
\usepackage[margin=1in]{geometry}

\usepackage{amsthm}
\usepackage{dsfont}
\usepackage{array}
\usepackage{makecell}
\newcommand{\Fe}{\ensuremath{\mathcal{F}}}
\newcommand{\Sh}{\ensuremath{\mathsf{Stab}}}

\newcommand{\C}{\ensuremath{\mathcal{C}}}

\newcommand{\A}{\ensuremath{\mathcal{A}}}

\newcommand{\poly}{\ensuremath{\mathsf{poly}}}

\newcommand{\id}{\ensuremath{\mathbb{I}}}

\usepackage{upgreek}
\usepackage{enumerate}

\usepackage[pagebackref]{hyperref}
\usepackage[usestackEOL]{stackengine}
\usepackage{thm-restate,mathrsfs}
\usepackage{enumerate}
\usepackage{array}
\usepackage{parskip}
\def\01{\{0,1\}}

\newcommand{\ketbra}[2]{|#1\rangle\langle#2|}
\hypersetup{
	colorlinks,
	linkcolor={blue!100!black},
	citecolor={red!100!black},
}

\newcommand{\be}{\begin{equation}}
\newcommand{\ee}{\end{equation}}
\newcommand{\ba}{\begin{array}}
\newcommand{\ea}{\end{array}}
\newcommand{\bea}{\begin{eqnarray}}
\newcommand{\eea}{\end{eqnarray}}

\usepackage{mathtools}
\DeclarePairedDelimiter\ceil{\lceil}{\rceil}

\DeclareMathOperator{\Tr}{Tr}
\newcommand{\ra}{\rangle}
\newcommand{\la}{\langle}

\newcommand{\opt}{\textsf{opt}}

\newcommand{\norm}[1]{\left\lVert#1\right\rVert}

\newcommand{\calA}{{\cal A }}
\newcommand{\calB}{{\cal B }}

\newcommand{\calL}{{\cal L }}

\newcommand{\calF}{{\cal F }}

\newcommand{\calC}{{\cal C }}
\newcommand{\calS}{{\cal S }}

\newcommand{\FF}{\mathbb{F}}

\newcommand{\PFR}{\textsf{PFR}}
\newcommand{\LCU}{\textsf{LCU}}

\newcommand{\spann}{\textsf{span}}

\newcommand{\weyl}[1]{\textsf{Weyl}\left({#1}\right)}

\usepackage{amsthm}

\def\01{\{0,1\}}

\newcommand{\wal}{\mathsf{WAL}}

\newcommand{\prep}{\mathsf{prep}}

\newcommand{\SWAP}{\textsf{SWAP}}

\newcommand{\CNOT}{\textsf{CNOT}}

\newcommand{\ridgeproj}{\Pi^{\diamond}}

\definecolor{citegreen}{HTML}{208054}
\definecolor{citeblue}{HTML}{0055cc}

\NineColors{saturation=high}
\hypersetup{
    breaklinks=true,   
    colorlinks=true, 
    linkcolor=blue3, 
    citecolor=green5, 
    urlcolor=blue3, 
}

\newtheorem{theorem}{Theorem}[section]
\newtheorem{definition}[theorem]{Definition}

\newtheorem{lemma}[theorem]{Lemma}
\newtheorem{remark}{Remark}

\newtheorem{fact}[theorem]{Fact}
\newtheorem{claim}[theorem]{Claim}

\usepackage{thm-restate,mathrsfs} 

\global\long\def\argmin{\operatornamewithlimits{argmin}}

\usepackage{tcolorbox}

\makeatletter
\def\widebreve{\mathpalette\wide@breve}
\def\wide@breve#1#2{\sbox\z@{$#1#2$}%
     \mathop{\vbox{\m@th\ialign{##\crcr
\kern0.08em\brevefill#1{0.8\wd\z@}\crcr\noalign{\nointerlineskip}%
                    $\hss#1#2\hss$\crcr}}}\nolimits}
\def\brevefill#1#2{$\m@th\sbox\tw@{$#1($}%
    \hss\resizebox{#2}{\wd\tw@}{\rotatebox[origin=c]{90}{\upshape(}}\hss$}
\makeatletter

\title{Tomography of quantum states with bounded extent}

\begin{document}

\author{
Srinivasan Arunachalam\\
\small IBM Research, Silicon Valley\\
\small \texttt{Srinivasan.Arunachalam@ibm.com}
\and
Arkopal Dutt\\
\small IBM Research, Cambridge \\
\small \texttt{arkopal@ibm.com}
}
\date{}
\maketitle
\begin{abstract}
We give a general framework for tomography of states that have {bounded-extent} with respect to a {structured class} of states. Let $\C$ be a family of $n$-qubit states such that: $(i)$ $\C$~is~succinctly~representable and $(ii)$ there is a \emph{weak} agnostic learner of $\mathcal{C}$. We give a tomography protocol for an unknown state $\ket{\psi}$ that is promised to admit a decomposition of the form
$\
\ket{\psi} = \sum_i c_i \ket{\phi_i}$,
where $\ket{\phi_i} \in \C$ with \emph{bounded} $\ell_1$-norm of the coefficients (which we call \emph{extent}). Our main contribution is to show that a weak agnostic learner for $\mathcal{C}$ can be \emph{boosted} into a tomography algorithm for states with bounded extent with respect to~$\mathcal{C}$.

\vspace{2mm}

Our reduction is black-box and applies broadly across model classes. As an application, when $\mathcal{C}$ is the class of stabilizer states, we obtain tomography algorithms for states with stabilizer extent $\xi$ up to trace distance $\varepsilon$, in time
$\
\poly(n,(\xi/\varepsilon)^{\log(\xi/\varepsilon)}),
$
which is improvable to $
\poly(n,\xi,1/\varepsilon)$ 
assuming the algorithmic polynomial Freiman--Ruzsa conjecture in the high-doubling regime. When the unknown state $\ket{\psi}$ is arbitrary, we give an algorithmic decomposition result in the spirit of a weak regularity lemma for quantum states with respect to $\calC$ and show that the structure in $\ket{\psi}$ that is explainable by $\calC$ can be efficiently learned. Our main conceptual message is that agnostic learning of a structured base class automatically yields learnability of its low-complexity linear span.
\end{abstract}

 \newpage

\setcounter{tocdepth}{2}
{ \tableofcontents}

\newpage 

\section{Introduction}\label{sec:intro}

\textbf{\emph{Quantum inspiration.}} A central theme in quantum tomography is to push efficient learning as far beyond simple classes as possible. From this perspective, a natural candidate frontier is \emph{simulability}: if a family of quantum states admits \emph{efficient classical simulation}, should it also admit \emph{efficient learning} from copies? Stabilizer-based models provide one of the most compelling testbeds for this question. On the one hand, stabilizer structure lies at the heart of some of the strongest classical simulation algorithms for non-Clifford quantum circuits. On the other hand, a growing body of work has shown that stabilizer-like structure can also be exploited algorithmically for learning and tomography in a variety of settings. Together, these developments suggest that stabilizer structure may mark an important boundary between \emph{efficiently simulability} and  \emph{efficiently~learnability}.

A particularly influential approach formalizes “closeness to stabilizer structure” via \emph{stabilizer decompositions}, by expressing a target state $\ket{\psi}$ as a linear combination of stabilizer states. This perspective underlies some of the strongest classical simulation algorithms for non-Clifford circuits and motivates refined measures of non-stabilizerness. Among them, the \emph{stabilizer extent} introduced by~\cite{bravyi2019simulation} quantifies, operationally, how economically $\ket{\psi}$ can be written as $\ket{\psi}=\sum_i c_i\ket{\phi_i}$ for stabilizers  $\ket{\phi_i}$, via the minimum possible $\sum_i |c_i|$, which they call \emph{extent}. Small extent (which we often refer to as ``low-complexity") enables simulation algorithms whose runtime scales polynomially with this parameter. A natural next question is whether the same structural promise also enables \emph{tomography}: given copies of $\ket{\psi}$, can we efficiently reconstruct a useful classical hypothesis whenever the extent is small? Several works have begun to leverage stabilizer structure for learning and tomography in different measurement models~\cite{aaronsontalk,montanaro2017learning,arunachalam2022optimal,grewal2023efficient,grewal2024improved,leone2024learningstab,hangleiter2024bell,leone2024learning,ad2025structure}, but all these works require further assumptions on the circuits producing these states. Understanding whether all states of low stabilizer extent are efficiently learnable remains an outstanding question in quantum learning,  explicitly raised in~\cite{arunachalam2022optimal,grewal2022low,anshu2024survey}. Answering this  is therefore a key step toward understanding what lies between simulability and learnability, motivating the central question of the work:
\begin{quote}
\centering
\emph{Can we efficiently learn an unknown quantum state that admits a low-complexity decomposition over a structured class
$\C$?}\end{quote}

\noindent \textbf{\emph{Classical inspiration.}} This question is not only about tomography, but about \emph{structure learning}. The question above should be viewed as an algorithmic version of a broader \emph{decomposition theme}. Although the target function may not itself belong to 
$\C$, expressing the function as a low-complexity decomposition over $\C$ serves as a natural  proxy for the amount of 
$\C$-structure it contains. Thus, one may ask more broadly: whenever a complex object is not exactly structured, but can be succinctly explained using structured building blocks, can one efficiently recover such an explanation? In classical mathematics and theoretical computer science, this perspective underlies weak regularity lemmas, dense model theorems, hard-core sets, and boosting arguments, where one approximates an arbitrary object by a bounded-complexity combination of simple pieces, leaving behind a residual that is pseudorandom relative to the class of interest~\cite{trevisan2009guest,goldreich2003three,green2008primes,tao2007structure,frieze1999quick,trevisan2009regularity}. 

The question above can be viewed as a ``quantum analogue" of these well-known classical topics.  To this end, a recent work~\cite{ad2025structure}  (which we discuss below) considered the case where $\C$ is the set of stabilizer states. They showed that given access to the state-preparation unitary of the unknown state $\ket{\psi}$, one can learn a compact stabilizer decomposition $\ket{\phi}$ such that the residual state $\ket{\psi} - \ket{\phi}$ has low stabilizer fidelity.  A natural question is whether this phenomenon extends beyond stabilizer states: can we efficiently recover 
$\C$-structure in an arbitrary quantum state  $\ket{\psi}$, and can we learn $\ket{\psi}$ efficiently promised it admits such $\C$-structure?

\subsection{Main results}
To describe our main results, we need a few definitions. Let $\C$ be a class of states. For a state~$\ket{\psi}$,~let
$$
\mathcal{F}_{\cal C}(\ket{\psi}):=\max_{\ket{\phi}\in \C}|\langle \psi|\phi\rangle|^2.
$$
be the fidelity of $\ket{\psi}$ with respect to $\C$ (i.e., quantifies how close $\ket{\psi}$ is to an element in $\C$). For example, when $\C$ is the class of stabilizer states, this is the well-known notion of \emph{stabilizer fidelity}. We say we can \emph{weak} agnostic learn $\C$ if there exists an algorithm that uses  $\calA_{\wal}$ copies of an unknown $\ket{\psi}$ and time $T_{\wal}$ to output a $\ket{\phi'}\in \C$ such that
\begin{align}
\label{eq:weakagnostic}
|\langle \psi|\phi'\rangle|^2\geq \eta\Big( \calF_{\C}\big(\ket{\psi}\big)\Big)
\end{align}
for a (non-decreasing) function $\eta:[0,1]\rightarrow [0,1]$. We call this guarantee \emph{weak} because, in the usual notion of \emph{strong} agnostic learning, one requires $|\langle \psi|\phi'\rangle|^2\geq \calF_{\C}(\ket{\psi})-\varepsilon$, as is common in classical and quantum learning theory.  Finally, we state our results for a particular class of states that we define as a \emph{model class} that we define now.\footnote{We remark that in the definition of the model class, one could also allow cases where $\C$ may not satisfy $(i),(iii)$, but instead there is a weak agnostic learner $\calA_{\wal}$ that outputs states in a different class $\C'$ with the same fidelity promise and $\C'$ satisfies $(i),(iii)$. All our results would also hold in that case.}
\begin{figure}[!ht]
\begin{tcolorbox}
\begin{definition}\label{def:model_class}
\textbf{Model Class $\calC$}\vspace{2mm}
\emph{
\begin{enumerate}[$(i)$]
    \item Every $n$-qubit state $\ket{\phi} \in \calC$ can be described in $\poly(n)$ time.\vspace{2mm}
    \item There is a \emph{weak} agnostic learner $\calA_{\wal}$ using $S_{\wal}$ copies of an unknown state~$\ket{\psi}$ and~in time $T_{\wal}$ outputs a $\ket{\phi} \in \calC$ such that $|\la \phi | \psi \ra|^2 \geq \eta(\calF_{\calC}(\ket{\psi}))$ where $\eta : [0,1] \rightarrow [0,1]$. \vspace{2mm}
    \item For each $\ket{\phi} \in \calC$, there is a classical procedure $\calA_{\prep}$ running in $T_{\prep}$ time that outputs a circuit $V$ with gate complexity $G_{\prep}$, that prepares $\ket{\phi}$, i.e., $V \ket{0^n} = \ket{\phi}$.
\end{enumerate}   
}     
\end{definition}
\end{tcolorbox}
\end{figure}

\noindent Note that several well known classes of states fall in the framework of model class $\C$ above, for example: product states, graph states, matrix product states, stabilizer states. On a high-level our main results are the following (with the formal theorem statements stated~below):
\begin{enumerate}
    \item We first give an algorithmic decomposition theorem, that given copies of $\ket{\psi}$, outputs a list of $\ket{\phi_1},\ldots,\ket{\phi_k}\in \C$  and $\beta \in \mathbb{C}^k,\alpha \in \mathbb{C}$,  such that one can write $\ket{\psi}$ as
$$
  \ket{\psi}=\sum_{i\in [k]} \beta_i \ket{\phi_i}+\alpha\ket{\phi^\perp},
$$
where the residual state $\ket{\phi^\perp}$ satisfies $|\alpha|^2\cdot \calF_{\C}({\ket{\phi^\perp}}) < \varepsilon$.
    \item Using this algorithmic decomposition theorem,  we give a boosting algorithm that boosts a weak agnostic learner (wherein $\eta$ is arbitrary in Eq.~\eqref{eq:weakagnostic}) to a strong agnostic learner.
    \item Using this boosting protocol we give new tomography protocols (up to infidelity $\varepsilon$) for states $\ket{\psi}$ which can be represented as $\ket{\psi}=\sum_i c_i \ket{\phi_i}$ where $\ket{\phi_i}\in \C$ and $\xi=\sum_i |c_i|$ is polynomially bounded.  In particular, we give a quasipolynomial (in $n,\xi,1/\varepsilon$) time algorithm for learning states with stabilizer extent $\xi$ and a $\poly(n,\xi,1/\varepsilon)$-time algorithm  for learning states with MPS extent~$\xi$. 
\end{enumerate}

\paragraph{Prior work.}  We remark that a previous work of~\cite{ad2025structure} considered the setting when specialized to $\C$ being stabilizer states. To this end, they  considered the task of learning states with stabilizer extent $\xi$  and gave protocols with the same complexity as the result above. However, they crucially required two assumptions: $(i)$ they required  access to the unitary $U_\psi$  that prepares the quantum state $\ket{\psi}$ and its controlled version; $(ii)$ they could learn $\ket{\psi}$ only up to constant trace distance. It is unclear if these limitations can be removed using their techniques. Several recent works~\cite{tang2025conjugate,tang2025controlled,liu2024exponential,van2023quantum} looked into the power of having access to $\ket{\psi},\ket{\psi^*}, U_\psi$ controlled-$U_\psi$, and for testing tasks, they showed exponential separations between these access models. Apart from the separations, these works explicitly suggest that any access to $\ket{\psi}$ apart from just copies, can substantially increase the algorithmic power of the algorithm.  These results suggest that the guarantees of~\cite{ad2025structure} may be difficult to improve in the model where one is given only copies of~$\ket{\psi}$.

Our main contribution is that one can remove the $U_\psi$ assumption and obtain a tomography protocol which achieves $\varepsilon$ error in trace distance,  running in quasipolynomial-time (and assuming a conjecture in additive combinatorics we give a $\poly(n,\xi,1/\varepsilon)$ algorithm, completely settling an open question in the field of quantum learning theory).  Prior to our work, we were not aware of efficient tomography algorithms for even states promised to have \emph{constant} stabilizer extent! Furthermore we give a general framework that extends beyond stabilizers (as considered in~\cite{ad2025structure}): for an arbitrary model class $\C$, we show that one can boost agnostic learners for  $\C$ into tomography algorithms for states that can be written as a sum of states in $\C$ with bounded~extent.

\subsection{Technical overview}\label{sec:tech_overview}
The main protocol at the core of our tomography protocols is the following decomposition result applicable to an arbitrary (unknown) $n$-qubit state $\ket{\psi}$. In particular, we show how to use copies of $\ket{\psi}$, to express $\ket{\psi}$ as a compact superposition over states in $\calC$ and a residual state that has low fidelity with all states in $\calC$. This is formally stated below.
\begin{restatable}{theorem}{algodecomposition}
\label{thm:learning_decompositions}
Let $\varepsilon,\delta \in (0,1)$, $\calC$ be a model class as in Definition~\ref{def:model_class} and $\ket{\psi}$ be an unknown~state. 
There is an algorithm $\calL$ that uses copies of $\ket{\psi}$ and with probability $\geq 1-\delta$, outputs a list of $\kappa$ many states $\{\ket{\phi_i}\}_{i \in [\kappa]}$ belonging to $\calC$ and coefficients $\beta \in \mathbb{C}^{\kappa}$ such that $\ket{\psi}$ can be expressed (up to a global phase) as 
$$
\ket{\psi} = \sum_{i=1}^\kappa \beta_i \ket{\phi_i} + \alpha \ket{\phi_R}, \text{ where } \, |\alpha|^2 \cdot \calF_{\calC}(\ket{\phi_R}) \leq \varepsilon,
$$
and $\kappa = \poly(1/\varepsilon,1/\eta(\varepsilon))$. The overall complexity of this algorithm is as follows
\begin{align*}
&\text{Sample complexity: } \widetilde{O}\Big(S_{\wal} \cdot \poly( 1/\varepsilon,1/\eta(\varepsilon),\log 1/\delta) \Big) \\
&\text{Time complexity: } \widetilde{O}\Big(T_{\wal} \cdot \poly( 1/\varepsilon,1/\eta(\varepsilon),\log 1/\delta)) +  S_{\wal} \cdot \poly(G_{\prep},T_{\prep}, 1/\varepsilon,1/\eta(\varepsilon))\Big).
\end{align*} 
\end{restatable}
Below we give a proof sketch of the decomposition theorem above. We start off by giving a high-level idea of the learning algorithm before describing the subroutines and analysis involved.
\paragraph{High-level idea.}

Our algorithm for learning decompositions runs in the following two stages, each of which we discuss in detail subsequently. 
\begin{enumerate}
    \item \emph{Structure learning} where we learn a list of $\kappa$ many states from $\calC$, which we denote by $L = \{\ket{\phi_i}\}_{i \in [\kappa]}$ such that its span captures the structure in $\ket{\psi}$ describable by $\calC$,
    \item \emph{Parameter learning} where we learn a list of $\kappa$ coefficients corresponding to the states in the list $L$ such that the promise of Theorem~\ref{thm:learning_decompositions} is satisfied.
\end{enumerate}

\paragraph{\emph{Stage 1: Structure learning.}}
In the $t$-th stage of structure-learning, the algorithm incrementally builds a list
$
L_t=\{\ket{\phi_1},\ldots,\ket{\phi_t}\}\subseteq \calC
$ 
of states from the model class $\calC$. The iterative algorithm starts with $\ket{\psi_1}=\ket{\psi}$, and at iteration $t$ it applies the weak agnostic learner $\calA_{\wal}$ to a so-called residual state $\ket{\psi_t}$ (which will become clearer shortly), obtaining a new state $\ket{\phi_t}\in \calC$ satisfying
$
|\langle \phi_t | \psi_t\rangle|^2 \ge \eta(\varepsilon),
$
where $\varepsilon$ is the error parameter for this stage and $\eta(\cdot)$ is the guarantee of $\calA_{\wal}$.

We now formally define the residual state $\ket{\psi_t}$  in terms of the structured approximation learned so far. After $t$ iterations, let $\Phi_t$ denote the dictionary matrix whose columns are the states $\ket{\phi_1},\ldots,\ket{\phi_t}$, and let $G_t=\Phi_t^\dagger \Phi_t$ be the corresponding Gram matrix. We then define the current structured approximation to $\ket{\psi}$ by applying the ridge projector
\[
\ridgeproj_{t,\lambda}=\Phi_t(G_t+\lambda \id)^{-1}\Phi_t^\dagger
\]
to $\ket{\psi}$. Equivalently,
$
\ridgeproj_{t,\lambda}\ket{\psi}=\sum_{i=1}^t \beta_i \ket{\phi_i},
$
where the coefficient vector $\vec{\beta}$ is the solution to the regularized least-squares problem
$$
\vec{\beta}
=
\argmin_{\vec a\in\mathbb C^t}
J(\vec a;\lambda),
\qquad
J(\vec a;\lambda)
=
\left\|
\ket{\psi}-\sum_{i=1}^t a_i\ket{\phi_i}
\right\|_2^2
+
\lambda \|\vec a\|_2^2.
$$
The next residual vector is then the unexplained part $
(\id-\ridgeproj_{t,\lambda})\ket{\psi}$, 
and $\ket{\psi_{t+1}}$ is the same with the appropriate normalization. We repeat this process until one of two conditions are satisfied: $(i)$ the norm of the residual state is low, in which case we have succeeded in the stronger task of \emph{tomography} of $\ket{\psi}$ or $(ii)$ if the fidelity of the residual state with respect to $\C$ is small (we discuss both in detail below). Thus, the role of structure learning is primarily combinatorial: it identifies the structured directions $\ket{\phi_1},\ldots,\ket{\phi_t}$ that should appear in the decomposition, but it does not yet attempt to estimate their coefficients accurately.

\noindent \paragraph{\emph{Why the ridge projector?}}
To understand the role of the \emph{ridge projector}, it is helpful to first examine the most natural alternative and why it fails. Suppose that after $t-1$ iterations we have already learned states
$
\ket{\phi_1},\ldots,\ket{\phi_{t-1}}\in \C$. 
A first instinct would be to remove from $\ket{\psi}$ as much of the already-explained structured component as possible, and then apply the weak learner to the remaining residual state. The first natural approach to finding $\ket{\phi_t}$ would be to directly mimic a gradient-descent style update taking the largest step possible in finding the ``best" $\ket{\phi_t}\in \C$ that has highest overlap with the residual state $\ket{\psi_t}$,  instead of using the ridge projector $\ridgeproj_{\lambda}$ which  in fact \emph{slows} down progress. The natural approach leads to the orthogonal projector
\[
\Pi_t=\Phi_t G_t^{-1}\Phi_t^\dagger,
\]
which projects onto $\spann(\{\ket{\phi_i}\}_{i\in[t]})$, where $\Phi_t$ is the dictionary matrix  and $G_t=\Phi_t^\dagger \Phi_t$ is its Gram matrix. The corresponding residual state would then be $\ket{\psi_t}\propto (\id-\Pi_t)\ket{\psi}$, which is orthogonal to the span learned so far and therefore represents the  best possible update.
 
The next step is to prepare copies of the state $\ket{\psi_t}$, for which we would need to implement a unitary encoding the \emph{inverse} of the Gram matrix $G_t$ (via block encoding) and apply it to copies of the state $\ket{\psi}$ (along with the interleaved operators $\Phi,\Phi^\dagger$) or devise a state preparation (say, via $\LCU$) procedure of the state after classically inverting the Gram matrix $G_t$. In either case, the gate complexity of the quantum circuits (to prepare the residual state) would depend on the condition number of $G_t$. Since $\C$ is an arbitrary class with non-orthogonal states, and we have no control on the states learned by the weak learner, potentially the Gram matrix $G_t$ could be ill- conditioned, leading to large gate complexity. In fact, this can be seen for the class of stabilizer states the condition number could  grow in general as $(1/\varepsilon)^{O(1/\varepsilon)}$ (see Appendix~\ref{appsec:barrier_orthoproj}), which isn't desirable.

The ridge projector resolves exactly this issue! In particular, this issue with the condition number of $G$ is  fixed by considering the \emph{regularized} ridge projector $\ridgeproj_{\lambda}$ for $\lambda > 0$ which involves inverting $(G+\lambda \id)$, which is guaranteed to have a minimum eigenvalue of $\lambda$. This still leaves the choice of a \emph{greedy} update which would update the running decomposition by adding a portion of the most recently learned state $\ket{\phi_t}$ proportional to its fidelity with the most recent residual state $\ket{\psi_t}$ i.e., $\la \phi_t | \psi_t \ra \ket{\psi_t}$. The residual state after this update is then $\ket{\psi_{t+1}} \propto (\id - \ket{\phi_t}\bra{\phi_t})\ket{\psi_t}$. While this allows us to learn a decomposition as has been shown in \cite{ad2025structure}, this cannot be immediately utilized for tomography of states with low extent up to any desired distance. The main issue is that the $\ell_2$ norm of the learned decomposition $\sum_i \beta_i \ket{\phi_i}$ cannot be bounded by $1$.

\paragraph{\emph{The finer details.}}
To comment on the overall complexity of the structure learning stage $\calL_1$ of the algorithm, we show the following.
\begin{enumerate}[$1.$]
\item \emph{Residual state preparation.} As part of $\calL_1$, we need the ability to prepare the intermediate residual states $\ket{\psi_t} \propto (\id - \ridgeproj_{t,\lambda})\ket{\psi}$ which is given as input to the agnostic learner. We give a protocol to prepare approximate residual states $\ket{\widetilde{\psi}_t}$ such that $\norm{\ket{\widetilde{\psi}_t} - \ket{\psi_t}} \leq \varepsilon$ by leveraging quantum singular value transformation (QSVT) and block-encodings (i.e., unitaries encoding linear operators that need not be unitary~\cite{gilyen2019qsvt}). In particular, we employ the kernel trick~\cite{bishop2006pattern} to note that $(\id - \ridgeproj_\lambda) = \lambda(K + \lambda \id)^{-1}$ for $\lambda > 0$ where $K=\Phi \Phi^\dagger$ is the kernel matrix. We show that a block-encoding of $K$ can be implemented using state-preparation unitaries generated by $\calA_{\prep}$ of $\calC$. Utilizing fairly standard manipulative operations on block-encodings, we then show that $(\id - \ridgeproj_\lambda)$ can be implemented approximately up to a desired error, using a block-encoding. This then allows us to generate copies of $\ket{\widetilde{\psi}_t}$ with high probability after application on copies of $\ket{\psi}$. 

\item \emph{Checking stopping criteria.} We stop carrying out iterations in the structure learning algorithm if $\|(\id - \ridgeproj_{t,\lambda})\ket{\psi}\| \leq \varepsilon$ i.e., accomplished state tomography or the fidelity of the residual state with class $\calC$ is low, i.e., $\calF_{\calC}(\ket{\psi_t}) \leq \varepsilon$. We show that both of these stopping conditions can be checked efficiently. The probability of success of the above residual state preparation algorithm directly corresponds to the norm of the residual state vector i.e., $\norm{(\id - \ridgeproj_{t,\lambda})\ket{\psi}}$. This probability and hence the norm can then be estimated directly from the measurement outcomes corresponding to the circuit implementing $(\id - \ridgeproj_\lambda)$ (approximately) on to $\ket{\psi}$. To check if the fidelity of the residual state with respect to class $\calC$ is low, we note that the weak agnostic learner $\calA_{\wal}$ only returns a state $\ket{\phi} \in \calC$ with fidelity $|\la \phi|\psi_t \ra|^2 \geq \eta(\varepsilon)$ if $\calF_{\calC}(\ket{\psi_t}) \leq \varepsilon$. We can check if $\calF_{\calC}(\ket{\psi_t}) \leq \varepsilon$ or not by estimating the fidelity of $\ket{\phi}$ with $\ket{\psi_t}$ (which can be done via a simple $\SWAP$ test). 

\item \emph{Upper bound on the number of iterations $\kappa$.} To argue that structure learning converges quickly and that the maximum number of iterations $\kappa$ that need to be executed is bounded, we follow an energy-increment approach~\cite{green2006montreal,tulsiani2014quadratic,ad2025structure}. As we use the ridge projector $\ridgeproj_{\lambda}$ to define the update in the structure learning algorithm, we consider the corresponding loss function
$$
R_{t,\lambda} := \min_{\vec{a} \in \mathbb{C}^t} J_{t}(\vec{a};\lambda) = \min_{\vec{a} \in \mathbb{C}^t} \norm{\ket{\psi} - \Phi_{t} \vec{a}}_2^2 + \lambda \norm{\vec{a}}_2^2,
$$
to be evaluated at the end of iteration $t$ after learning $\ket{\phi_t} \in \calC$. It can be observed that $0 \leq R_{t,\lambda} \leq \norm{\ket{\psi}}_2^2 \leq 1$ for any $t$. We show that over consecutive iterations, we are guaranteed~that
$$
R_{t,\lambda} - R_{t+1,\lambda} \geq \frac{\varepsilon \eta(\varepsilon)}{4(1+\lambda)}, \quad \text{ for all } t\leq \kappa
$$
where $\varepsilon$ is the specified error parameter for structure learning and $\eta(\cdot)$ is the promise of the $\calA_{\wal}$. En route to showing the above result, we show that $\ket{\phi_t}$, which is obtained by agnostic learning on $\ket{\widetilde{\psi}_t}$, also has high fidelity with the true residual state $\ket{\psi_t}$. We then show that the ridge projector update is at least as good as the greedy local update of considering $\la \phi_t | \psi_t \ra$. The final result is then obtained from the promise of the weak agnostic learning and that we have not yet accomplished state tomography of $\ket{\psi}$. Summing the above expression from $t=1,\ldots,\kappa$ gives us $\kappa = O(1/(\varepsilon \eta(\varepsilon)))$ for an appropriate choice of $\lambda$.

\end{enumerate}

\paragraph{\emph{Stage 2: Parameter learning.}} The goal in the second stage is to  convert the dictionary of structured states learned in the first phase into an explicit hypothesis state. After structure learning, we have in our hand a list $L$ such that $\ridgeproj_{L,\lambda}$ is the best ridge projection onto the model class $L$ and the residual state has low fidelity. At this point, we need to determine the coefficients in $\ridgeproj_{L,\lambda}\ket{\psi}=\sum_i \beta_i \ket{\phi_i}$, which is a $\kappa$-dimensional vector. We determine this (up to a global phase) by carrying out tomography on the state $\ket{\beta} \propto (G+\lambda \id)^{-1} \Phi^\dagger \ket{\psi}$ that encodes the coefficients. We show this can be implemented efficiently by re-using the block-encoding constructions from before and use  well-known pure-state tomography protocols which use $O(\kappa/\varepsilon^2)$ copies of the state. Moreover, the circuit used to prepare $\ket{\beta}$ can also be used to estimate $\norm{\beta}$. Putting these together, we can determine $\widehat{\beta}\in \mathbb{C}^\kappa$ such that $\sum_i \widehat{\beta}_i\ket{\phi_i}$ is a good approximation of $\ridgeproj_{L,\lambda}\ket{\psi}$ and the corresponding residual state satisfies the promise of having low fidelity. Putting together the structure learning and parameter learning concludes the proof sketch of our algorithmic decomposition result.

\subsection{Utility of algorithmic decomposition}

\subsubsection{Agnostic boosting}
In computational learning theory, boosting was first considered in the well-known Probably Approximately Correct (PAC) model to answer the fundamental question: if a concept class of Boolean functions can be weakly learned (i.e., an algorithm can learn  an unknown Boolean function with bias $\geq 1/2+1/\textsf{poly}(n)$) is this good enough to strongly learn the class (i.e., an algorithm can learn  an unknown Boolean function with bias $\geq 0.99$) ? Seminal works of Freund and Schapire~\cite{freund:boostfirst,schapire:boostfirst} invented the well-known AdaBoost algorithm for boosting weak learners to strong learners, which made boosting one of the most influential ideas in modern learning theory. Subsequently, researchers considered the (more challenging) agnostic learning model  and demonstrated that boosting can also be done in the agnostic learning model~\cite{ben2001agnostic,kalai2008agnostic,feldman2009distribution}.

Given the recent works in learning states in the quantum agnostic model, a natural question is, can we also boost quantum agnostic learners? Prior to our work, boosting was only considered in the PAC model for \emph{function classes}~\cite{bshouty1995learning,arunachalam2020quantum,izdebski2020improved} and a very recent work,~\cite{adgo2025boosting} considered  a toy version of boosting for states: they considered the \emph{highly structured} class of parity states and showed how to boost weak agnostic learners for parity states to strong learners. However, we are not aware of any general framework for boosting algorithms in the state model. In this work, we give a general framework of boosting  weak agnostic learners for  succinctly describable class of states just given copies.

\begin{restatable}{theorem}{agnboosting}\label{thm:boostingagnlearning}
Let $\delta, \varepsilon \leq \opt \in (0,1)$. Let $\C$ be a model class. Suppose $\ket{\psi}$ is an unknown state with (unknown) optimal fidelity $\calF_{\C}(\ket{\psi}) = \opt$. There is an algorithm, that with probability~$\geq 1-\delta$,  uses copies of $\ket{\psi}$ to output a state $\ket{\phi}$, which is a sum of $\kappa = \poly(1/(\varepsilon\cdot \eta(\varepsilon)))$ many $\ket{\phi_i}\in \C$ such~that
$$
|\la \phi | \psi \ra|^2 \geq \opt - \varepsilon,
$$
The sample, time complexity is polynomial in $n,1/\varepsilon,\log 1/\delta$ and the complexity of the weak learner. 
\end{restatable}
The idea to prove this theorem is as follows:  the algorithmic decomposition theorem first identifies the ``near best" set of states in $\C$ whose span best overlaps with $\ket{\psi}$. In particular,  the algorithmic decomposition theorem to identify a small subspace
$$
L_\kappa=\mathrm{span}\{\ket{\phi_1},\ldots,\ket{\phi_\kappa}\},
$$
generated by states from $\mathcal C$, such that the ridge projection of $\ket{\psi}$ onto $L$ is already a near-optimal agnostic hypothesis. We then show that one  does not need to recover this projection exactly; rather, it suffices to estimate the projection coefficients accurately enough to prepare a nearby normalized state inside the same span, which we show is a valid state satisfying the strong agnostic promise. More concretely defining the (unnormalized)  $\ket{\varphi}=\ridgeproj_{L,\kappa}\ket{\psi}=\sum_{i=1}^{\kappa}\beta_i\ket{\phi_i}$ 
then the decomposition theorem guarantees that
$$
|\langle \varphi | \psi\rangle|^2 \ge \opt-\varepsilon.
$$
Next, we  estimate coefficients $\widehat\beta_i$ in the decomposition and show that for a specified error $\gamma \in (0,1)$
$$
\langle \psi | \sum_{i=1}^{\kappa}\widehat\beta_i\ket{\phi_i}|^2
\ge
|\langle \psi|\ridgeproj_{L,\lambda}|\psi\rangle|^2-\gamma,
\qquad
\alpha:=\|\sum_i\widehat\beta_i\ket{\phi_i}\|_2^2 \le |\langle \psi|\ridgeproj_{L,\lambda}|\psi\rangle|^2+\gamma.
$$
After normalizing, so that $\ket{\widehat{\varphi}}$ is a valid quantum state
$$
\ket{\widehat\varphi}:=\frac{1}{\sqrt{\alpha}}\sum_{i=1}^{\kappa}\widehat\beta_i\ket{\phi_i},
$$
we are able to show that 
$
|\langle \psi | \widehat\varphi\rangle|^2 \ge \opt-\varepsilon
$ 
for a suitable choice of $\gamma$.

\subsubsection{New tomography protocols} 
\paragraph{Stabilizer extent.} In this work, we give a tomography protocol for states with low stabilizer extent: we  show that one can learn these states in quasipolynomial time unconditionally and if we assumed a conjecture in additive combinatorics, we give a polynomial time tomography protocol for these states. We now state this more formally.
\begin{theorem}
\label{result:learn_stab_extent}
Let $\xi > 0,\varepsilon \in (0,1)$. Suppose $\ket{\psi}$ is an unknown $n$-qubit state with stabilizer extent $\xi(\ket{\psi})\leq \xi$. Then, there exists the following algorithms that learn $\ket{\psi}$  up to trace distance $\varepsilon$:
\begin{enumerate}[$(i)$]
      \item $\poly(n,(\xi/\varepsilon)^{\log (\xi/\varepsilon)})$-sample and time algorithm unconditionally,
      \item $\poly(n,\xi,1/\varepsilon)$-sample and time algorithm \emph{assuming}  an algorithmic~$\PFR$~conjecture,\footnote{The \emph{combinatorial version} of the polynomial Freiman-Ruzsa theorem was  resolved in a recent breakthrough~\cite{gowers2023conjecture}. The \emph{algorithmic} analogue of their result remains open in a certain ``high-doubling" regime (for the low-doubling it was recently resolved~\cite{algopfr}).}
\end{enumerate}
Furthermore, every $0.1$-error tomography protocol   needs $\Omega(n+\xi^2)$ copies.
\end{theorem}
We also remark that if we are only concerned with a polynomial sample complexity bound, one can use the stabilizer estimation algorithm in~\cite{grewal2024improved} to obtain an algorithm  using $\poly(n,\xi,1/\varepsilon)$-copies and $2^{n \cdot \poly(\xi/\varepsilon)}$-time algorithm (note that this is better in time complexity than the shadow tomography protocol which would use $2^{n^2\poly(\xi/\varepsilon)}$).

The proof of this theorem follows from our algorithmic decomposition algorithm.  Let $\ket{\psi}=\sum_i c_i\ket{\phi_i}$ be the unknown  state with $\sum_i |c_i|\leq \xi$. The output of our algorithmic decomposition result on input $\ket{\psi}$ (with error set to $\varepsilon'/\poly(\xi$)) is a state  $\ket{\varphi}$ such that stabilizer fidelity of $\ket{\psi}-\ket{\varphi}$ is $\leq \varepsilon'/\poly(\xi)$. In particular, 
$$
|1-\langle \psi|\varphi\rangle|=|\langle \psi|{\psi}\rangle-\langle \psi|\varphi\rangle|\leq \sum_i |c_i|\cdot |\langle \phi_i| \psi-\varphi\rangle|\leq \xi\cdot \varepsilon=\varepsilon',
$$
which is precisely the requirement for tomography of states with low stabilizer extent. However, recall that $\ket{\varphi}$ isn't a valid quantum state. To this end, one needs to renormalize by $\|\ket{\varphi}\|$. Calling the resulting state $\ket{\widehat{\varphi}}$ and using ideas as in the quantum boosting algorithm, one can show that $|\langle \psi|\widehat{\varphi}\rangle|\geq 1-O(\varepsilon)$ thereby satisfying the tomography promise for $\ket{\psi}$. 

\paragraph{Stabilizer rank.} Furthermore, for \emph{stabilizer-rank} $k$ states (i.e., when $\ket{\psi}=\sum_{i \in [k]} c_i \ket{s_i}$ where $\ket{s_i}$ are stabilizer states), one can use the result of $\xi\leq k^k$~\cite{kalra2025stabilizer} to give an unconditional $\poly(n,k^{k^2})$ learning algorithm for stabilizer-rank $k$ states, which is $\poly(n)$ as long as $k\leq O(\sqrt{(\log n)/\log \log n})$. Prior to our work, we were not aware of tomography protocols for even states promised to have stabilizer~rank~$2$!

\paragraph{MPS extent.} 
So far, we looked at stabilizer states, which capture  structured states generated by Clifford circuits.  We now turn to the more expressive class of \emph{matrix product states} (MPS), which describe a much broader class of quantum states through local tensor contractions. This added expressivity allows MPS to represent many entangled states far beyond the stabilizer formalism, making them a central model for low-dimensional quantum many-body systems.  Formally, a matrix product state with \emph{bond dimension r} can be expressed as
$$
\ket{\psi}=\sum_{i\in \{0,1\}^n}\Tr(A^1_{i_1}\cdots A^n_{i_n})\ket{i_1,\ldots,i_n},
$$
where each matrix $A^j_i$ is an $r\times r$ complex matrix for $i\in \{0,1\},j\in [n]$. We denote  the class of states $\ket{\psi}$ satisfying the above as $\textsf{MPS}_{n,r}$.  Since these states are succinctly representable and recent work of~\cite{bakshi2024learning} gave agnostic learning algorithms for MPS, these states naturally fit into the framework of our model class.   Using this observation we obtain a tomography algorithm for the following class
$$
\mathcal{S}_{n,r,\xi}=\{\ket{\phi}=\sum_{i=1}^k c_i\ket{\psi_i}: \ket{\psi_i}\in \textsf{MPS}_{n,r}, \sum_i|c_i|\leq \xi\},
$$
for which we give the following tomography algorithm.
\begin{theorem}
\label{result:learn_MPS_extent}
Let $\xi > 0,\varepsilon \in (0,1)$. Suppose $\ket{\psi}$ is an unknown $n$-qubit state from $\mathcal{S}_{n,r,\xi}$. Then, there is an $\varepsilon$-error tomography protocol that uses $\poly(n,r,\xi,1/\varepsilon)$ many copies of $\ket{\psi}$.
\end{theorem}
We make a remark about our result above.  Note that bond dimension of a sum of $\textsf{MPS}$ states increase additively, so we have that $\mathcal{S}\subseteq \textsf{MPS}_{n,kr}$. Hence, one could have employed the \textsf{MPS} learning algorithms from~\cite{cramer2010efficient,bakshi2024learning,soleimanifar2022testing}  in which case the running time of the algorithm would be $\poly(n,r,k,1/\varepsilon)$.  However, it is conceivable that  $k=\omega(n,r)$ but $\xi=\poly(n,r)$, in which case the complexity of these algorithms is superpolynomial in $n,r$. Although there are superpolynomially many MPS states in the decomposition of $\ket{\psi}$ and  $\sum_i|c_i|$ being polynomially large, suggesting that most of the MPS states in the decomposition have ``negligible" coefficients, we are not aware of any technique to do ``MPS sampling", i.e., find only the ``large weight" MPS states in the decomposition of $\ket{\psi}$.\footnote{In fact this is precisely the reason, why we haven't been able to also learn states with bounded stabilizer rank, since we do not know how to do the analogous ``stabilizer sampling".}  It is thus unclear if these tomography algorithms for MPS states can be improved with the extra assumption that $\xi=\poly(n,r)$.  Our main contribution is that, we can use our general framework of boosting to get a $\poly(n,r,\xi,1/\varepsilon)$   algorithm for $\varepsilon$-learning states with MPS extent $\xi$. 

\subsection{Discussion and open questions}
 Our work opens up several interesting directions for future work.
\begin{enumerate}
\item \emph{Improved dependencies.} For learning stabilizer-extent $\xi$ states, can we improve the  dependence from quasipolynomial in $\xi$ to polynomial in $\xi$, either by resolving the algorithmic $\PFR$ conjecture or using other techniques? Similarly, can we  improve the dependence of our algorithm on the error parameters and constants, which we did not~optimize~here. 
\item \emph{Implications of agnostic learnability.} An inspiration of this work was that in Boolean analysis, agnostic learning  of  depth-$2$ circuits implies PAC learnability of functions produced by depth-$3$ circuits. In turn, agnostic learning of depth-$3$ circuits implies learnability of depth-$4$ circuits. In this paper, analogously, we have observed that agnostic learning of stabilizer states implies learnability of states with low stabilizer extent (or rank). Could we explore this idea further? Are there implications of agnostic learning low stabilizer extent states  beyond what is suggested by our boosting framework? 
Similarly,~can we extend our work to mixed states and consider learning mixtures of Gibbs states of structured~Hamiltonians?
\item \emph{Structured decompositions of quantum unitaries.} Numerous learning algorithms for structured unitaries have exploited its Pauli decomposition (e.g.,~\cite{chen2023juntas,adgp2024degree,grewal2025query}), which is often viewed as the Fourier analysis of unitaries and channels. Recent work has focused on other structural properties of unitaries when tackling problems of testing Clifford unitaries~\cite{hinsche2025clifford,bu2025quantum} and learning bosonic Gaussian unitaries~\cite{fanizza2025efficient}. A natural question inspired from the current work is then: "Can unitaries with compact decompositions over structured classes such as Cliffords or bosonic Gaussian unitaries, be learned efficiently?" 
\item \emph{Continuous variable setting.} More recently, there has been recent progress in improved tomography protocols for fermionic Gaussian states~(e.g., \cite{bittel2025optimal}). Our tomography protocols imply that one can learn unknown states $\ket{\psi}$ that admit low-complexity decompositions over Gaussian states~\cite{bravyi2017complexity,cudby2023gaussian} if given an agnostic learner for Gaussian states. Such Gaussian decompositions have begun to play a role in simulation of matchgate circuits~\cite{cudby2023gaussian} and we show that these can be learned efficiently if there is an efficient agnostic learner. The question is then: ``Can fermionic Gaussian states be agnostically learned efficiently?".
\end{enumerate}

\paragraph{Acknowledgments.} AD thanks Madhur Tulsiani and Fernando Granha Jeronimo for discussions regarding the weak abstract regularity lemma. AD and SA thank Sitan Chen and Jordan Cotler for conversations on states with compact product state structure. ChatGPT Plus contributed routine combinatorial calculations and bounds, which were checked by hand.

\section{Preliminaries}

We will often use the following definitions and facts about matrices.  For a Hermitian matrix $A\in \mathbb{C}^{n\times n}$, we define the operator norm $\|A\|=\max_i |\lambda_i(A)|$. For vectors $v\in \mathbb{C}^n$, by $\|v\|_2$ we mean $\|v\|_2=\sqrt{\sum_i|v_i|^2}$.

\begin{fact}[Weyl's identity]
\label{fact:weyl}
   For Hermitian matrices $A,B\in \mathbb{C}^{n\times n}$ we have that 
   $$\lambda_{min}(A)+\lambda_{min}(B)\leq \lambda_{min}(A+B)\leq \lambda_{max}(A+B)\leq \lambda_{max}(A)+\lambda_{max}(B).
   $$ 
\end{fact}
\begin{fact}[Woodbury matrix identity~\cite{hager1989inverse}]
\label{fact:woodburry}
   Let  $A,U,D,V$ be matrices of compatible dimensions (and as long as involved inverses exist), we have that
$$
(A + UDV)^{-1} = A^{-1} - A^{-1}U(D^{-1} + V A^{-1} U)^{-1} V A^{-1}.
$$
\end{fact}
\begin{fact}[Push-through identity~\cite{henderson1981inverse}]
\label{fact:pushthrough}
   For  matrices $A$ and $B$ of compatible dimensions (and as long as inverses exist), we have that
   $$
   A(BA + \lambda \id)^{-1} = (AB + \lambda I)^{-1}A.
   $$
\end{fact}

\subsection{Stabilizer states}

\paragraph{Paulis and Cliffords.} The $2$-qubit Pauli  matrices are defined as follows
$$\id=\begin{pmatrix}
1 & 0\\
0 & 1
\end{pmatrix}, X=\begin{pmatrix}
0 & 1\\
1 & 0
\end{pmatrix}, Y=\begin{pmatrix}
0 & -i\\
i & 0
\end{pmatrix},Z=\begin{pmatrix}
1 & 0\\
0 & -1
\end{pmatrix}
$$
It is well-known that the $n$-qubit Pauli matrices $\{\id,X,Y,Z\}^n$ form an  {orthonormal basis} for $\mathbb{C}^n$.    In particular, for every $x=(a,b)\in \mathbb{F}_2^{2n}$, one can define the \emph{Weyl operator}
$$
W_x=i^{a\cdot b} (X^{a_1}Z^{b_1}\otimes X^{a_2}Z^{b_2} \otimes \cdots \otimes X^{a_n}Z^{b_n}).
$$  
and these operators $\{W_x\}_{x \in \FF_2^{2n}}$ are orthonormal.  Clifford unitaries are those generated by Hadamard gate $\textsf{Had}=\frac{1}{\sqrt{2}}\begin{pmatrix}
1 & 1\\
1 & -1
\end{pmatrix}$, controlled-$X$ gate and $S=\begin{pmatrix}
1 & 0\\
0 & i
\end{pmatrix}$ gate. The output of Clifford circuits on the all $\ket{0^n}$ input are called \emph{stabilizer states}.

\paragraph{Stabilizer dimension.} We first define the \emph{unsigned stabilizer group} as 
$$
\weyl{\ket{\psi}}=\{x\in \FF_2^{2n}: \langle \psi|W_x|\psi\rangle \in \{-1,1\}\}
$$ to be the Pauli matrices that stabilize $\ket{\psi}$.  We say that an $n$-qubit pure quantum state $\ket{\psi}$ has stabilizer dimension of $k$ if $\ket{\psi}$ is stabilized by an Abelian group of $2^k$ Pauli operators, in other words $\dim(\weyl{\ket{\psi}}) = k$. A stabilizer state has the maximal stabilizer dimension of $n$. Let $\Sh(n-t)$ be the states with stabilizer dimension $n-t$, i.e., if $\ket{\psi}\in \Sh(n-t)$, then $\dim(\weyl{\ket{\psi}})\geq n-t$.

\subsection{Agnostic learning}
We define the model of state agnostic learning. Consider a concept class of states $\C=\{\ket{\psi_1},\ldots,\ket{\psi_m}\}$.
In agnostic learning, an algorithm is given copies of an unknown $\ket{\psi}$ and the goal is to output an $\ket{\phi}\in \C$ such that
$$
|\langle \psi|\phi'\rangle|^2\geq \max_{i\in [m]}|\langle \psi|\phi_i\rangle|^2-\varepsilon.
$$
We also define a \emph{weak} agnostic learner as one, wherein the learner outputs an $\ket{\phi}\in \C$ such that
$$
|\langle \psi|\phi'\rangle|^2\geq \eta\hspace{0.2mm}(\max_{i\in [m]}|\langle \psi|\phi_i\rangle|^2),
$$
where $\eta:[0,1]\rightarrow [0,1]$ is an arbitrary function. The sample complexity of a learner is the total number of copies used by the algorithm to satisfy the above guarantee and the time complexity is the total time used by the algorithm (i.e., number of one and two-qubit quantum gates used in the algorithm). We say an algorithm is \emph{sample and time efficient} these complexities  scale polynomial in $n,1/\varepsilon$ and the the description size of the class. Furthermore,  if the learner outputs $\ket{\phi'}\in \C$, then the learner is called \emph{proper}, else its an \emph{improper} learner.

\subsection{Extent of states with respect to a class}
The notion of \emph{stabilizer extent} was introduced in~\cite{bravyi2016trading}. In this work, we will be concerned with a more general notion of extent (which also captures stabilizer extent). Let $\C$ be a model class of states. For an arbitrary state $\ket{\psi}$, we define the extent of $\ket{\psi}$ with respect to $\C$ as\footnote{We remark that in literature~\cite{bravyi2016trading,bravyi2019simulation} extent was defined the \emph{squared} $\ell_1$ norm. For the purposes of this paper, the squared $\ell_1$ norm or the $\ell_1$ norm as defined below will not change the main results.}
\begin{equation}\label{def:extent_C}
\xi(\ket{\psi})=\min\Big\{ \sum_i |c_i| : \ket{\psi}=\sum_i c_i \ket{\phi_i}, \ket{\phi_i}\in \C\Big\}.    
\end{equation}
Finally, we define 
$$
\C(t)=\{\ket{\psi}:\xi(\ket{\psi})\leq t\},
$$ 
as the set of $\ket{\psi}$ such that $\xi(\ket{\psi})\leq t$. We remark that although $\xi(\cdot)$ depends on the class $\C$, we do not parameterize $\xi$ by $\C$ (to avoid notational clutter); the class $\C$ in consideration will be  clear from context when we discuss the $\xi(\cdot)$ parameter. If $\C$ is the class of stabilizer states, then $\xi(\ket{\psi})$ is precisely the \emph{stabilizer extent} of $\ket{\psi}$.
Similarly, we say a quantum state $\ket{\psi}$ has \emph{stabilizer rank $k$}, if $\ket{\psi}$ can be expressed $\ket{\psi}=\sum_{i=1}^k c_i\ket{s_i}$ where $c_i\in \mathbb{C}$ and $\ket{s_i}$ are stabilizer states. More formally, we define stabilizer rank of a state $\ket{\psi}$ as
$$
\chi(\ket{\psi})=\min\{k: \ket{\psi}=\sum_{i=1}^k c_i \ket{s_i}, \ket{s_i} \text{ are stabilizer states}\}.
$$
Recent results~\cite{mehraban2024improved,mehraban2024quadratic,kalra2025stabilizer} pruned the relationship between  stabilizer extent and rank and we have the following result.  
\begin{theorem}[{\cite{kalra2025stabilizer}}]\label{thm:ub_stab_extent_stab_rank_states}
For an $n$-qubit $\ket{\psi}$  with stabilizer rank $\kappa$, we have
$
\xi(\ket{\psi}) \leq O((2 \kappa)^{(2\kappa + 1)/2}).
$
\end{theorem}

\subsection{Block-encodings}
In this work, we will crucially use a well-known idea of \emph{block-encodings}, which gives a technique to implement an arbitrary linear operator by embedding it into a unitary operator with larger dimensions after appropriate scaling.  The usual definition of block-encoding of square matrices is as~follows.
\begin{definition}[\cite{gilyen2019qsvt}]
\label{def:BE}
Given a matrix $A \in \mathbb{C}^{N \times N}$ with $N=2^n$, if we can find $\alpha \in \mathbb{R}^+$, and an $(m+n)$-qubit unitary matrix $U_A$ such that
$$
\norm{A - \alpha (\la 0^m | \otimes I_N) U_A (\ket{0^m} \otimes I_N)}_2 \leq \varepsilon,
$$
then $U_A$ is called an $(\alpha,m,\varepsilon)$-block-encoding of $A$. In particular, when the block encoding is exact with $\varepsilon=0$, $U_A$ is called an $(\alpha,m)$-block-encoding of $A$.
\end{definition}

We will now describe results that we will require for manipulating block-encodings. Firstly, we will require the ability to execute arithmetic operations on block-encodings. The following lemma collects relevant results from \cite[Lemma~17]{chakraborty2023quantum} and \cite[Lemmas 54, arXiv version]{gilyen2019qsvt}.
\begin{lemma}\label{lem:sum_prod_BEs}
If $A$ has an $(\alpha, a, \varepsilon)$ block encoding with gate complexity $T_A$ and $B$ has a $(\beta, b, \delta)$ block encoding with gate complexity $T_B$, then
\begin{enumerate}[$(i)$]
\item $c_1 A + c_2 B$ for $c_1,c_2 \in \mathbb{R}^{+}$ has an $(c_1 \alpha + c_2 \beta, \max(a,b) + 1, c_1 \varepsilon + c_2 \delta)$-block encoding which is implementable using $O(T_A + T_B)$~gates.
\item $AB$ has an $(\alpha \beta, a + b, \alpha \delta + \beta \varepsilon)$ block encoding with gate complexity $O(T_A + T_B)$.
\end{enumerate}
\end{lemma}
Lastly, we will require the ability to implement a block-encoding of the inverse of $A$. Among the results available~(e.g., \cite{harrow2009hhl,gilyen2019qsvt,chakraborty2019power,chakraborty2023quantum}), we will use the following result that can be shown to follow from \cite[Theorem~56, arXiv]{gilyen2019qsvt}.
\begin{lemma}\label{lem:inverse_BE}
Let $\delta, \varepsilon \in (0,1/2]$. Suppose $A$ is a Hermitian matrix with a $(\alpha,a,0)$ block-encoding $U_A$ with gate complexity $T_A$ and the spectrum of $(A/\alpha)$ lies in $[\delta,1]$. Then, there is a $(2/(\alpha \delta), a + 2, \varepsilon)$ block-encoding of $A^{-1}$ with gate complexity $O\Big( (T_A + a)/\delta \log(2/(\alpha \delta \varepsilon))\Big)$.
\end{lemma}
We will use the following result from \cite[Theorem~26]{chakraborty2023quantum} when we require the ability to implement $A^{-1}$ given an approximate block-encoding of $A$ and not an exact block-encoding.
\begin{lemma}\label{lem:inverse_BE2}
Let $\varepsilon \in (0,1]$. Let $A$ be an invertible normalized matrix with singular values in the range $[1/\kappa_A,1]$ for some $\kappa_A \geq 1$. For every $\varepsilon_1 = o(\varepsilon/(\kappa_A^2\log(\kappa_A/\varepsilon)))$ and $\alpha \geq 2$, let $U_A$ be an $(\alpha,a,\varepsilon_1)$-block-encoding of $A$ implemented in time $T_A$. Then, we can implement a $(2 \kappa_A,a+1,\varepsilon)$-block-encoding of $A^{-1}$ at a cost of
$
O\Big(\kappa_A \alpha \log(\kappa_A/\varepsilon) T_A\Big).
$
\end{lemma}
\noindent Additionally, we will need the following quantum state tomography protocol~\cite{od2016efficient,haah2016sample}.
\begin{theorem}\label{thm:qst}
Let $\varepsilon,\delta \in (0,1)$. Let $\ket{\psi}$ be an unknown $n$-qubit pure state, then there is a tomography algorithm that outputs an estimate $\ket{\phi}$ satisfying $\norm{\ket{\psi} - \ket{\phi}}_2 \leq \varepsilon$ with probability at least \(1-\delta\) using $O(2^n / \varepsilon^2 \cdot \log(1/\delta))$ copies of the state.
\end{theorem}

\section{Algorithmic decomposition theorem} \label{sec:algo_decomposition_thm}
In this section, we prove our main result showing how to learn \emph{structured decompositions} of an arbitrary $n$-qubit state $\ket{\psi}$ with respect to a model class $\calC$ as in Definition~\ref{def:model_class}.  In particular, given copies of $\ket{\psi}$ we give a procedure to construct a decomposition as follows
\begin{equation}
\label{eq:structuredpart}
\ket{\psi} = \underbrace{\sum_{t=1}^k \beta_t \ket{\phi_t}}_{\text{rank-$k$ structured decomposition}} + \,\alpha \underbrace{\ket{\phi_R}}_{\text{unstructured}},
\end{equation}
where $\ket{\phi_t} \in \calC$s are $n$-qubit states belonging to a model class $\calC$ with $\beta \in \calB_\infty^{k},\alpha \in \calB_\infty$,  and $\ket{\phi_R}$ is \emph{unstructured}, by which we mean that its fidelity with $\calC$ i.e., $\calF_{\calC}(\ket{\phi_R})$ is small. Formally we restate and prove the following theorem statement 

\algodecomposition*
We structure the proof of this theorem as follows: first we show how to implement the orthogonal and ridge projections that we discussed in the introduction for structure finding, which will involve several matrix manipulations to construct block encodings of operators. Next we give our structure learning and parameter learning algorithm as part of our algorithmic decomposition proof and prove its correctness and convergence.
\subsection{Projections and their implementation}
In this section, we will discuss the \emph{ridge projector}, a projection that will be crucial in our main algorithm. Before describing the ridge projection, we begin by first discussing the well-known Gram-Schmidt orthonormalization procedure, which gives the intuition eventually for the ridge regression procedure. Eventually we will discuss its implementation via block-encodings. Throughout this section, we will assume $k\geq 2$: in the first iteration of our algorithmic decomposition we directly use our weak learner to find the ``best" state $\ket{\phi_1} \in \C$, so we can assume that we need to perform these projections from the second iteration. 
\subsubsection{Ridge projector}
Let $\calL = \{\ket{\phi_1},\ket{\phi_2},\ldots,\ket{\phi_t}\}$ be a list of $t$ states from $\calC$. Since the states in $\calC$ are classically describable,   one can run the \emph{classical} procedure of Gram-Schmidt orthonormalization to obtain a list of orthonormal states $\calL'=\{\ket{\phi'_1},\ket{\phi'_2},\ldots,\ket{\phi'_t}\}$. In particular, one can express the $\ket{\phi'_i}$s as~follows
\begin{align}
    \ket{\phi_1'} = \ket{\phi_1}, \quad 
    \ket{\phi_2'} = \frac{\ket{\phi_2} - \la \phi_1'| \phi_2 \ra \ket{\phi_1'}}{\norm{\ket{\phi_2} - \la \phi_1'| \phi_2 \ra \ket{\phi_1'}}},\quad 
    \ket{\phi_t'} = \frac{\ket{\phi_t} - \sum_{j=1}^{t-1} \la \phi_j'| \phi_t \ra \ket{\phi_j'}}{\norm{\ket{\phi_t} - \sum_{j=1}^{t-1} \la \phi_j'| \phi_t \ra \ket{\phi_j'}}}
\end{align}
Furthermore, one can  define an \emph{orthogonal projector} onto the span of these $\{\ket{\phi'_j}\}_{j\in [k]}$ as follows
\begin{equation}
    \Pi = \sum_{j=1}^k \ket{\phi_j'}\bra{\phi_j'} = \Phi G^{-1} \Phi^\dagger,
\end{equation}
where $\Phi \in \mathbb{C}^{2^n \times k}$ is the matrix whose $j$th column is $\ket{\phi_j}$ and $G = \Phi^\dagger \Phi$ is the Gram matrix with its elements taking values of $G_{ij} = \la \phi_i | \phi_j \ra$. In particular, $\Pi \ket{\psi}$ coincides with the state obtained by solving the $\ell_2$ minimization problem of Eq.~\eqref{eq:L2_projection_state}. It is not hard to see that,  if $Q = \textsf{span}(\{\ket{\phi_i}\}_{i \in [k]})$ and $\Pi$ is a projection onto $Q$. Then projection of $\Pi$ onto $\ket{\psi}$ is given by
\begin{equation}
\label{eq:L2_projection_state}
    \Pi \ket{\psi} = \sum_{i=1}^k \beta_i \ket{\phi_i} \enspace \text{ where } \enspace \{\beta_i\}_{i \in [k]} = \argmin_{a_1,\ldots,a_k \in \mathbb{C}} \norm{\ket{\psi} - \sum \limits_{i=1}^k a_i \ket{\phi_i}}_2
\end{equation}
Like we mentioned in the introduction, solving this minimization problem is useful to determine the so-called residual state after the algorithm has learned $\mathcal{L}$, and the application of the operator $\Phi G^{-1} \Phi^\dagger$ lets us solve this indirectly. However, there is an issue in applying this operator, we do not have handle on the \emph{minimum eigenvalue} of $G$, which could make the implementation of $G^{-1}$~hard!

To circumvent this issue, we  define the \emph{ridge projector} as 
\begin{equation}\label{eq:ridge_proj}
    \ridgeproj_{\lambda} = \Phi (G + \lambda \id)^{-1} \Phi^\dagger,
\end{equation}
where we used the superscript $^{\diamond}$ to distinguish this projector from the orthogonal projector and the subscript $\lambda$ emphasizes the dependence of this projector on the parameter $\lambda \in (0,1]$. Throughout this section we will assume $\lambda >0$. Note that for $\lambda=0$, this coincides with orthogonal projector $\Pi$. The corresponding minimization problem is that of \emph{ridge regression}
\begin{equation}\label{eq:ridge_regression}
    J(\vec{a} ; \lambda) = \min_{a_1,\ldots,a_k} \norm{\ket{\psi} - \sum_{i=1}^k a_i \ket{\phi_i}}_2^2 + \lambda \norm{a}_2^2,
\end{equation}
where a regularization term is now additionally present compared to Eq.~\eqref{eq:L2_projection_state}. This particular form is commonly also called Tikhonov regularization~\cite{tikhonov1977solutions}. It can be shown that the solution $\vec{\beta} = (\beta_1,\ldots,\beta_k)$ to Eq.~\eqref{eq:ridge_regression} is 
\begin{equation}\label{eq:soln_ridge_regression}
    \vec{\beta} = (G + \lambda \id)^{-1} \Phi^\dagger \ket{\psi}.
\end{equation}
We thus have that the action of the ridge projector by definition is given by
\begin{equation}\label{eq:action_ridge_proj}
    \ridgeproj_{\lambda} \ket{\psi} = \Phi (G + \lambda \id)^{-1} \Phi^\dagger \ket{\psi}= \Phi \vec{\beta} = \sum_{i=1}^k \beta_i \ket{\phi_i}
\end{equation}
We will also use the following equivalent forms of the $\ridgeproj_{\lambda}$, which are commonly used in machine learning as part of ridge regression and the so-called kernel trick~\cite{bishop2006pattern}. 
\begin{claim}\label{claim:kernel_trick}
Let $\lambda > 0$, $N = 2^n$ and $K = \Phi \Phi^\dagger \in \mathbb{C}^{N \times N}$ be the kernel matrix. We then have
\begin{enumerate}[$(i)$]
    \item $\ridgeproj_\lambda = K(K+\lambda \id_N)^{-1}$,
    \item $\id - \ridgeproj_\lambda = \lambda(K+\lambda \id_N)^{-1}$,
\end{enumerate}
as long as the inverses exist.
\end{claim}
\begin{proof}
We firstly note that for $\lambda > 0$, both $G + \lambda \id_k$ and $K + \lambda \id_N$ are positive definite and hence invertible. This can be quickly observed: let $x \in \mathbb{C}^k \setminus \{0^k\}$, then $x^\star (G+\lambda \id_k) x = \norm{y}_2^2 + \lambda \norm{x}_2^2 > 0$ where $y = \Phi x$. Similarly, for $x \in \mathbb{C}^N \setminus \{0^N\}$, $x^\star (K + \lambda \id_N)x = \norm{y}_2^2 + \lambda \norm{x}_2^2 > 0$ where $y = \Phi^\dagger x$.

\noindent $(i)$ 
Instantiating the push-through identity in Fact~\ref{fact:pushthrough} for $A = \Phi$ and $B = \Phi^\dagger$, we have that
$$
\Phi(\Phi^\dagger \Phi + \lambda \id_k)^{-1} = (\Phi \Phi^\dagger + \lambda \id_N)^{-1} \Phi \implies \Phi(\Phi^\dagger \Phi + \lambda \id_k)^{-1} \Phi^\dagger = (\Phi \Phi^\dagger + \lambda \id_N)^{-1} \Phi \Phi^\dagger = (K + \lambda \id_N)^{-1} K,
$$
where we have multiplied by $\Phi^\dagger$ from the right on both sides in the implication. The result follows from noting first that $\Phi^\dagger\Phi=G$ and next observing that $K$ commutes with $K+\lambda \id)$, hence also commutes with $(K+\lambda \id)^{-1}$.

\noindent $(ii)$ 
Invoking the Woodburry identity in Fact~\ref{fact:woodburry} with $A = \lambda \id_N$, $U = \Phi$, $D = \id_k$,  $V = \Phi^\dagger$, we~get
\begin{align*}
&( \Phi \Phi^\dagger+\lambda \id_N )^{-1} = \frac{\id_N}{\lambda} - \frac{\id_N}{\lambda^2}\Phi\Big(\id_k + \frac{\Phi^\dagger \Phi}{\lambda}\Big)^{-1} \Phi^\dagger=\frac{\id_N}{\lambda} - \frac{\id_N}{\lambda}\Phi\Big( {\Phi^\dagger \Phi}+\lambda\id_k \Big)^{-1} \Phi^\dagger\\
\implies &\lambda (K + \lambda \id_N)^{-1} = \id_N - \ridgeproj_{\lambda}, 
\end{align*}
where we multiplied by $\lambda$ on both sides in the implication and used the definitions of $G,K,\ridgeproj_\lambda$. This gives us the desired result.
\end{proof}

We will often also use that the ridge projector is a contraction.
\begin{fact}\label{fact:ridge_proj_contraction}
For $\lambda > 0$, the ridge projector $\ridgeproj_\lambda$ satisfies
$$
0 \preceq \ridgeproj_\lambda \preceq \id, \quad {\ridgeproj_\lambda}^2 \preceq \ridgeproj_\lambda.
$$
\end{fact}
\begin{proof}
Consider the SVD of $\Phi = U \Sigma V^\dagger$ (for $U \in \mathbb{C}^{N \times N}, \Sigma \in \mathbb{R}^{N \times k}, V \in \mathbb{C}^{k \times k}$). The ridge projector is then
$$
\ridgeproj_\lambda = \Phi (G + \lambda \id)^{-1} \Phi^\dagger = U \Sigma V^\dagger V (\Sigma^2 + \lambda \id_k)^{-1} V^\dagger V \Sigma U^\dagger = U \Sigma (\Sigma^2 + \lambda \id_k)^{-1} U^\dagger,
$$
where we used $(G + \lambda \id)^{-1} = V (\Sigma^2 + \lambda \id_k)^{-1} V^\dagger$ in the second equality since $G = \Phi^\dagger \Phi = (V \Sigma U^\dagger)(U \Sigma V^\dagger) = V \Sigma^2 V^\dagger$ and $\id = V V^\dagger$. Hence, the non-zero singular values of $\ridgeproj_\lambda$ are of the form $\sigma_i^2/(\sigma_i^2 + \lambda), \forall i \in [k]$ where $\sigma_i$ are the singular values of $\Phi$. As every singular value of $\ridgeproj_\lambda$ lies in $[0,1]$, we immediately obtain that
$$
0 \preceq \ridgeproj_\lambda \leq \id_N.
$$
Consequently, we also have that ${\ridgeproj_\lambda}^2 \preceq \ridgeproj_\lambda$. This completes the proof.
\end{proof}

\subsubsection{Implementation of projections}
Next, our goal is to discuss the implementation of these projections using known  block-encoding constructions for implementing the projector $\ridgeproj$. Throughout the results here, by gate complexity of a circuit, we will mean the number of arbitrary one-qubit and two qubit $\CNOT$ gates used. 
  Consider a list $L = \{\ket{\phi_i}\}_{i \in [k]}$ of $k$ many $n$-qubit states in $\calC$. Let $\Phi \in \mathbb{C}^{2^n \times k}$ be the \emph{dictionary} matrix whose $j$th column is $\ket{\phi_j}$. We now discuss how to implement block-encodings of the ridge projectors $\ridgeproj_{L,\lambda}$ (which is defined with respect to the list $L$).  We then have the following claim regarding preparation of states in $L$.
\begin{claim}\label{claim:circ_UPhi}
An $n + \ceil{\log k}$-qubit circuit implementing the following map 
$$
U_\Phi : \ket{j} \ket{0^n} \rightarrow \ket{j} \ket{\phi_j},
$$
can be classically produced in $\widetilde{O}(k\cdot (T_{\prep}  + G_{\prep}) )$ time, and has gate complexity $\widetilde{O}(k\cdot G_{\prep})$.
\end{claim}
\begin{proof}
First we recall that one can use $\calA_{\prep}$ (as defined in the model class definition) to determine a circuit $C_j$ that implements a unitary $V_j$ that prepares $\ket{\phi_j}$, i.e., $V_j \ket{0^n} = \ket{\phi_j}$ for all $j \in [k]$. Each $C_j$ has gate complexity $G_{\prep}$. We can then implement the circuit
\begin{align}
\label{eq:defnofUphi}
U_\Phi = \sum_{j=1}^k \ketbra{j}{j}_Q \otimes C_j,
\end{align}
which is multiplexed controlled circuit where $Q$ is a $\ceil{\log k}$-qubit index register. It can be immediately checked that $U_\Phi\in \mathbb{C}^{2^{n+\log k}\times 2^{n+\log k}}$ has the following action
$$
U_\Phi \ket{j} \ket{0^n} \rightarrow \ket{j} V_j \ket{0^n} = \ket{j} \ket{\phi_j}.
$$
The total gate complexity of $U_\Phi$ is $O(G_{\prep} k \log k)$ as controlling every gate in each $C_j$ for all $j \in [k]$ gives an overhead of $O(\log k)$ per gate. The total time complexity consumed is $O(T_{\prep} k + G_{\prep} k \log k)$ where the $O(T_{\prep} k)$ contribution is due to determining all $C_j$, and $O(G_{\prep} k \log k)$ is the time consumed in going through each gate and determining the appropriate controlled version in terms of arbitrary one-qubit gates and two-qubit $\CNOT$ gates. 
\end{proof}

\paragraph{Block-encodings of Gram and kernel matrices.} We will also require block-encodings of the Gram matrix $G$ and kernel matrix $K$. In particular, we have the following result.
\begin{lemma}\label{lem:G_K_BE}
Let $\lambda > 0$, $\varepsilon \in (0,1)$ and $\mu \in [0,1]$. Suppose $G$ satisfies $\lambda_\min(G) \geq \mu$. Then, one can determine circuits implementing     
\begin{enumerate}[$(i)$]
    \item $(k,n+1,0)$ block-encoding of $G$ acting on $n+\ceil{\log k}$ qubits  with gate complexity $\widetilde{O}(G_{\prep}k)$.
    \item $(k,\ceil{\log k}+1,0)$ block-encoding of $K$ acting on $n+\ceil{\log k}$ qubits  with gate complexity $\widetilde{O}(G_{\prep}k)$.
    \item $(2/(\mu + \lambda),n+4,\varepsilon)$ block-encoding of $(G+\lambda \id_k)^{-1}$ acting on $O(n+ \ceil{\log k})$ qubits with gate complexity 
    $$
    \widetilde{O}\left( \frac{k(k+\lambda)G_{\prep} }{\mu + \lambda} \log 1/\varepsilon\right).
    $$
\end{enumerate}
 The classical time to find these encodings is the same as the gate complexity replacing $G_{\prep}$ by~$T_{\prep}$.
\end{lemma}
\begin{proof}
Let $U_{\Phi}$ be the map defined in Eq.~\eqref{eq:defnofUphi} and $m = \ceil{\log k}$. 

$(i)$ We first give a block-encoding for $G$. We define two maps $W_R\in \mathbb{C}^{2^n\times 2^m}$ and $W_L\in \mathbb{C}^{2^m\times 2^n}$ as follows
\begin{equation}\label{eq:WR_WL}
W_R := (\bra{+}^m \otimes \id_n)U_\Phi(\id_{m} \otimes \ket{0}^n), \quad W_L := (\id_m \otimes \bra{0}^n)U_\Phi^\dagger(\ket{+}^m \otimes \id_{n}).
\end{equation}
The action of $W_R$ on a computational basis state $\ket{j}$ on the index register is
\begin{equation}\label{eq:action_WR}
W_R \ket{j} = (\bra{+}^m \otimes \id_n)U_\Phi(\ket{j} \otimes \ket{0}^n) = (\bra{+}^m \otimes \id_n) (\ket{j} \otimes \ket{\phi_j}) = \frac{\ket{\phi_j}}{\sqrt{k}} \implies W_R = \frac{1}{\sqrt{k}} \Phi.   \end{equation}
Similarly, the action of $W_L$ is
\begin{align}
\bra{j} W_L = (\bra{j} \otimes \bra{0}^n)U_\Phi^\dagger(\ket{+}^m \otimes \id_n) 
= (\bra{j} \otimes \bra{\phi_j})(\ket{+}^m \otimes \id_n) = \frac{\bra{\phi_j}}{\sqrt{k}} \implies W_L = \frac{1}{\sqrt{k}} \Phi^\dagger,
\label{eq:action_WL}
\end{align}
where we  used $\bra{j} \bra{0^n} U_\Phi^\dagger = \bra{j} \bra{\phi_j}$ from the definition of $U_\Phi$. 
Now, let us define the following $n+m$-qubit unitary
\begin{equation}
    W := U_\Phi^\dagger (H^{\otimes m} R_0 H^{\otimes m} \otimes \id_n) U_\Phi,
\end{equation}
where $R_0 = 2\ket{0^m}\bra{0^m} - \id$. 
It can then be shown that
\begin{align*}
    (\id_m \otimes \bra{0}^n) W (\id_m \otimes \ket{0}^n) &=(\id_m \otimes \bra{0}^n) U_\Phi^\dagger \Big(\big(2(\ketbra{+}{+})^{\otimes m}-\id_m\big) \otimes \id_n\Big) U_\Phi (\id_m \otimes \ket{0}^n)\\
    &= 2 W_L W_R - \id\\
    &= \frac{2}{k} G - \id,
\end{align*}
where we used  Eqs~\eqref{eq:action_WR},\eqref{eq:action_WL} in the second equality. Thus, $W$ is a $(1,n,0)$-block-encoding of $2G/k - \id$. The corresponding gate complexity of $W$ is $O(G_{\prep} k \log k + m^2)$ where the first contribution is due to applications of $U_\Phi,U_\Phi^\dagger$ and the second contribution is the gate complexity of $R_+$.
Using Lemma~\ref{lem:sum_prod_BEs}$(i)$, we can combine $W$ with the trivial block-encoding of $\id$ as $G = k/2 \cdot (W+\id)$ to get a $(k,n+1,0)$-block-encoding of $G$, which we will denote by $U_G$. The gate complexity of $U_G$ will be $O(G_{\prep} k \log k + m^2)$. This completes the proof of $(i)$.

$(ii)$ Similarly, for the block-encoding of the kernel matrix $K$, we will define the following $n+m$-qubit unitary matrix
\begin{equation}
    V := (H^{\otimes m} \otimes \id_n) U_\Phi (\id_m \otimes R_0) U_\Phi^\dagger  (H^{\otimes m} \otimes \id_n),
\end{equation}
where $R_0 =2 \ket{0^n}\bra{0^n} - \id$. It can then be shown that
\begin{align*}
    (\bra{0^m} \otimes \id_n) V (\ket{0^m} \otimes \id_n) &= (\bra{+}^m \otimes \id_n) U_\Phi \Big(\id_m \otimes (2 \ket{0^n}\bra{0^n} - \id_n) \Big) U_\Phi^\dagger  (\ket{+}^m \otimes \id_n) \\
    &= 2 W_R W_L - \id \\
    &= \frac{2}{k} K - \id,
\end{align*}
where we used Eqs.~\eqref{eq:action_WR},\eqref{eq:action_WL} in the second equality. Thus, $V$ is a $(1,m,0)$-block-encoding of $2K/k - \id$. The corresponding gate complexity of $V$ is $O(G_{\prep} k \log k + m)$ where the first contribution is due to applications of $U_\Phi,U_\Phi^\dagger$ and the second contribution is the gate complexity of $R_0$. Using Lemma~\ref{lem:sum_prod_BEs}$(i)$, we can combine $V$ with the trivial block-encoding of $\id$ as $K=k/2\cdot (V+\id)$ to get a $(k,m+1,0)$-block-encoding of $K$, which we will denote by $U_K$. The gate complexity of $U_K$ will be $O(G_{\prep} k \log k + m)$. This completes the proof of $(ii)$.

$(iii)$ Let $A = G + \lambda \id$. To implement $A$ for $\lambda > 0$, we use Lemma~\ref{lem:sum_prod_BEs} to combine $G$ with $\id$ to obtain an $(k+\lambda,n+2,0)$ block-encoding of $A$, which we will denote by $U_A$. If $\lambda=0$, we can proceed with the block-encoding $U_G$. The gate complexity of $U_A$ will be $T_A = O(G_{\prep} k \log k + m)$.

To implement a block-encoding of $A^{-1}$, we will use Lemma~\ref{lem:inverse_BE} as $A$ is Hermitian (since $G$ and $\id$ are Hermitian). Using Weyl's inequality in Lemma~\ref{fact:weyl}, observe that spectrum of $A$ lies in $[\mu + \lambda, k + \lambda]$\footnote{Note that the maximum eigenvalue of $G$ can be upper bounded by its trace: $\lambda_\max(G) \leq \Tr(G) = k$, the bound follows from Weyl's inequality.} as $\lambda_\min(G) \geq \mu$. The spectrum of $A/(k+\lambda)$ then lies in $[(\mu + \lambda)/(k + \lambda), 1]$. Considering the desired error parameter $\varepsilon \in (0,1)$, Lemma~\ref{lem:inverse_BE} then gives us a $(2/(\mu + \lambda), n+4, \varepsilon)$-block-encoding of $A^{-1}$, which we denote by $U_{A^{-1}}$, with gate complexity
\begin{equation}\label{eq:gate_complexity_Ainv}
T_{A^{-1}} = O\left( \frac{(k+\lambda)(G_{\prep} k \log k + m)}{\mu + \lambda} \log \frac{2}{\varepsilon(\mu + \lambda)}\right).    
\end{equation}
In other words,
\begin{equation}\label{eq:Ainv_BE}
\norm{A^{-1} - (2/(\mu + \lambda)) (\bra{0}^{n+3} \otimes \id_m) U_{A^{-1}} (\ket{0}^{n+3} \otimes \id_m)} \leq \varepsilon.
\end{equation}
 The classical time to find these encodings is the same as the gate complexity replacing $G_{\prep}$ by~$T_{\prep}$.
 \end{proof}

\paragraph{Block-encoding of ridge projector.} Using the block encoding constructions of of $G,K$ from the lemma above, we now give block-encodings for $\ridgeproj_{L,\lambda}$ and $\id - \ridgeproj_{L,\lambda}$ corresponding to the list $L$ of states $\{\ket{\phi_i}\}_{i \in [k]}$ from model class $\calC$. Note that the result below is restricted to  $\lambda > 0$. 
\begin{lemma}\label{lem:ridgeproj_BE}
Let $\lambda > 0$ and $\varepsilon \in (0,1)$. Then, one can determine circuits implementing 
\begin{enumerate}[$(i)$]
\item $(2k/\lambda, 2\ceil{\log k} + 5, \varepsilon)$-block-encoding of $\ridgeproj_{L,\lambda}$ acting on $n+3$ qubits  with gate complexity $\widetilde{O}\left( (G_{\prep} k^2 )/\lambda\cdot  \log(1/\varepsilon)\right)$.
\item $(2,\ceil{\log k} + 4,\varepsilon)$-block-encoding of $(\id - \ridgeproj_{L,\lambda})$ acting on $n+\ceil{\log k} + 3$ qubits with gate complexity $\widetilde{O}\left( (G_{\prep} k^2)/\lambda\cdot \log(1/\varepsilon)\right)$.
\end{enumerate}
 The classical time to find these encodings is the same as the gate complexity replacing $G_{\prep}$ by~$T_{\prep}$.
 \end{lemma}
\begin{proof}
Let $m = \ceil{\log k}$. We will implement block-encodings of $\ridgeproj_{L,\lambda}$ and $\id - \ridgeproj_{L,\lambda}$ by utilizing Claim~\ref{claim:kernel_trick} and the $(k,m + 1, 0)$-block-encoding $U_K$ of $K$ from Lemma~\ref{lem:G_K_BE}$(ii)$. 

$(i)$ From Claim~\ref{claim:kernel_trick}, we have that $\ridgeproj_{L,\lambda} = K(K+\lambda \id_N)^{-1}$. Let $B = K + \lambda \id_N$. Using Lemma~\ref{lem:sum_prod_BEs}$(i)$, we can combine $U_K$ with $\id$ to get a $(k + \lambda,m+2,0)$-block-encoding of $B$, which we will denote by $U_B$. The gate complexity of $U_B$ will be $O(G_{\prep} k \log k)$. To implement $B^{-1}$, we now use Lemma~\ref{lem:inverse_BE}. We firstly note that $B$ is a Hermitian matrix as $K$ and $\id$ are Hermitian. The spectrum of $B/(k + \lambda)$ lies in $[\lambda/(k + \lambda),1]$ as the spectrum of $K$ lies in $[0, k + \lambda]$. Considering the error parameter $\varepsilon_1 \in (0,1)$ to be fixed later, Lemma~\ref{lem:inverse_BE} then gives us a $(2/\lambda, m + 4, \varepsilon_1)$-block-encoding of $B^{-1}$, which we denote by $U_{B^{-1}}$ with gate complexity
\begin{equation}
T_{B^{-1}} = O\left( (G_{\prep} k^2 \log k)/\lambda + m) \log(1/(\lambda \varepsilon_1))\right).    
\end{equation}
In other words,
\begin{equation}\label{eq:Binv_BE}
\norm{B^{-1} - (2/\lambda) (\bra{0}^{m+3} \otimes \id_n) U_{B^{-1}} (\ket{0}^{m+3} \otimes \id_n)} \leq \varepsilon_1.
\end{equation}
 The classical time to find these encodings is the same as the gate complexity replacing $G_{\prep}$ by~$T_{\prep}$.
 
We now use Lemma~\ref{lem:sum_prod_BEs}$(ii)$ to obtain a $(2k/\lambda, 2m+5, k\varepsilon_1)$-block-encoding of $\ridgeproj_{L,\lambda} = K B^{-1}$, which we denote by $U_{\ridgeproj_{L,\lambda}}$. Setting $\varepsilon_1 = \varepsilon/k$ gives us a $(2k/\lambda, 2m+5, \varepsilon)$-block-encoding $U_{\ridgeproj_{L,\lambda}}$ of $\ridgeproj_{L,\lambda}$ with gate complexity 
$$
T_{\ridgeproj_{L,\lambda}} = O\left( (G_{\prep} k^2 \log k)/\lambda \log(k /(\lambda \varepsilon))\right).
$$
This completes the proof of item $(i)$.

$(ii)$ From Claim~\ref{claim:kernel_trick}, we have that $(\id - \ridgeproj_{L,\lambda}) = \lambda(K+\lambda \id_N)^{-1} = \lambda B^{-1}$. The unitary $U_{B^{-1}}$ is thus directly a $(2,m+4,\lambda \varepsilon_1)$-block-encoding of $(\id - \ridgeproj_{L,\lambda})$. This can be directly observed by multiplying Eq.~\eqref{eq:Binv_BE} by $\lambda$ on both sides. Setting $\varepsilon_1 = \varepsilon/\lambda$ gives us a $(2,m+4,\varepsilon)$-block-encoding $U_{B^{-1}}$ of $(\id - \ridgeproj_{L,\lambda}) := \lambda B^{-1}$ with gate complexity
\begin{equation}\label{eq:gate_complexity_Binv}
T_{\id - \ridgeproj_{L,\lambda}} = O\left( (G_{\prep} k^2 \log k)/\lambda \log(1 /\varepsilon)\right).    
\end{equation}
This completes the proof of item $(ii)$ and the overall lemma.
\end{proof}

\subsubsection{Coefficient state preparation}
As part of our learning algorithm, we will also require the ability to prepare the quantum \emph{coefficient} state encoding the the \emph{coefficient} vector $\vec{\beta} = (G + \lambda \id)^{-1} \Phi^\dagger \ket{\psi}$ (Eq.~\eqref{eq:soln_ridge_regression}) that is the solution to the ridge regression problem~(Eq.~\eqref{eq:ridge_regression}). To this end, we denote this state by $\ket{\beta}$ and will take the form
\begin{equation}
\label{eq:ridge_coeff_state}
    \ket{\beta} := \frac{\vec{\beta}}{\norm{\vec{\beta}}_2} = \frac{(G + \lambda \id)^{-1} \Phi^\dagger \ket{\psi}}{\norm{(G + \lambda \id)^{-1} \Phi^\dagger \ket{\psi}}_2}.
\end{equation}
We now show that copies of $\ket{\beta}$ can be prepared efficiently using the block-encodings that we have discussed so far, for all values of the parameter $\lambda \geq 0$. We also show that the norm of the coefficient vector $\norm{\vec{\beta}}$ can be estimated efficiently.
\begin{claim}
\label{claim:prep_coeff_state}
Let $\lambda > 0$, $\varepsilon, \upsilon, \delta \in (0,1)$, and $\mu \in [0,1]$. Let the minimum eigenvalue of $G$ satisfy $\lambda_\min(G) \geq \mu$. Define  $\ket{\beta}$ as in Eq.~\eqref{eq:ridge_coeff_state} and suppose ${\norm{(G + \lambda \id)^{-1} \Phi^\dagger \ket{\psi}}_2} \geq \gamma$. Then, there is 
\begin{enumerate}[$(i)$]
    \item a circuit acting on $n+ \ceil{\log k} + 3$ qubits which prepares the state $\ket{\widetilde{\beta}}$ s.t. $\norm{\ket{\widetilde{\beta}} - \ket{\beta}} \leq \varepsilon$, with success probability 
    $$p_{\mathrm{succ}} \geq \frac{\gamma^2 (\mu + \lambda)^2}{4k}.$$ 
    The circuit has gate complexity
    $$
    \widetilde{O}\left( \frac{k(k+\lambda)G_{\prep}}{\mu + \lambda} \log \frac{1}{\varepsilon\gamma}\right).
    $$
    \item an algorithm that outputs an estimate $\widehat{b}$ of the norm $\norm{\beta}_2^2$, with probability $\geq 1-\delta$, such that $\Big|\widehat{b} - \norm{\beta}_2^2\Big| \leq \upsilon$ using $O(k^2/(\varepsilon^2 (\mu + \lambda)^4 \log(1/\delta)))$ copies of $\ket{\psi}$ and has gate complexity
    $$
    \widetilde{O}\left( \frac{k^3(k+\lambda)G_{\prep} }{\varepsilon^2 (\mu + \lambda)^5} \log \frac{k}{\varepsilon(\mu + \lambda)} \right).
    $$
\end{enumerate}
 The classical time to find these encodings is the same as the gate complexity replacing $G_{\prep}$ by~$T_{\prep}$.
\end{claim}
\begin{proof}
Let $m=\ceil{\log k}$, $A = (G + \lambda \id)$, and $\varepsilon' \in (0,1)$ be an error parameter to be fixed later. We consider a system of $n$ qubits $\calS$ along with a system of $a:=n+4$ ancillary qubits $\calA_1$ and another system of $m$ ancillary qubits $\calA_2$. To prepare $\ket{\beta}$, we will require the following. Let $U_\Phi$ be the unitary defined in Eq.~\eqref{eq:defnofUphi} as part of Claim~\ref{claim:circ_UPhi}, and $U_{A^{-1}}$ be the $O(2/(\mu + \lambda), n+4, \varepsilon')$-block-encoding of $A^{-1}$ from Lemma~\ref{lem:G_K_BE}$(iii)$. 

$(i)$ We now describe the circuit to prepare $\ket{\beta}$ approximately. We consider the initial state of $\ket{0^a}_{\calA_1}\ket{0^m}_{\calA_2}\ket{\psi}_\calS$, and then apply the unitary $(\id_a \otimes H^{\otimes m} \otimes \id_n)$ followed by $(\id_a \otimes U_{\Phi}^\dagger)$ and then $(U_{A^{-1}} \otimes \id_n)$. We then measure $\calA_1$ and $\calS$ in the computational basis, and post-select for them being $\ket{0^a}$ and $\ket{0^n}$. The corresponding circuit is shown below in Figure~\ref{fig:prepare_coeff_state}.
\begin{figure}[h!]
\centering
$$
\Qcircuit @C=1.5em @R=1.2em {
\lstick{\ket{0^a}} & \qw & \qw & \multigate{1}{U_{A^{-1}}} & \meter & \rstick{\ket{0^a}} \\
\lstick{\ket{0^m}} & \gate{H^{\otimes m}} & \multigate{1}{U_\Phi^\dagger} & \ghost{U_{A^{-1}}} & \qw & \rstick{\ket{\widetilde{\beta}}}\\
\lstick{\ket{\psi}} & \qw & \ghost{U_\Phi^\dagger} & \qw & \meter & \rstick{\ket{0^n}}
}
$$
\caption{Circuit for preparing the state $\ket{\widetilde{\beta}}$.}
\label{fig:prepare_coeff_state}
\end{figure}

We now argue that the resulting post-selected state on $\calA_2$, which we denote by $\ket{\widetilde{\beta}}$, would then be approximately $\ket{\beta}$. Let us define $E =  \alpha(\bra{0}^{a} \otimes \id_m) U_{A^{-1}} (\ket{0}^{a} \otimes \id_m) - A^{-1}$ where $\alpha = 2/(\mu + \lambda)$ and $E$ corresponds to the error in the implementation of $A^{-1}$ by $U_{A^{-1}}$ which satisfies $\norm{E} \leq \varepsilon'$. We first show that 
\begin{equation}\label{eq:interim_def_W}
W := (\bra{0^a} \otimes \id_m \otimes \bra{0^n})\cdot (U_{A^{-1}} \otimes \id_n)(\id_a \otimes U_{\Phi}^\dagger)(\id_a \otimes H^{\otimes m} \otimes \id_n) \cdot (\ket{0^a,0^m} \otimes \id_n)    
\end{equation}
implements $A^{-1} \Phi^\dagger$ up to some subnormalization. Let this normalization be $\nu > 0$ to be determined and let us evaluate $A^{-1} \Phi^\dagger - \nu W$:
\begin{align*}
&A^{-1} \Phi^\dagger - \nu (\bra{0^a} \otimes \id_m \otimes \bra{0^n})(U_{A^{-1}} \otimes \id_n)(\id_a \otimes U_{\Phi}^\dagger)(\id_a \otimes H^{\otimes m} \otimes \id_n)(\ket{0^a}\ket{0^m} \otimes \id_n) \\
=&A^{-1} \Phi^\dagger - \nu(\id_m \otimes \bra{0^n}) \Big( \Big[ (\bra{0^a} \otimes \id_m)U_{A^{-1}}(\ket{0^a} \otimes \id_m) \Big] \otimes \id_n \Big) U_{\Phi}^\dagger(\ket{+}^m \otimes \id_n) \\
=&A^{-1} \Phi^\dagger - \nu(\id_m \otimes \bra{0^n}) \Big( \frac{1}{\alpha}(A^{-1}+ E) \otimes \id_n \Big) U_{\Phi}^\dagger(\ket{+}^m \otimes \id_n) \\
=&A^{-1} \Phi^\dagger - \frac{\nu}{\alpha}(A^{-1} + E) (\id_m \otimes \bra{0^n}) U_{\Phi}^\dagger(\ket{+}^m \otimes \id_n) \\
=&A^{-1} \Phi^\dagger - \frac{\nu}{\alpha}(A^{-1} + E)\frac{\Phi^\dagger}{\sqrt{k}},
\end{align*}
where we have used that $U_{A^{-1}}$ is a $(\alpha,a,\varepsilon')$-block-encoding of $A^{-1}$ in the third line and used Eq.~\eqref{eq:action_WL} in the final line. Setting $\nu = \alpha \sqrt{k}$ gives us
$$
\norm{A^{-1} \Phi^\dagger - \nu W} = \norm{\frac{\nu E \Phi^\dagger}{\alpha \sqrt{k}}} \leq \norm{E}\cdot \norm{\Phi^\dagger} \leq \sqrt{k} \varepsilon',
$$
where we used that $\norm{\Phi^\dagger} = \sqrt{\norm{\Phi^\dagger \Phi}} = \sqrt{\norm{G}} \leq \sqrt{k}$ and that $\norm{E} \leq \varepsilon'$. This implies that the action on an input state $\ket{\psi}$ on the system register $\calS$ satisfies
\begin{equation}\label{eq:interim_action_W_is_good}
\norm{\frac{1}{\nu} A^{-1}\Phi^\dagger \ket{\psi} - W \ket{\psi}} \leq \frac{\sqrt{k} \varepsilon'}{\nu} = \frac{\varepsilon'}{\alpha}.    
\end{equation}
The action of the circuit in Figure~\ref{fig:prepare_coeff_state} on the input state $\ket{0^a}_{\calA_1}\ket{0^m}_{\calA_2}\ket{\psi}_\calS$ then produces the state 
\begin{equation}\label{eq:interim_state_on_circ}
\ket{0^a} \Big(W \ket{\psi}\Big) \ket{0^n} + \ket{\perp},    
\end{equation}
which is $\varepsilon'/\alpha$-close to
$$
\ket{0^a}\Big(\frac{1}{\nu} A^{-1}\Phi^\dagger \ket{\psi} \Big) \ket{0^n} + \ket{\perp},
$$
and where $\ket{\perp}$ is an unnormalized state that is orthogonal to every state with $\ket{0^a}$ on the register $\calA_1$ and $\ket{0^n}$ on the register $\calS$. Thus, if we measure the registers $\calA_1$ and $\calS$ in the states $\ket{0^a}$ and $\ket{0^n}$ respectively, we get the state 
\begin{equation}\label{eq:interim_def_approx_beta}
\ket{\widetilde{\beta}} := \frac{W \ket{\psi}}{\norm{W \ket{\psi}}}.    
\end{equation}
We now argue that $\ket{\widetilde{\beta}}$ is close to the desired state $\ket{\beta}$. For this, let us define the vectors $\widetilde{y} = W \ket{\psi}$ and $y = A^{-1} \Phi^\dagger \ket{\psi}/\nu$. We then have that
\begin{align*}
    \norm{\ket{\widetilde{\beta}} - \ket{\beta}}_2 = \norm{\frac{\widetilde{y}}{\norm{\widetilde{y}}} - \frac{y}{\norm{y}}}
    \leq \norm{\frac{\widetilde{y}}{\norm{\widetilde{y}}} - \frac{\widetilde{y}}{\norm{y}}} + \norm{\frac{\widetilde{y}}{\norm{y}} - \frac{y}{\norm{y}}}
    &\leq \norm{\widetilde{y} \frac{\norm{y} - \norm{\widetilde{y}}}{\norm{\widetilde{y}} \norm{y}}} + \norm{\frac{\widetilde{y} - y}{\norm{y}}} \\
     &\leq \Big| \frac{\norm{y} - \norm{\widetilde{y}}}{\norm{y}}\Big| + \norm{\frac{\widetilde{y} - y}{\norm{y}}} \\
    &\leq 2\frac{\norm{\widetilde{y} - y}}{\norm{y}} \\
    &\leq 2\frac{\nu \varepsilon'}{\alpha \gamma},
\end{align*}
where we have added and subtracted $\widetilde{y}/\norm{y}$ before applying the triangle inequality in the second inequality in first line. We used submultiplicativity of norms in the first term in the second line. In the third line, we used reverse triangle inequality and in the fourth line we used Eq.~\eqref{eq:interim_action_W_is_good} and the fact that $\norm{y} = \nu^{-1} \norm{A^{-1} \Phi^\dagger \ket{\psi}} \geq \nu^{-1} \gamma$ since $\norm{A^{-1} \Phi^\dagger \ket{\psi}} \geq \gamma$ (by assumption of the claim).

Setting $\varepsilon' = \varepsilon \alpha \gamma/(2 \nu) = \varepsilon \gamma/(2 \sqrt{k})$ gives us that $\ket{\widetilde{\beta}}$ is $\varepsilon$-close to $\ket{\beta}$ i.e., $\norm{\ket{\widetilde{\beta}} - \ket{\beta}} \leq \varepsilon$. Towards obtaining the success probability of producing $\ket{\widetilde{\beta}}$ which is given by $p_{\mathrm{succ}} = \norm{W \ket{\psi}}_2^2$, we note that
\begin{align*}
\norm{W \ket{\psi}} \geq \norm{\frac{1}{\nu} A^{-1} \Phi^\dagger \ket{\psi}} -\norm{W \ket{\psi} - \frac{1}{\nu}A^{-1} \Phi^\dagger \ket{\psi}}
\geq \frac{\gamma}{\nu} - \frac{\varepsilon'}{\alpha}
= \frac{\gamma}{\nu} - \frac{\varepsilon \gamma}{2\nu}
\geq \frac{\gamma}{2\nu} = \frac{\gamma(\mu + \lambda)}{2\sqrt{k}},
\end{align*}
where we used the reverse triangle inequality in the first inequality, the promise $\norm{A^{-1} \Phi^\dagger \ket{\psi}}_2 \geq \gamma$ in the second inequality along with Eq.~\eqref{eq:interim_action_W_is_good}, and substituted $\nu = \alpha \sqrt{k}$ along with $\alpha = 2/(\mu + \lambda)$ in the final equality. The success probability is then 
$$
p_{\mathrm{succ}} = \norm{W \ket{\psi}}_2^2 \geq \frac{\gamma^2 (\mu + \lambda)^2}{4k}.
$$
The gate complexity of the circuit is due to implementation of $H^{\otimes m}$, $U_{A^{-1}}$ (with $\varepsilon'$ as the relevant error) and $U_\Phi^\dagger$ which gives a total complexity of
$$
O\left( \frac{(k+\lambda)(G_{\prep} k \log k + m)}{\mu + \lambda} \log \frac{\sqrt{k}}{\varepsilon \gamma (\mu + \lambda)}\right),
$$
where we used Lemma~\ref{lem:G_K_BE}$(iii)$. This completes the proof of item $(i)$.

$(ii)$ We note that the circuit before measurement produces the state in Eq.~\eqref{eq:interim_state_on_circ}. The probability of measuring $\ket{0}^a$ on $\calA_1$ and $\ket{0^n}$ on $\calS$ after measurement, which we earlier called the success probability, is $p_{\mathrm{succ}} = \norm{W \ket{\psi}}_2^2$. We can then estimate $p_{\mathrm{succ}}$ from these measurement outcomes. Consider the error parameter $\xi \in (0,1)$ and let the estimate be $\widehat{p}$ such that $|\widehat{p} - p_{\mathrm{succ}}| \leq \xi$. 
This can be obtained from $O(1/ \xi^2\cdot \log(1/\delta))$ measurement outcomes and hence applications of the circuit. Set $\widehat{b} = \nu^2 \widehat{p}$.
Now observe that
\begin{align*}
\Big| \widehat{b} - \norm{A^{-1}\Phi^\dagger \ket{\psi}}_2^2 \Big|
&\leq \Big| \nu^2 \widehat{p} - \nu^2 p_{\mathrm{succ}}\Big|+\Big| \nu^2 p_{\mathrm{succ}} - \norm{A^{-1}\Phi^\dagger \ket{\psi}}_2^2 \Big|\\
&\leq \nu^2 \xi +\Big| \nu \norm{W \ket{\psi}}_2^2 - \norm{A^{-1}\Phi^\dagger \ket{\psi}}_2^2 \Big|\\
&= \nu^2 \xi + \nu^2 \norm{ W \ket{\psi} - \frac{1}{\nu} A^{-1}\Phi^\dagger \ket{\psi}} \cdot \norm{ W \ket{\psi} + \frac{1}{\nu} A^{-1}\Phi^\dagger \ket{\psi}}  \\
&\leq \nu^2 \xi + \nu^2 \frac{\varepsilon'}{\alpha} \cdot \Big( \norm{W \ket{\psi}} + \frac{1}{\nu} \norm{A^{-1}\Phi^\dagger \ket{\psi}} \Big)  \\
&\leq \nu^2 \xi + \nu^2 \frac{\varepsilon'}{\alpha} \cdot \Big(1 + \frac{\sqrt{k}}{\nu(\lambda + \mu)} \Big) \\
&= \nu^2 \xi + \nu^2 \frac{\varepsilon'}{\alpha} \cdot (1 + 2) \\
&= \nu^2 \xi + \nu^2 \frac{3 \varepsilon'}{\alpha}
\end{align*}
where the first inequality used triangle inequality and the fourth line used Eq.~\eqref{eq:interim_action_W_is_good} where $\varepsilon'$ is the parameter as defined in $(i)$ corresponding to the error of the block-encoding $U_{A^{-1}}$ will be set differently here. In the fifth line, we used the triangle inequality again and then used the definition of $\nu = \alpha \sqrt{k}$ and $\alpha = 2/(\lambda + \mu)$ in the sixth line. We set $\xi = \varepsilon/(2\nu^2) = \varepsilon(\mu + \lambda)^2/(8k)$ where we used $\nu = \alpha \sqrt{k}$, and $\varepsilon' = \alpha \varepsilon/(6 \nu^2) = \varepsilon(\mu + \lambda)/(12 k)$. This then gives us that $\widehat{b}$ is the a $\varepsilon$-approximation of $\norm{\beta}_2^2$. The main contribution to gate complexity is $O(1/ \xi^2\cdot \log(1/\delta))$ applications of $U_W$. Using the gate complexity of $U_W$ from item $(i)$ for the $\varepsilon'$ defined here, we obtain a total gate complexity~of
$$
\widetilde{O}\left( \frac{k^3(k+\lambda)G_{\prep} }{\varepsilon^2 (\mu + \lambda)^5} \log \frac{k}{\varepsilon(\mu + \lambda)} \right).
$$
This completes the proof. Finally, like we mentioned in the previous proofs,  the classical time to find these encodings is the same as the gate complexity replacing $G_{\prep}$ by~$T_{\prep}$.
\end{proof}

\subsubsection{Residual state preparation} 
Like we mentioned in the introduction, after the learning algorithm has learned the states in $L$, the goal is to construct the so-called residual state, which is proportional to $(\id - \ridgeproj_{L,\lambda})\ket{\psi}$, so that it can find the next element $\ket{\phi_i}$ to include in $L$. To this end, we use the above block-encoding of $\id - \ridgeproj_{L,\lambda}$ to prepare copies of the state proportional to $(\id - \ridgeproj_{L,\lambda})\ket{\psi}$. This is stated formally below. Note that the following result is again restricted to positive $\lambda$ i.e., $\lambda > 0$. 

\begin{claim}
\label{claim:prep_ridgeproj_residual_state}
Let $\lambda > 0$ and $\varepsilon \in (0,1)$. Define $$\ket{\psi_L} := \frac{(\id - \ridgeproj_{L,\lambda}) \ket{\psi}}{\norm{(\id - \ridgeproj_{L,\lambda}) \ket{\psi}}}.$$ 
Then, there is a circuit acting on $n+ \ceil{\log k} + 4$ qubits which
\begin{enumerate}[$(i)$]
    \item prepares the state $\ket{\psi_L}$ up to error $\varepsilon$, with success probability $\geq \gamma^2/16$ if ${\norm{(\id - \ridgeproj_{L,\lambda}) \ket{\psi}}} \geq \gamma$, 
    and has gate complexity 
    $$
    \widetilde{O}\left( (G_{\prep} k^2 )/\lambda \cdot \log(1 /(\gamma \varepsilon))\right).
    $$
    \item can be used to output an estimate $\widehat{b}$ of the norm $\norm{(\id - \ridgeproj_{L,\lambda}) \ket{\psi}}_2^2$ up to error $\varepsilon$, with probability $\geq 1-\delta$ using $O(1/\varepsilon^2 \log(1/\delta)))$ copies of $\ket{\psi}$ and has gate complexity
    $$
    \widetilde{O}\left( (G_{\prep} k^2 )/\lambda \cdot \log(1 /\varepsilon)\right).
    $$
\end{enumerate}
 The classical time to find these encodings is the same as the gate complexity replacing $G_{\prep}$ by~$T_{\prep}$.
\end{claim}
\begin{proof}
Let $\varepsilon' \in (0,1)$ be an error parameter that we will fix later and $m = \ceil{\log k}$. From Lemma~\ref{lem:ridgeproj_BE}$(ii)$, we have a $(2, m+4, \varepsilon')$-block-encoding $U_{\id - \ridgeproj_\lambda}$ of $(\id - \ridgeproj_\lambda)$. Hence,
\begin{align}\label{eq:interim2:action_good}
    &\norm{(\id - \ridgeproj_\lambda) - 2(\bra{0}^{m+3} \otimes \id_n)U_{\id - \ridgeproj_\lambda}(\ket{0}^{m+3} \otimes \id_n)} \leq \varepsilon' \\
    \implies & \norm{(\id - \ridgeproj_\lambda)\ket{\psi}/2 - (\bra{0}^{m+3} \otimes \id_n)U_{\id - \ridgeproj_\lambda}(\ket{0}^{m+3} \otimes \id_n)\ket{\psi}} \leq \varepsilon'/2.
\end{align}
Let $W:=(\bra{0}^{m+3} \otimes \id_n)U_{\id - \ridgeproj_\lambda}(\ket{0}^{m+3} \otimes \id_n)$. We have then observed that by applying $U_{\id - \ridgeproj_\lambda}$ to $\ket{\psi}$, we obtain the state 
$$
\ket{0}^{m+3}\Big(W \ket{\psi} \Big) + \ket{\perp},
$$
where $\ket{\perp}$ is some unnormalized state that is orthogonal to every state with $\ket{0}^{m+3}$ in the first register and that is $\varepsilon'/2$-close to a state of the form
$$
\ket{0}^{m+3}\Big(\frac{1}{2} (\id - \ridgeproj_\lambda)\ket{\psi} \Big) + \ket{\perp}.
$$
Thus, if we measure the registers $\calA_1$ and $\calS$ in the states $\ket{0^a}$ and $\ket{0^n}$ respectively, we get the state 
\begin{equation}\label{eq:interim_def_approx_res}
\ket{\widetilde{\psi}_L} := \frac{W \ket{\psi}}{\norm{W \ket{\psi}}}.    
\end{equation}
We now argue that $\ket{\widetilde{\psi}_L}$ is close to the desired state $\ket{\psi}_L$. For this, let us define the vectors $\widetilde{y} = W \ket{\psi}$ and $y =\frac{1}{2} (\id - \ridgeproj_\lambda)\ket{\psi}$. We then have that (as in Claim~\ref{claim:prep_coeff_state})
\begin{align*}
\norm{\ket{\widetilde{\psi}_L} - \ket{\psi}_L}_2 &\leq 2\frac{\norm{\widetilde{y} - y}}{\norm{y}} \leq \frac{2\varepsilon'}{\gamma},
\end{align*}
where we used Eq.~\eqref{eq:interim2:action_good} and the fact that $\norm{y} = 1/2 \norm{\id - \ridgeproj_\lambda)\ket{\psi}} \geq \gamma/2$ by assumption of the~claim. Setting $\varepsilon' = \varepsilon \gamma/2$ gives us that $\ket{\widetilde{\psi}_L}$ is $\varepsilon$-close to $\ket{\psi_L}$ i.e., $\norm{\ket{\widetilde{\psi}_L} - \ket{\psi}_L} \leq \varepsilon$. Towards obtaining the success probability of producing $\ket{\widetilde{\psi}_L}$ which is given by $p_{\mathrm{succ}} = \norm{W \ket{\psi}}_2^2$, we note that
\begin{align*}
\norm{W \ket{\psi}} \geq \frac{1}{2}\norm{\id - \ridgeproj_\lambda)\ket{\psi}} -\norm{W \ket{\psi} - \frac{1}{2}\norm{\id - \ridgeproj_\lambda)\ket{\psi}}}
\geq \frac{\gamma}{2} - \frac{\varepsilon'}{2}
= \frac{\gamma}{2} - \frac{\varepsilon \gamma}{4}
\geq \frac{\gamma}{4}
\end{align*}
where we used the reverse triangle inequality in the first inequality and the promise $1/2 \norm{\id - \ridgeproj_\lambda)\ket{\psi}} \geq \gamma/2$ in the second inequality along with Eq.~\eqref{eq:interim2:action_good}. The success probability is then 
$$
p_{\mathrm{succ}} = \norm{W \ket{\psi}}_2^2 \geq \gamma^2/16.
$$
Item $(i)$ follows since the gate complexity of implementing $U_{\id - \ridgeproj_\lambda}$ from Lemma~\ref{lem:ridgeproj_BE}$(ii)$ with  error~$\varepsilon'$~is 
$$
O\left( (G_{\prep} k^2 \log k)/\lambda \log(1 /(\gamma \varepsilon))\right).    
$$

$(ii)$ We note in the circuit corresponding to $W$, the probability of measuring $\ket{0}^{m+3}$ on the first register after measurement, which we earlier called the success probability, is $p_{\mathrm{succ}} = \norm{W \ket{\psi}}_2^2$. We can then estimate $p_{\mathrm{succ}}$ from these measurement outcomes. Consider the error parameter $\xi \in (0,1)$ and let the estimate be $\widehat{p}$ such that $|\widehat{p} - p_{\mathrm{succ}}| \leq \xi$. This can be obtained from $O(1/ \xi^2\cdot \log(1/\delta))$ measurement outcomes and hence applications of the circuit. Set $\widehat{b} = 4 \widehat{p}$.
Now observe that
\begin{align*}
\Big| \widehat{b} - \norm{\id - \ridgeproj_\lambda)\ket{\psi}}_2^2 \Big|
&\leq \Big| 4 \widehat{p} - 4 p_{\mathrm{succ}}\Big|+\Big| \nu^2 p_{\mathrm{succ}} - \norm{\id - \ridgeproj_\lambda)\ket{\psi}}_2^2 \Big|\\
&\leq 4 \xi + 4 \frac{\varepsilon'}{2} \cdot \Big( \norm{W \ket{\psi}} + \frac{1}{2} \norm{\id - \ridgeproj_\lambda)\ket{\psi}} \Big)  \\
&\leq 4 \xi + 4 \frac{\varepsilon'}{2} \cdot (1 + 1/2) \\
&= 4 \xi + 6\varepsilon'
\end{align*}
where we proceeded in a similar fashion as to we had done in the proof of Claim~\ref{claim:prep_coeff_state}. Note that $\varepsilon'$ is the parameter as defined in $(i)$ corresponding to the error of the block-encoding $U_{A^{-1}}$ will be set differently here. We set $\xi = \varepsilon/8$ and $\varepsilon' = \varepsilon/12$. This then gives us that $\widehat{b}$ is the a $\varepsilon$-approximation of $\norm{\id - \ridgeproj_\lambda)\ket{\psi}}_2^2$. The main contribution to gate complexity is $O(1/\xi^2\cdot \log(1/\delta))$ applications of $U_W$. Using the gate complexity of $U_W$ from item $(i)$ for the $\varepsilon'$ defined here, we obtain a total gate complexity~of
$$
O\left( (G_{\prep} k^2 \log k)/\lambda \log(1 /\varepsilon)\right).
$$
This completes the proof. Finally, like we mentioned in the previous proofs,  the classical time to find these encodings is the same as the gate complexity replacing $G_{\prep}$ by~$T_{\prep}$.
\end{proof}

\subsection{Approach}
We have so far given protocols for preparing a residual state proportional to $(\id - \ridgeproj_{L,\lambda})\ket{\psi}$ given a list $L$ of states from $\calC$. We have not yet described how this list is learned and in particular pertaining to Theorem~\ref{thm:learning_decompositions}. In the following two sections (Section~\ref{subsec:structure_learning} and Section~\ref{subsec:parameter_learning}), we describe the algorithm of Theorem~\ref{thm:learning_decompositions} which is divided into two stages: $(i)$ \emph{structure learning} where we learn a list $L$ of states from $\calC$ that describe the relevant $\calC$-structure in $\ket{\psi}$, and $(ii)$ \emph{parameter learning} where we learn the coefficients corresponding to states in $L$ such that the guarantee of Theorem~\ref{thm:learning_decompositions} if fulfilled. We refer the reader to Section~\ref{sec:tech_overview} for a high-level intuition behind this approach. 
As we discussed in the introduction, the structure learning algorithm will run over multiple iterations $t \in \{0,1,\ldots,\kappa\}$. Accordingly, we will require the following notation as part of the analysis and to denote different objects as part of the algorithm.
\paragraph{Residual state.} Suppose we are in iteration $t$ and we have learned so far $t-1$ many states from $\calC$, call it $L = \{\ket{\phi_i}\}_{i \in [t-1]}$. We will then denote the \emph{true residual vector} as
\begin{equation}\label{eq:residual_vector}
\Psi_{t} = (\id - \ridgeproj_{t-1,\lambda}) \ket{\psi},    
\end{equation} 
and the corresponding (normalized) state or the \emph{true residual state} as
\begin{equation}\label{eq:residual_state}
    \ket{\psi_{t}} = \Psi_{t}/\alpha_{t},
\end{equation}
where  $\alpha_{t} = \norm{\Psi_{t}}_2$. However, we are only able to approximately prepare the residual state using Claim~\ref{claim:prep_ridgeproj_residual_state}, which we denote by $\ket{\widetilde{\psi}_t}$. Thus, the $\ket{\widetilde{\psi}_{t}}$ is what we prepare during the course of structure learning and upon which we carry out agnostic learning. Considering the true residual vector and true residual state will, however, be useful for our \emph{analysis}.

\paragraph{Residual loss.} Recall that the matrix $\Phi_t$ is a matrix whose rows are given by the $\ket{\phi_i}\in L$, so its dimension is $2^n\times t$. Define
\begin{equation}\label{eq:residual_loss}
R_{t,\lambda} := \min_{\vec{a} \in \mathbb{C}^t} J_{t}(\vec{a};\lambda) = \min_{\vec{a} \in \mathbb{C}^t} \norm{\ket{\psi} - \Phi_t \vec{a}}_2^2 + \lambda \norm{\vec{a}}_2^2,
\end{equation}
where we have used the definition of $J_t(\vec{a};\lambda)$ from Eq.~\eqref{eq:ridge_regression} and use the subscript $t$ to denote the iteration in which this is being evaluated and emphasizes that the dictionary $\Phi_t$ and dimension of $\vec{a}$ changes with iteration $t$. 

\subsection{Structure learning}\label{subsec:structure_learning}
In this section, we first give the algorithm for structure learning and then state our main theorem.

\begin{myalgorithm}
\begin{algorithm}[H]\label{algo:structure_learning}
    \caption{Structure learning}
    \setlength{\baselineskip}{1.8em} 
    \DontPrintSemicolon 
    \KwInput{$\varepsilon_s \in (0,1)$, copies of  $\ket{\psi}$, weak learner $\calA_{\wal}$ (Def.~\ref{def:model_class}) with promise $\eta$.}
    \KwOutput{List of states in $\calC$: $L = \{\ket{\phi_i}\}_{i \in [k]}$}
    \vspace{2mm}
    Initialize residual state $\ket{\psi_1} = \ket{\psi}$, coefficient $\alpha_1 = 1$, list $L = \emptyset$, and array $\Phi_0 = [\,]$. \\
    Set $\eta = \eta(\varepsilon_s)$ with $\eta(\cdot)$ being the promise of $\calA_{\wal}$. \\
    Set parameters $\lambda=1/2$, $t_{\max}=\left \lceil 2/(\varepsilon_s \eta) \right\rceil$, $\delta'=\delta/(2t_{\textsf{max}})$, $\kappa=0$. \\
    \For{$t=1$ \KwTo $t_{\max}$}
    {
        Run $\calA_{\wal}$ on $S_{\wal}$ copies of $\ket{\widetilde{\psi}_t}$ to learn $\ket{\phi_t} \in \calC$. \\
        Run $\SWAP$ test on $O(1/\eta^2 \log(1/\delta'))$ copies of $\ket{\widetilde{\psi}_t},\ket{\phi_t}$ to obtain an $\eta/2$-approximate estimate $\nu_t$ of $|\la \psi_t | \phi_t \ra|^2$. \\ 
        \lIf {$\nu_t < \eta$}{let $\ket{\phi_R} = \ket{\psi_t}$ and break from loop. \label{algo_step:est_fidelity}} 
        Update list $L \leftarrow L \cup \{ \ket{\phi_t} \}$, matrix $\Phi_t = [\Phi_{t-1}, \ket{\phi_t}]$ and counter $\kappa \leftarrow \kappa + 1$. \\
        Set Gram matrix $G_t = \Phi_t^\dagger \Phi_t \in \mathbb{C}^{\kappa \times \kappa}$ and projector $\ridgeproj_{t,\lambda} = \Phi_t (G+\lambda \id)^{-1} \Phi_t^{\dagger}$. \\
        Let $\widehat{\alpha}_{t+1}^2$ be an $\varepsilon_s/2$-estimate of $\norm{(\id-\ridgeproj_{t,\lambda} )\ket{\psi}}_2^2$ obtained via \ref{claim:prep_ridgeproj_residual_state}$(ii)$. \label{algo_step:est_alphat}\\
        \lIf { $|\widehat{\alpha}_{t+1}|^2 < \varepsilon_s$ } { let $\ket{\phi_R} = \ket{\psi_t}$ and break from loop. \label{algo_step:step_est_alphat}} 
        Prepare copies of residual state $\ket{\widetilde{\psi}_{t+1}}$ using Claim~\ref{claim:prep_ridgeproj_residual_state}$(i)$.
        }
    \textbf{Return} List of $\kappa$ many states in $\calC$: {$L=\{\ket{\phi_i}\}_{i \in [\kappa]}$}, held classically.
\end{algorithm}
\end{myalgorithm}
In the remaining of the section, we will prove the correctness and soundness of the algorithm above, stated as the theorem below.
\begin{theorem}\label{thm:structure_learning}
Let $\varepsilon_s,\delta\in (0,1)$ and $\calC$ be a model class as in  Definition~\ref{def:model_class}. Suppose $\ket{\psi}$ is an unknown $n$-qubit quantum state. Then, there is an algorithm that with probability $\geq 1-\delta$, determines a list of $\kappa \leq 8/(\varepsilon_s \eta(\varepsilon_s))$ states $L := \{\ket{\phi_i}\}_{i \in [\kappa]}$ from $\calC$ such that $\ket{\psi}$ can be written~as
\begin{align}
\label{eq:guaranteeofstructurelearning}
\ket{\psi} = \ridgeproj_{L,\lambda} \ket{\psi} + \alpha_{\kappa + 1} \ket{\psi_{\kappa+1}}
\end{align}
and the residual state $\ket{\psi_{\kappa+1}}$ satisfies\footnote{We remark that throughout this work $\eta$ will be a bijective function (often even a polynomial function of the argument), so $\eta^{-1}$ is well-defined.} 
$$
|\alpha_{\kappa + 1}|^2 \cdot \calF_\calS(\ket{\psi_{\kappa + 1}}) \leq \max\{(3/2) \varepsilon_s, \, \eta^{-1}(3 \eta(\varepsilon_s)/2)\}.
$$ 
Moreover, $\norm{\ridgeproj_{L,\lambda} \ket{\psi}} \geq \kappa \varepsilon_s \eta(\varepsilon_s)/8$. The algorithm has the following complexity for
\begin{align*}
&\text{Sample complexity: } \widetilde{O}\Big(S_{\wal} \cdot \poly( 1/\varepsilon_s,1/\eta(\varepsilon_s),\log 1/\delta) \Big) \\
&\text{Time complexity: } \widetilde{O}\Big(T_{\wal} \cdot \poly( 1/\varepsilon_s,1/\eta(\varepsilon_s),\log 1/\delta)) +  \poly(G_{\prep},T_{\prep}, 1/\varepsilon_s,1/\eta(\varepsilon_s)\Big).
\end{align*} 
\end{theorem}

We remark that a consequence of the above theorem's proof is the following existential decomposition theorem of quantum states against any class of states $\C$. Algorithmizing this existential result requires us to impose some additional constraints on the class $\C$ as done in Definition~\ref{def:model_class}. This existential result should be interpreted as a weak regularity lemma for quantum states.
\begin{theorem}\label{thm:quantum_weak_regularity}
Let $\varepsilon \in (0,1), \lambda \geq 0$ and $\calC$ be any class of state that does not necessarily satisfy  Definition~\ref{def:model_class}. Suppose $\ket{\psi}$ is an arbitrary $n$-qubit pure quantum state. Then, there is a list of $\kappa = O(1/\varepsilon^2)$ states $L := \{\ket{\phi_i}\}_{i \in [\kappa]}$ from $\calC$ such that $\ket{\psi}$ can be written~as
$$
\ket{\psi} = \ridgeproj_{L,\lambda} \ket{\psi} + \alpha_{\kappa + 1} \ket{\psi_{\kappa+1}}, \quad \text{ where } |\alpha_{\kappa + 1}|^2 \cdot \calF_\calS(\ket{\psi_{\kappa + 1}}) \leq \varepsilon.
$$
\end{theorem}

\subsubsection{Correctness}
We firstly upper bound the maximum number of iterations $\kappa$, the structure learning algorithm (Algorithm~\ref{algo:structure_learning}) runs for. To this end, we need to comment on the promise that the state $\ket{\phi_t} \in \calC$ learned in each iteration satisfy with respect to the true residual state $\ket{\psi_t}$. Recall that as part of Algorithm~\ref{algo:structure_learning}, we only prepare the approximate residual state $\ket{\widetilde{\psi}_t}$ and carry out agnostic learning on that state. 
\begin{claim}\label{claim:promise_WAL}
Consider Algorithm~\ref{algo:structure_learning} and let $\delta' \in (0,1)$. Suppose we have accomplished agnostic learning in iteration $t$ on the prepared residual state $\ket{\widetilde{\psi}_t}$. If $\norm{\ket{\psi_t} - \ket{\widetilde{\psi}_t}} \leq \gamma_t$ and $\gamma_t = \eta(\varepsilon_s)/8$, then with probability $\geq 1- \delta'$
\begin{enumerate}[$(i)$]
    \item in step~(6), the estimate $\Big| \nu_t - |\la \phi_t | \psi_t \ra|^2 \Big| \leq \eta(\varepsilon_s)/2$,
    \item $|\la \phi_t | \psi_t \ra|^2 \geq \eta(\varepsilon_s)/2$ when $|\la \phi_t | \widetilde{\psi}_t \ra|^2 \geq \eta(\varepsilon_s)$.
\end{enumerate}
\end{claim}
\begin{proof}
Let $\ket{\widetilde{\psi}_t} - \ket{\psi_t} = \ket{e_t}$ where $\ket{e_t}$ is not necessarily a normalized quantum state and corresponds to the error in preparing the residual state. We now note that
\begin{equation}\label{eq:implication_good_res_state_prep}
\Big| \la \phi_t | \widetilde{\psi}_t \ra - \la \phi_t | \psi_t \ra \Big| = |\la \phi_t | e_t \ra| \leq \norm{e_t}_2 \leq \gamma_t
\end{equation}
where we have used the promise of $\norm{\ket{\psi_t} - \ket{\widetilde{\psi}_t}} \leq \gamma_t$. We then have that the estimate $\nu_t$ of $|\la \phi_t | \widetilde{\psi}_t \ra|^2$ also serves as a good estimate of $|\la \phi_t | \psi_t \ra|^2$:
\begin{align*}
\Big| \nu_t - |\la \phi_t |\psi_t \ra|^2 \Big| 
&\leq \Big| \nu_t - |\la \phi_t | \widetilde{\psi}_t \ra|^2 \Big| + \Big| |\la \phi_t | \widetilde{\psi}_t \ra|^2 - |\la \phi_t | \widetilde{\psi}_t \ra|^2 \Big| \\
&= \Big| \nu_t - |\la \phi_t | \widetilde{\psi}_t \ra|^2 \Big| + \Big| |\la \phi_t | \widetilde{\psi}_t \ra| - |\la \phi_t | \psi_t \ra| \Big| \cdot \Big| |\la \phi_t | \widetilde{\psi}_t \ra| + |\la \phi_t | \psi_t \ra| \Big| \\
&\leq \Big| \nu_t - |\la \phi_t | \widetilde{\psi}_t \ra|^2 \Big| + 2 \gamma_t,
\end{align*}
where we used $|\la \phi_t | \psi_t \ra|, |\la \phi_t | \widetilde{\psi}_t\ra| \leq 1$ and Eq.~\eqref{eq:implication_good_res_state_prep} as $\Big| |\la \phi_t | \widetilde{\psi}_t \ra| - |\la \phi_t | \psi_t \ra| \Big| \cdot \Big| \leq \Big| \la \phi_t | \widetilde{\psi}_t \ra - \la \phi_t | \psi_t \ra \Big| \leq \gamma_t$. We now note that $\Big|\nu_t - |\la \phi_t | \widetilde{\psi}_t \ra|^2 \Big| \leq \eta(\varepsilon_s)/4$ with probability at least $1-\delta'$ if we performed the $\SWAP$ test on $O(1/\eta(\varepsilon_s)^2 \log(1/\delta'))$ copies of $\ket{\phi_t}$ and $\ket{\widetilde{\psi}_t}$. Item $(i)$ then follows from setting $\gamma_t = \eta(\varepsilon_s)/8$:
$$
\Big| \nu_t - |\la \phi_t |\psi_t \ra|^2 \Big| \leq \Big| \nu_t - |\la \phi_t | \widetilde{\psi}_t \ra|^2 \Big| + 2 \gamma_t \leq \eta(\varepsilon_s)/4 + \eta(\varepsilon_s)/4 \leq \eta(\varepsilon_s)/2.
$$
To prove item $(ii)$, we note that it follows from Eq.~\eqref{eq:implication_good_res_state_prep} that
\begin{align*}
|\la \phi_t | \psi_t \ra| \geq |\la \phi_t | \widetilde{\psi}_t \ra| - \gamma_t \geq (7/8) \sqrt{\eta(\varepsilon_s)}
\implies |\la \phi_t | \widetilde{\psi}_t \ra|^2 \geq (49/64) \eta(\varepsilon_s) \geq \eta(\varepsilon_s)/2,    
\end{align*}
where we have used the promise $|\la \phi_t | \widetilde{\psi}_t \ra|^2 \geq \eta(\varepsilon_s)$ and that $\gamma_t = \eta(\varepsilon_s)/8 \leq \sqrt{\eta(\varepsilon_s)}/8$ as $\eta(\varepsilon_s) \in [0,1]$. This completes the proof.
\end{proof}
The above claim thus shows that if we prepare the residual state $\ket{\widetilde{\psi}_t}$ up to low error and then learn $\ket{\phi_t}$ after performing agnostic learning on $\ket{\widetilde{\psi}_t}$, $\ket{\phi_t}$ will also have high fidelity with the true residual state $\ket{\psi_t}$. We will now comment on the trend of the minimum residual loss $R_{t,\lambda}$ (Eq.~\eqref{eq:residual_loss}) across consecutive iterations before we stop. We first give results considering that the ridge projector $\Pi_{t,\lambda}$ is used in iteration $t$ and then specialize to the case when $\lambda=0$ i.e., the orthogonal projector is used.
\begin{claim}\label{claim:greedy_ridge_loss}
Let $\lambda \in [0,1]$ and suppose we are in the $t$'th iteration of Algorithm~\ref{algo:structure_learning}. Let $\Psi_t = (\id - \ridgeproj_{t-1,\lambda}) \ket{\psi}$ be the current true residual state and $\ket{\psi_t} = \Psi_t/\norm{\Psi_t}$. The following function is convex in $a_t \in \mathbb{C}$\footnote{We remark that the ridge regression $J_t$ was defined over $\mathbb{R}^k$ and we define $\widetilde{J}$ as the \emph{greedy univariate} version defined on $\mathbb{R}$.}
$$
\widetilde{J}(a_t;\lambda) := \norm{\Psi_t - a_t \ket{\phi_t}}_2^2 + \lambda |a_t|^2,
$$
with the unique minimizer of $\argmin_{a_t} \widetilde{J}(a_t;\lambda) = a_t^\star = \norm{\Psi_t} \la \phi_t | \psi_t \ra/(1 + \lambda)$.
\end{claim}
\begin{proof}
Let $\alpha_t = \norm{\Psi_t}_2$. Noting that for $u,v \in \mathbb{C}^{2^n}$, $\|u - v\|_2^2=\|u\|_2^2+\|v\|_2^2 - 2\textsf{Re}(\langle u,v\rangle)$, we can expand the relevant function as
\begin{align*}
\widetilde{J}(a_t;\lambda) = \norm{\Psi_t - a_t \ket{\phi_t}}_2^2 + \lambda|a_t|^2
= \alpha_t^2 + |a_t|^2 - 2 \alpha_t \textsf{Re}(\overline{a_t} \langle \phi_t| \psi_t \rangle) + \lambda |a_t|^2 \,,
\end{align*}
where we used $\overline{a_t}$ to denote the complex conjugate of $a_t$. Consider $a_t = x + iy$ where $x,y \in \mathbb{R}$. The function above can then be rewritten as
$$
\widetilde{J}(x,y;\lambda) = \alpha_t^2 + (1+\lambda)(x^2 + y^2) - 2\alpha_t \Big(x \textsf{Re}(\langle \phi_t| \psi_t \rangle) + y \textsf{Im}(\langle \phi_t| \psi_t \rangle) \Big).
$$
It can then be shown that the gradient of $\widetilde{J}(x,y;\lambda)$ is
\begin{align*}
    \frac{\partial \widetilde{J}}{\partial x} = 2(1+\lambda)x - 2\alpha_t \textsf{Re}(\langle \phi_t| \psi_t \rangle), \enspace \frac{\partial \widetilde{J}}{\partial y} = 2(1+\lambda)y - 2\alpha_t \textsf{Im}(\langle \phi_t| \psi_t \rangle), 
\end{align*} 
and Hessian of $\widetilde{J}(x,y;\lambda)$
$$
\nabla^2 \widetilde{J} = \begin{pmatrix}
2(1+\lambda) & 0\\
0 & 2(1+\lambda)
\end{pmatrix}=2(1+\lambda) \id,
$$
which is positive definite for $\lambda \geq 0$. Hence, $\widetilde{J}(x,y;\lambda)$ is strictly convex and has a unique minimum. To compute the solution at the minimum, which we denote by $a_t^\star = x^\star + i y^\star$, one can compute this by setting the derivatives to zero.  We determine the solution at the minimum $a_t^\star = x^\star + iy^\star = \alpha_t \la \phi_t|\psi_t\ra/(1 + \lambda)$. This gives us the desired result.
\end{proof}

\begin{lemma}\label{lem:promise_iteration_ridgeproj}
Let $\lambda \in [0,1]$ and $\delta' \in (0,1)$ be the failure probability of any iteration in Algorithm~\ref{algo:structure_learning}. Suppose we stop after $\kappa$ many iterations. Steps $(6)-(12)$ in the algorithm ensure that for each $t \leq \kappa$, with probability $\geq 1 - \delta'$ the residual loss drops from iteration $t$ to $t+1$ as
$$
R_{t,\lambda} - R_{t+1,\lambda} \geq \frac{\varepsilon_s \eta(\varepsilon_s)}{4(1+\lambda)}.
$$
\end{lemma}
\begin{proof}
Consider iteration $t \geq 1$. Suppose at the start of the iteration, we have the following decomposition:
\begin{equation}\label{eq:decomp_start_t}
\ket{\psi} = \ridgeproj_{t-1,\lambda} \ket{\psi} + \Psi_t,
\end{equation}
where the residual vector $\Psi_t = (\id - \ridgeproj_{t-1,\lambda})\ket{\psi}$ can be written as $\Psi_t = \alpha_t \ket{\psi_t}$ with $\alpha_t = \norm{\Psi_t}_2$. Moreover, the current running estimate is $\widehat{\Psi}_{t-1} = \ridgeproj_{t-1,\lambda}\ket{\psi} = \Phi_{t-1} \vec{\beta}_{t-1}$ where $\vec{\beta}_{t-1} \in \mathbb{C}^{t-1}$ is the solution to ridge regression (Eq.~\eqref{eq:ridge_regression}).\footnote{In iteration $t=1$, by definition, $\widehat{\Psi}_{t-1}=0$ and $\Psi_{t+1} = \ket{\psi}$.} 
Steps $(6)$ ensures that with probability $\geq 1-\delta'/2$, we have that $\nu_t$ is an $\eta(\varepsilon_s)/2$-approximation to $|\la \phi_t| \psi_t \ra|^2$ and Step $(11)$  ensures that $\widehat{\alpha}_t$ is an $\varepsilon_s/2$-approximation to $\alpha_t := \norm{\Psi_t}_2$ respectively. A union bound ensures with probability $\geq 1-\delta'$,
$$
\Big| |\widehat{\alpha}_{t}| - |\alpha_t| \Big| \leq \varepsilon_s/2, \enspace \text{and} \enspace \Big| \nu_{t} - |\la \phi_t | \psi_t \ra|^2 \Big| \leq \eta(\varepsilon_s)/2.
$$
Now, observe that if we have not exited from the loop i.e., $\nu_t \geq \eta(\varepsilon_s)$ and $\widehat{\alpha}_t^2 \geq \varepsilon_s$, then the true values satisfy
\begin{equation}
\label{eq:promise_true_valuesmain}
|\la \phi_t | \psi_t \ra|^2 \geq \eta(\varepsilon_s)/2, \enspace \text{and} \enspace \alpha_t^2 \geq \varepsilon_s/2.    
\end{equation}
The weak agnostic learner $\calA_{\wal}$ has then learned a state $\ket{\phi_t} \in \calC$ such that
\begin{equation}\label{eq:promise_WAL_iter_t}
|\la \phi_{t} | \psi_{t} \ra|^2 \geq \eta(\varepsilon_s)/2.    
\end{equation}
After learning $\phi_t$, we can update the dictionary matrix $\Phi_t = [\Phi_{t-1}, \ket{\phi_t}]$.

At this point, the improved running estimate would be $\widehat{\Psi}_t = \ridgeproj_{t,\lambda} \ket{\psi} = \Phi_t \vec{\beta}_t$ where $\vec{\beta}^{(t)} \in \mathbb{C}^t$ is the solution to ridge regression (Eq.~\eqref{eq:ridge_regression}) i.e., $\vec{\beta}^{(t)} = \argmin_{\vec{a} \in \mathbb{C}^t} J_t(\vec{a};\lambda)$. However, the global optimal minimum value $R_{t+1,\lambda} = J(\vec{\beta}_t;\lambda)$ (Eq.~\eqref{eq:residual_loss}) is at most the value you get by holding the previous minimum solution $\vec{\beta}^{(t-1)} \in \mathbb{C}^{t-1}$ and optimizing over $a_t$ (since we have reduced the minimization space from $\vec{a}\in \mathbb{C}^t$ to $\{(\vec{\beta}^{(t-1)},a_t):a_t\in \mathbb{C}\}$), i.e.,
\begin{equation}
\label{eq:global_optimum_vs_greedy_optimum}
R_{t+1,\lambda} = \min_{\vec{a} \in \mathbb{C}^t} J_t(\vec{a};\lambda) \leq \min_{a_t} J\big((\vec{\beta}^{(t-1)}, a_t);\lambda\big).
\end{equation}
We now note that 
\begin{equation}\label{eq:1d_ridge_regression}
J\big((\vec{\beta}^{(t-1)}, a_t);\lambda\big) = \norm{\Psi_{t-1} - a_t \ket{\phi_t}}_2^2 + \lambda \norm{\vec{\beta}^{(t-1)}}_2^2 + \lambda |a_t|^2.    
\end{equation}
As the second term does not depend on $a_t$, the minimum of $J(\vec{\beta}^{(t-1)}, a_t;\lambda)$ occurs at the minimum of the following function 
\begin{equation}\label{eq:equiv_1d_ridge_regression}
g(a_t) := \norm{\Psi_{t-1} - a_t \ket{\phi_t}}_2^2 + \lambda |a_t|^2.    
\end{equation}
From Claim~\ref{claim:greedy_ridge_loss}, we know that the minimum of $g(a_t)$ occurs at $a_t^\star = \alpha_t \la \phi_t|\psi_t\ra/(1+\lambda)$. Noting the minimum value of $J\big((\vec{\beta}^{(t-1)}, a_t);\lambda\big)$ occurs at $a_t^\star$ as well, we then have 
\begin{align}
\min_{a_t} J(\vec{\beta}^{(t-1)}, a_t;\lambda) &= \norm{\Psi_{t-1} - a_t^\star \ket{\phi_t}}_2^2 + \lambda \norm{\vec{\beta}_{t-1}}_2^2 + \lambda |a_t^\star|^2 \\    
&= \norm{\Psi_{t-1}}^2 + (1+\lambda)|a_t^\star|^2 - 2 \alpha_t \textsf{Re}(\overline{a_t^\star} \langle \phi_t| \psi_t \rangle) + \lambda \norm{\vec{\beta}_{t-1}}_2^2 \\
&= \norm{\Psi_{t-1}}^2 - \frac{\alpha_t^2 |\langle \phi_t| \psi_t \rangle|^2}{1+\lambda} + \lambda \norm{\vec{\beta}_{(t-1)}}_2^2 \\
&= R_{t,\lambda} - \frac{\alpha_t^2 |\langle \phi_t| \psi_t \rangle|^2}{1+\lambda},
\label{eq:1d_ridge_regression_min_value}
\end{align}
where we have used $\|u + v\|_2^2=\|u\|_2^2+\|v\|_2^2 + 2\textsf{Re}(\langle u,v\rangle)$ for $u,v \in \mathbb{C}^{2^n}$ in the second line, we have substituted the value of $a_t^\star$ to obtain the third line and used the expression for $R_{t,\lambda} := \min_{\vec{a} \in \mathbb{C}^{t-1}} J(\vec{a};\lambda) = J(\vec{\beta}^{(t-1)};\lambda)$ (Eq.~\eqref{eq:residual_loss}) to obtain the final equality. Substituting Eq.~\eqref{eq:1d_ridge_regression_min_value}) into Eq.~\eqref{eq:global_optimum_vs_greedy_optimum} combined with the promises of Eq.~\eqref{eq:promise_true_valuesmain} immediately gives us the following implication
\begin{align}
R_{t,\lambda} - R_{t+1,\lambda} \geq \frac{\alpha_t^2 |\langle \phi_t| \psi_t \rangle|^2}{1+\lambda} \implies R_{t,\lambda} - R_{t+1,\lambda} \geq \frac{\varepsilon_s \eta(\varepsilon_s)}{4(1+\lambda)},
\end{align}
which is the desired result. This completes the proof.
\end{proof}
Using these claims, we now obtain an upper bound on the maximum number of iterations $\kappa$ and  prove the promise on the final residual state $\ket{\psi_{\kappa + 1}}$ in Theorem~\ref{thm:structure_learning} after $\kappa$ many~iterations.
\begin{claim}\label{claim:ub_kappa}
Consider the context of Theorem~\ref{thm:structure_learning}. Let $\varepsilon_s, \delta \in (0,1)$ and $\lambda \in [0,1]$. The maximum number of iterations $\kappa$ in the structure learning algorithm (Algorithm~\ref{algo:structure_learning}) of Theorem~\ref{thm:structure_learning} with probability $\geq 1-\delta$ is bounded as 
$$
\kappa \leq 4(1+\lambda)/(\varepsilon_s \eta(\varepsilon_s)).
$$    
Furthermore,  we have that 
$$
|\alpha_{\kappa + 1}|^2 \cdot \calF_{\calC}(\ket{\psi_{\kappa + 1}}) < \max\{(3/2) \varepsilon_s ,\eta^{-1}(3 \eta(\varepsilon_s)/2)\}
$$
\end{claim}
\begin{proof}
We first observe that $1 \geq R_{t,\lambda} \geq 0$ for any choice of $\lambda \geq 0$ and iteration $t \in \mathbb{N}$. To see the upper bound, we note that $R_{t,\lambda}$ is the minimum value of $J_t(\vec{a};\lambda)$ for $\vec{a} \in \mathbb{C}^t$ (Eq.~\eqref{eq:residual_loss}) and the value obtained by choosing $\vec{a} = 0^t$ gives us an upper bound:
$$
R_{t,\lambda} := \min_{\vec{a} \in \mathbb{C}^t} J_{t}(\vec{a};\lambda) = \min_{\vec{a} \in \mathbb{C}^t} \norm{\ket{\psi} - \Phi_t \vec{a}}_2^2 + \lambda \norm{\vec{a}}_2^2 \leq \norm{\ket{\psi}}_2^2 = 1.
$$
Suppose the algorithm ran for $\kappa$ many iterations before stopping. Then, we have that
$$
1\geq R_{1,\lambda} - R_{\kappa+1,\lambda} = \sum_{t=1}^{\kappa} R_{t,\lambda} - R_{t+1,\lambda} \geq  \frac{\kappa \varepsilon_s \eta(\varepsilon_s)}{4(1+\lambda)} \implies \kappa \leq \frac{{4(1+\lambda)}}{\varepsilon_s \eta(\varepsilon_s)},
$$
where we used that $\ket{\psi_1} = \ket{\psi}$ in the first inequality and Lemma~\ref{lem:promise_iteration_ridgeproj} in the third inequality. This is true with success probability $\geq 1 - \kappa \delta'$ where $\delta'$ is the failure probability of Lemma~\ref{lem:promise_iteration_ridgeproj}. Setting $\delta'=\delta \varepsilon_s \eta(\varepsilon_s)/4$ gives us the desired success probability. This proves the first part of the claim.

We now prove the second part of the claim. Recall from the proof of Lemma~\ref{lem:promise_iteration_ridgeproj} that we obtain an estimate of $\eta(\calF_{\calC}(\ket{\psi_t}))$, denoted  $\nu_t$, up to error $\eta(\varepsilon_s)/2$ and an estimate of $\alpha_t := \norm{(\id - \ridgeproj_{t-1,\lambda})\ket{\psi}}_2$, denoted by $\widehat{\alpha}_t$, up to error $\varepsilon_s/2$. In other words,
$$
\Big| |\widehat{\alpha}_{t}| - |\alpha_t| \Big| \leq \varepsilon_s/2, \enspace \text{and} \enspace \Big| \nu_{t} - \eta(\calF_{\calC}(\ket{\psi_t}))] \Big| \leq \eta(\varepsilon_s)/2.
$$
If we stop after $\kappa$ many iterations, then either $\nu_{\kappa + 1} < \eta(\varepsilon_s)$ or $\widehat{\alpha}_{\kappa + 1}^2 < \varepsilon_s$. The true values then satisfy
\begin{equation}\label{eq:promise_true_values}
\calF_{\calC}(\ket{\psi}_{\kappa + 1}) < \eta^{-1}(3 \eta(\varepsilon_s)/2), \enspace \text{or} \enspace \alpha_{\kappa + 1}^2 < 3\varepsilon_s/2.    
\end{equation}
We then have 
$$
|\alpha_{\kappa + 1}|^2 \cdot \calF_{\calC}(\ket{\psi_{\kappa + 1}}) \leq \max\{(3/2) \varepsilon_s ,\eta^{-1}(3 \eta(\varepsilon_s)/2)\}.
$$
This completes the proof.
\end{proof}

Finally, we also have the following claim that lower bounds the norm of the projected state.
\begin{claim}\label{claim:norm_projected_state}
Consider the context of Theorem~\ref{thm:structure_learning}. Let $\varepsilon_s \in (0,1)$ and $\lambda \in [0,1]$. Suppose the we stop after $\kappa$ many iterations in Algorithm~\ref{algo:structure_learning}. The norm of the projected state $\ridgeproj_{\kappa,\lambda}$ satisfies 
$$
\norm{\ridgeproj_{\kappa,\lambda} \ket{\psi}} \geq |\la \psi | \ridgeproj_{\kappa, \lambda} | \psi \ra| \geq \kappa \frac{\varepsilon_s \eta(\varepsilon_s)}{4(1+\lambda)}.
$$
Additionally, $\ridgeproj_{\kappa,\lambda} \ket{\psi} = \Phi_\kappa \vec{\beta}$ where $\vec{\beta} \in \mathbb{C}^\kappa$ and satisfies $\|\vec{\beta}\|^2 \geq |\la \psi | \ridgeproj_{\kappa, \lambda} | \psi \ra|/(\kappa + \lambda)$.
\end{claim}
\begin{proof}
Starting from the definition of $R_{t,\lambda}$ (Eq.~\eqref{eq:residual_loss} and expanding, we have that
\begin{align}
R_{t,\lambda} = \min_{\vec{a} \in \mathbb{C}^t} \norm{\ket{\psi} - \Phi_t \vec{a}}_2^2 + \lambda \norm{\vec{a}}_2^2 
&= \min_{\vec{a} \in \mathbb{C}^t} \left[ \norm{\ket{\psi}}^2 - 2\mathrm{Re}(\la \psi | \Phi_t \vec{a} \ra) + \vec{a}^\dagger \Phi_t^\dagger \Phi_t \vec{a} + \lambda \vec{a}^\dagger \vec{a} \right] \nonumber \\
&= \min_{\vec{a} \in \mathbb{C}^t} \left[ 1 - 2\mathrm{Re}(\la \psi | \Phi_t \vec{a} \ra) + \vec{a}^\dagger(G_t + \lambda \id) \vec{a} \right],
\end{align}
where we denoted $G_t = \Phi_t^\dagger \Phi_t$. Noting that the minimum occurs at $\vec{a}_t = (G_t + \lambda \id)^{-1} \Phi_t^\dagger \ket{\psi}$ (Eq.~\eqref{eq:soln_ridge_regression}) and substituting into the above equation gives us 
\begin{align}\label{eq:residual_loss_simplified}
R_{t,\lambda} = 1 - 2\mathrm{Re}( \la \psi | \Phi_t (G_t + \lambda \id)^{-1} \Phi_t^\dagger | \psi \ra) + \la \psi| \Phi_t (G_t + \lambda \id)^{-1} \Phi_t^\dagger \ket{\psi} = 1 - \la \psi | \ridgeproj_{t,\lambda} | \psi \ra
\end{align}
where we used that $G_t+\lambda \id$ is Hermitian positive definite and hence $(G+\lambda \id)^{-1}$ is Hermitian in the first equality\footnote{For $\lambda = 0$, it can be shown that $G_t$ is invertible. See Appendix~\label{appsec:}.}. In the second equality, we used the definition of $\ridgeproj_{t,\lambda}$ and that $\la \psi | \ridgeproj_{t,\lambda} | \psi \ra$ is real as $\ridgeproj_{t,\lambda}$ is Hermitian. From Lemma~\ref{lem:promise_iteration_ridgeproj}, we have that
\begin{equation}\label{eq:interim_progress_proj_state_overlap}
\la \psi | \ridgeproj_{t+1,\lambda} | \psi \ra - \la \psi | \ridgeproj_{t,\lambda} | \psi \ra = R_{t,\lambda} - R_{t+1,\lambda} \geq \varepsilon_s \eta(\varepsilon_s)/(4(1+\lambda)).
\end{equation}
We then have that 
\begin{equation}\label{eq:proj_state_overlap_kappa}
\la \psi | \ridgeproj_{\kappa, \lambda} | \psi \ra = \sum_{t=0}^{\kappa-1} \left(\la \psi | \ridgeproj_{t+1,\lambda} | \psi \ra - \la \psi | \ridgeproj_{t,\lambda} | \psi \ra \right) \geq \kappa \varepsilon_s \eta(\varepsilon_s)/(4(1+\lambda))
\end{equation}
Applying Cauchy-Schwarz and using Eq.~\eqref{eq:proj_state_overlap_kappa}, we obtain
$$
\kappa \varepsilon_s \eta(\varepsilon_s)/(4(1+\lambda)) \leq |\la \psi | \ridgeproj_{\kappa, \lambda} | \psi \ra| \leq \norm{\ket{\psi}} \cdot \norm{\ridgeproj_{\kappa, \lambda} | \psi \ra} = \norm{\ridgeproj_{\kappa, \lambda} | \psi \ra},
$$
where we used $\norm{\ket{\psi}} = 1$. This completes the proof of the first part of the claim. To obtain the second part, we first note that $\ridgeproj_{\kappa,\lambda} \ket{\psi} = \Phi_\kappa \vec{a}_\kappa$ and earlier observed (as part of the proof of Eq.~\eqref{eq:residual_loss_simplified}) that
$$
\la \psi | \ridgeproj_{\kappa,\lambda} | \psi \ra = \vec{a}_\kappa^\dagger(G_\kappa + \lambda \id) \vec{a}_\kappa \implies |\la \psi | \ridgeproj_{\kappa,\lambda} | \psi |\ra| \leq (\norm{G_\kappa} + \lambda) \norm{\vec{a}_\kappa}^2 \implies \norm{\vec{a}_\kappa}^2 \geq \frac{\kappa \varepsilon_s \eta(\varepsilon_s)}{4(\kappa + \lambda)(1+\lambda)},
$$
where we used Eq.~\eqref{eq:proj_state_overlap_kappa} in the last inequality along with the fact that $\norm{G_\kappa} \leq \Tr(G_\kappa) \leq \kappa$.\footnote{For small $\lambda$, the above improves upon the slightly weaker result of $\norm{\ridgeproj_{\kappa,\lambda} \ket{\psi}}^2 \leq \norm{G_\kappa} \norm{\vec{a}_\kappa}^2 \implies \norm{\vec{a}_\kappa}^2 \geq |\la \psi | \ridgeproj_{\kappa,\lambda}| \psi \ra|^2/\kappa$.} This completes the proof of the second part and the overall claim.
\end{proof}

\subsubsection{Complexity}

We now analyze the time complexities of different steps in Algorithm~\ref{algo:structure_learning}: $(i)$ preparation of the residual state $\ket{\widetilde{\psi}_t}$, and $(ii)$ estimation of quantities associated with our stopping condition i.e., $\calF_{\calC}(\ket{\psi_t})$ and norm of the residual state $\alpha_t := \norm{(\id - \ridgeproj_{t,\lambda_t}) \ket{\psi}}_2$.

\paragraph{Agnostic learning.}
In each iteration, we need to prepare $S_{\wal}$ many copies of $\ket{\psi_t}$ for agnostic learning. Each preparation of $\ket{\psi_t}$ using sample complexity and time complexity as given by Claim~\ref{claim:prep_ridgeproj_residual_state}. Summing over all $\kappa$ many iterations gives us a sample complexity of $O(S_{\wal} \cdot \poly(1/\varepsilon_s, 1/\eta(\varepsilon_s))$ and a gate complexity of 
$$
O(T_{\wal} \cdot \poly( 1/\varepsilon_s,1/\eta(\varepsilon_s),\log 1/\delta)) +  \poly(G_{\prep},T_{\prep}, 1/\varepsilon_s,1/\eta(\varepsilon_s))
$$ 
where in the latter, there are contributions from the agnostic learning protocol itself as well as the gate complexity corresponding to prepare $S_{\wal}$ many residual states in each iteration.

\paragraph{Estimation of stopping criteria.}
The cost of estimating $\alpha_t$ which is the norm of the residual state up to error $\eta(\varepsilon_s)/2$ uses $O(\kappa/\eta(\varepsilon_s)^2 \log(1/\delta)))$ copies of $\ket{\psi}$ and has gate complexity
$$
O\left( (G_{\prep} \kappa^3)/\lambda \log(1 /\eta(\varepsilon_s))\right).
$$
from Claim~\ref{claim:prep_ridgeproj_residual_state} across all $\kappa$ many iterations.
The cost of obtaining the estimate $\nu_t$ of the fidelity of the learned state $\ket{\phi_t}$ with the residual state $\ket{\psi_t}$ uses $\poly(1/\eta(\varepsilon) \log(1/\delta))$ copies of $\ket{\psi}$ in each iteration using Claim~\ref{claim:prep_ridgeproj_residual_state}. Across all $\kappa$ many iterations, the corresponding sample complexity is $\poly(\kappa/\eta(\varepsilon) \log(1/\delta))$. The overall gate complexity is $\poly(G_\prep, 1/\varepsilon_s, 1/\eta(\varepsilon_s))$.

Putting together all the costs throughout all the iterations, the overall  complexity is as stated and proves Theorem~\ref{thm:structure_learning}. 

\subsection{Parameter learning}
\label{subsec:parameter_learning}

\begin{myalgorithm}
\begin{algorithm}[H]\label{algo:parameter_learning}
\caption{Parameter learning}
\setlength{\baselineskip}{1.6em} 
\DontPrintSemicolon 
\KwInput{Parameter $\lambda > 0$, $\varepsilon_p, \delta \in (0,1)$, copies of  $\ket{\psi}$, list of $\kappa$ many states $L = \{\ket{\phi_i}\}_{i \in [\kappa]}$ from model class $\calC$ with promise of $\norm{(G+\lambda \id)^{-1} \Phi^\dagger \ket{\psi}} \geq \gamma$.}
\vspace{2mm}
Set $\upsilon_0 = \lambda \varepsilon_p/(4\kappa)$ and $\upsilon_1 = \varepsilon_p^2/4$ \\
Set $M_0 = \widetilde{O}\left(\kappa^4/\Big(\lambda^4 \varepsilon^2 \gamma^2\Big) \log(1/\delta)\right)$ and $M_1 = \widetilde{O}\left(\kappa^2/\Big(\lambda^4 \varepsilon^4\Big) \log(1/\delta)\right)$ \\
Prepare $\ket{\widehat{\beta}}$ that is $\upsilon_1$-close to the coefficient state $\ket{\beta} = \vec{\beta}/\norm{\vec{\beta}}_2$, where $\vec{\beta} = (G+\lambda \id)^{-1} \Phi^\dagger \ket{\psi}$ using Claim~\ref{claim:prep_coeff_state} \\
Output an $\upsilon_0$-close estimate $\ket{\widehat{\beta}}$ of $\ket{\widetilde{\beta}}$ using tomography protocol of Theorem~\ref{thm:qst} and $M_0$ copies of $\ket{\psi}$. \\
Obtain an $\upsilon_1$-close estimate $b$ of $\norm{\widetilde{\beta}}_2$ using Claim~\ref{claim:prep_coeff_state}$(ii)$ and $M_1$ copies of $\ket{\psi}$. \\
Set $\widehat{\beta}_j \leftarrow \sqrt{b} \cdot \widehat{\beta}_j$. \\
\Return List of coefficients $B=\{\widehat{\beta}_i\}_i$.
\end{algorithm}
\end{myalgorithm}

In the previous section we showed how to learn a set of states $L=\{\ket{\phi_i}\}_{i\in [\kappa]}$ in $\calC$ that satisfied the structure learning promise (Theorem~\ref{thm:structure_learning}). In this section, we show how to learn the \emph{coefficients} of these states $\{\ket{\phi_i}\}$ and thereby obtain a classical description of $\ridgeproj_{L,\lambda} \ket{\psi}$. Crucial to proving our main theorem is the following lemma which shows how to determine the projection of the state $\ridgeproj_{L,\lambda} \ket{\psi}$ up to a global phase for an an arbitrary quantum state $\ket{\psi}$ using only copies of $\ket{\psi}$. We state the lemma in full generality below since we will use it as a blackbox in another context.

\begin{lemma}
\label{lem:estimate_projection_psi}
Let $\calC$ be a class of  states satisfying Definition~\ref{def:model_class}, $\lambda > 0$, $k \in \mathbb{N}$ and $\varepsilon,\gamma,\delta \in (0,1)$. Suppose $\ket{\psi}$ is an unknown state. Let $L=\{\ket{\phi_i}\}_{i \in [k]}$ be a list of known states from $\calC$ such~that 
$$
\norm{(G+\lambda \id)^{-1} \Phi^\dagger \ket{\psi}}_2 \geq \gamma,
$$
where $\Phi$ is the dictionary matrix over $L$ and $G = \Phi^\dagger \Phi$ is the Gram matrix. Then, there is an algorithm that, with probability $\geq 1 - \delta$, outputs  $\widehat{\beta}\in \mathbb{C}^k$ such that for some $\theta \in (0,\pi]$
\begin{align}
\label{eq:conditiononphii-psi}
\norm{\sum_{i=1}^\kappa \widehat{\beta_i} \ket{\phi_i} - e^{i\theta} \ridgeproj_{L,\lambda} \ket{\psi}}_2 \leq \varepsilon,
\end{align}
The complexity of the algorithm is: 
\begin{align*}
&\text{ Sample complexity: }\widetilde{O}\left( \frac{k^4}{\lambda^4 \varepsilon^2 \gamma^2} \log\frac{1}{\delta}\right), \quad \text{ Time complexity: } \widetilde{O}\left( \frac{k^5(k+\lambda)G_{\prep} }{\lambda^5 \varepsilon^4 \gamma^2} \log \frac{1}{\gamma \varepsilon} \cdot \log\frac{1}{\delta}\right).
\end{align*}
\end{lemma}
\begin{proof}
Let $m = \ceil{\log k}$. Let us define the vector $\beta = (G + \lambda \id)^{-1} \Phi^\dagger \ket{\psi}$. Using Claim~\ref{claim:prep_coeff_state}$(i)$, we produce the $m$-qubit state $\ket{\widetilde{\beta}}$ such that $\norm{\ket{\widetilde{\beta}} - \ket{\beta}} \leq \upsilon_0$. We can then use state tomography (using Theorem~\ref{thm:qst}) to learn a description $\ket{\widehat{\beta}}$  of $\ket{\widetilde{\beta}}$ up to $\upsilon_0$ error. So by triangle inequality, we have that
$$
\norm{\ket{\widehat{\beta}} - \ket{\beta}} \leq 2 \upsilon_0.
$$
Moreover, using Claim~\ref{claim:prep_coeff_state}$(ii)$, we can output an estimate $b \geq 0$ that approximates the norm of $\beta$
$$
\Big|b - \norm{\beta}_2^2 \Big| \leq \upsilon_1
$$
Noting that $\ridgeproj_{L,\lambda} = \Phi \beta$ and $\sum_{i=1} \widehat{\beta}_i \ket{\phi_i} = \Phi \widehat{\beta} = \sqrt{b} \Phi \ket{\widehat{\beta}}$, we then have
\begin{align*}
\norm{\Phi \widehat{\beta} - \Phi \beta}_2 \leq \norm{\Phi}_2 \norm{\widehat{\beta} - \beta}_2 
&\leq \sqrt{k} \norm{\sqrt{b} \ket{\widehat{\beta}} - \norm{\beta} \ket{\beta}} \\
&\leq \sqrt{k} \Big(\norm{\sqrt{b} - \norm{\beta}} + \norm{\beta} \norm{\ket{\widehat{\beta} - \ket{\beta}}}\Big)\\
&\leq \sqrt{k} \Big(\sqrt{\upsilon_1} + 2\frac{\sqrt{k}}{\lambda} \upsilon_0 \Big)\\
&\leq \sqrt{\upsilon_1} + \frac{2k}{\lambda} \upsilon_0,
\end{align*}
where we have used $\norm{\beta} = \norm{(G+\lambda \id)^{-1} \Phi^\dagger \ket{\psi}} \leq \norm{\Phi^\dagger}/\lambda \leq \sqrt{k}/\lambda$ in the third line. Setting $\upsilon_0 = \lambda \varepsilon/(4k)$ and $\upsilon_1 = \varepsilon^2/4$ gives us the desired error of $\varepsilon$. We thus need to prepare $O(k^3/(\lambda^2 \varepsilon^2) \log(1/\delta))$ copies of the $m$-qubit state $\ket{\beta}$ for state tomography (Theorem~\ref{thm:qst}) which in turn requires 
$$
\widetilde{O}\left( \frac{k^4}{\lambda^4 \varepsilon^2 \gamma^2} \log(1/\delta)\right)
$$
applications of the corresponding circuit as defined in Claim~\ref{claim:prep_coeff_state}$(i)$ (where we have noted that the probability of success here is at least $\gamma$ here) and hence that many copies of $\ket{\psi}$ as well. The total gate complexity corresponding to this is then
$$
\widetilde{O}\left( \frac{k^5 (k+\lambda)G_{\prep}}{\lambda^5 \varepsilon^2 \gamma^2} \log \frac{1}{\varepsilon \gamma} \log \frac{1}{\delta} \right).
$$
Additionally, we need $O(k^2/(\upsilon_1^2\lambda^4)\log(1/\delta))=O(k^2/(\varepsilon^4 \lambda^4)\log(1/\delta))$ copies of $\ket{\psi}$ for estimation of $b$ and the corresponding overall gate complexity is
$$
\widetilde{O}\left( \frac{k^3(k+\lambda)G_{\prep} }{\varepsilon^4 \lambda^5} \log \frac{k}{\varepsilon \lambda} \right).
$$
Putting these together gives us the desired result. This completes the proof.
\end{proof} 
With the lemma above, we are now ready to state our main theorem corresponding to parameter learning (Algorithm~\ref{algo:parameter_learning}).
\begin{theorem}[Parameter learning]
\label{thm:parameter_learning}
Let $\lambda > 0$, $\kappa \in \mathbb{N}$ and $\varepsilon_p,\gamma,\delta \in (0,1)$. Suppose $\ket{\psi}$ is an unknown $n$-qubit state. Let $L = \{\ket{\phi_i}\}_{i \in [\kappa]}$ be a list of $\kappa$ known states such that $$
\norm{(G+\lambda \id)^{-1} \Phi^\dagger \ket{\psi}}_2 \geq \gamma,
$$
where $\Phi$ is the dictionary matrix over $L$, $G = \Phi^\dagger \Phi$ is the Gram matrix, and $\ket{\psi}$ can be expressed~as
$$
\ket{\psi} = \ridgeproj_{L,\lambda} \ket{\psi} + \alpha_{\kappa + 1} \ket{\psi_{\kappa + 1}},
$$
where $|\alpha_{\kappa+1}|^2 \cdot \calF_{\calS}(\ket{\psi_{\kappa+1}}) < \varepsilon_p$. Then, there exists an algorithm that with probability at least $1-\delta$, outputs coefficients $\{\widehat{\beta}_i\}_{i \in [\kappa]}$ such that $\ket{\psi}$ can be written (up to a global phase) as
$$
\ket{\psi} = \sum_{i=1}^\kappa \widehat{\beta}_i \ket{\phi_i} + \widetilde{\alpha}_{\kappa+1} \ket{\widetilde{\psi}_{\kappa + 1}},
$$
where $|\widetilde{\alpha}_{\kappa+1}|^2 \cdot \calF_{\calS}(\ket{\widetilde{\psi}_{\kappa+1}}) < 2 \varepsilon_p$ and $\norm{\sum_{i=1}^\kappa \widehat{\beta}_i \ket{\phi_i}}_2 \leq \norm{\ridgeproj_{L,\lambda}\ket{\psi}}_2 + \sqrt{\varepsilon_p}$. The complexity of the algorithm is as~follows:
\begin{align*}
&\text{ Sample complexity: }\widetilde{O}\left( \frac{\kappa^4}{\lambda^4 \varepsilon_p \gamma^2} \log \frac{1}{\delta}\right),\quad 
\text{ Time complexity: } \widetilde{O}\left( \frac{\kappa^5(\kappa+\lambda)G_{\prep} }{\lambda^5 \varepsilon_p^2 \gamma^2} \log\frac{1}{\gamma \varepsilon_p} \cdot \log \frac{1}{\delta} \right)
\end{align*}
\end{theorem}
\begin{proof}
We are given a list $L$ of $\kappa$ many states in $\C$ such that 
$$
\ket{\psi} = \ridgeproj_{L,\lambda} \ket{\psi} + (\id - \ridgeproj_{L,\lambda})\ket{\psi},
$$
along with the promise that the residual state satisfies the following for any state $\ket{\varphi} \in \C$:
\begin{equation}\label{eq:interim_input_promise}
|\la \varphi |(\id - \ridgeproj_{L,\lambda})\ket{\psi})| = |\alpha_{\kappa + 1}| \cdot |\la \varphi | \psi_{\kappa + 1}\ra| \leq |\alpha_{\kappa + 1}| \cdot \sqrt{\calF_{\calC}(\ket{\psi_{\kappa + 1}}} \leq \sqrt{\varepsilon_p},    
\end{equation}
where we used the given promise on the residual state. 

Let $\varepsilon_1 \in (0,1)$ be an error parameter to be fixed later. Using Lemma~\ref{lem:estimate_projection_psi}, we can learn a set of $\kappa$ coefficients $\{\widehat{\beta}_i\}_{i \in [\kappa]}$ corresponding to the states in $L$ such that
\begin{equation}\label{eq:interim_good_coeffs}
\norm{e^{i \theta} \sum_{i=1}^\kappa \widehat{\beta}_i \ket{\phi}_i - \ridgeproj_{L,\lambda} \ket{\psi} } \leq \varepsilon_1,    
\end{equation}
up to some (unknown) global phase $\theta \in (0,\pi]$. Let us denote the state vector $\widehat{\phi} = e^{i \theta} \sum_{i=1}^\kappa \widehat{\beta}_i \ket{\phi}_i$ and the corresponding residual state vector $\widehat{\phi}_R = \ket{\psi} - \widehat{\phi}$. We will also denote $\widetilde{\alpha}_{\kappa + 1} = \norm{\widehat{\phi}_R}$ and $\ket{\widetilde{\psi}_{\kappa + 1}} = \widehat{\phi}_R/\widetilde{\alpha}_{\kappa + 1}$ which is a valid quantum state.

We observe that $\widehat{\phi}_R$ is close to the true residual state $(\id - \ridgeproj_{L,\lambda})\ket{\psi}$:
\begin{equation}\label{eq:interim_close_residual_state}
\norm{\widehat{\phi}_R - (\id - \ridgeproj_{L,\lambda})\ket{\psi}} = \norm{\ket{\psi} - \widehat{\phi} - (\id - \ridgeproj_{L,\lambda})\ket{\psi}} = \norm{\ridgeproj_{L,\lambda}\ket{\psi} - \widehat{\phi}} \leq \varepsilon_1,
\end{equation}
where we have used the definition of $\widehat{\phi}_R$ in the first equality and Eq.~\eqref{eq:interim_good_coeffs} in the final inequality.

Let us denote the error $e_R = \widehat{\phi}_R - (\id - \ridgeproj_{L,\lambda})\ket{\psi}$, which we have just shown satisfies $\norm{e_R} \leq \varepsilon_1$. For any $\ket{s} \in \C$, we then have that
\begin{align*}
|\widetilde{\alpha}_{\kappa + 1}| \cdot |\la s | \widetilde{\psi}_{\kappa+1} \ra| = |\la s | \widehat{\phi}_R \ra| 
= \left|\la s | \Big((\id - \ridgeproj_{L,\lambda})\ket{\psi}\Big) + \la s | e_R \ra \right|
&\leq \left|\la s | \Big((\id - \ridgeproj_{L,\lambda})\ket{\psi}\Big)\right| + \norm{e_R} \\
&\leq \sqrt{\varepsilon_p} + \varepsilon_1
\end{align*}
where we used the triangle inequality in the third inequality along with the fact that $|\la s | e_R \ra| \leq \norm{e_R}$ and used Eq.~\eqref{eq:interim_input_promise} in the second line as well as that $\norm{e_R} \leq \varepsilon_1$. We set $\varepsilon_1 = 0.4 \sqrt{\varepsilon_p}$ which gives us~that
$$
|\widetilde{\alpha}_{\kappa + 1}| \cdot |\la s | \widetilde{\psi}_{\kappa+1} \ra| \leq \sqrt{2 \varepsilon_p} , \,\forall \ket{s} \in \C.
$$
As this is true for any $\ket{s} \in \C$ and in particular the one that maximizes $\calF_{\calC}(\ket{\widetilde{\psi}_{\kappa + 1}})$, we have the following implication
$$
|\widetilde{\alpha}_{\kappa + 1}|^2 \cdot \calF_{\calC}(\ket{\widetilde{\psi}_{\kappa + 1}}) \leq 2 \varepsilon_p.
$$
This gives us the desired guarantee on the residual state. To argue the upper bound on the norm of $\sum_{i=1}^\kappa \widehat{\beta}_i \ket{\phi_i}$, we observe that
\begin{align*}
\norm{\sum_{i=1}^\kappa \widehat{\beta}_i \ket{\phi_i}} = \norm{e^{i\theta} \sum_{i=1}^\kappa \widehat{\beta}_i \ket{\phi_i} - \ridgeproj_{L,\lambda}\ket{\psi} + \ridgeproj_{L,\lambda}\ket{\psi}} 
&\leq \norm{e^{i\theta} \sum_{i=1}^\kappa \widehat{\beta}_i \ket{\phi_i} - \ridgeproj_{L,\lambda}\ket{\psi}} + \norm{\ridgeproj_{L,\lambda}\ket{\psi}} \\
&\leq \sqrt{\varepsilon_p} + \norm{\ridgeproj_{L,\lambda}\ket{\psi}},
\end{align*}
where we again used Eq.~\eqref{eq:interim_good_coeffs} in the second line. The complexity of the protocol is due to learning of the coefficients $\{ \widehat{\beta}_i \}$ which consumes $\widetilde{O}\Big(\kappa^4/(\lambda^4 \varepsilon_p \gamma^2) \log(1/\delta)\Big)$ copies of $\ket{\psi}$ and $\widetilde{O}\Big(\kappa^5(\kappa + \lambda)G_{\prep}/(\lambda^5 \varepsilon_p^2 \gamma^2) \log(1/(\gamma \varepsilon_p)) \log(1/\delta) \Big)$ time complexity using Lemma~\ref{lem:estimate_projection_psi} (where the relevant error parameter is $\varepsilon_1 = 0.4\sqrt{\varepsilon_p}$). This completes the proof.
\end{proof}

\paragraph{Putting everything together.}
We are now ready to provide a proof of Theorem~\ref{thm:learning_decompositions}.
\begin{proof}[Proof of Theorem~\ref{thm:learning_decompositions}]
The proof follows immediately from our structure learning (Theorem~\ref{thm:structure_learning}) and parameter learning statements (Theorem~\ref{thm:parameter_learning}). Let $\varepsilon_s, \lambda \in (0,1)$ be parameters to be fixed later and run Algorithm~\ref{algo:structure_learning} of Theorem~\ref{thm:structure_learning} with these parameter values. It produces a list $L = \{\ket{\phi_i}\}_{i \in [\kappa]}$ of $\kappa = O(1/(\varepsilon_s \eta(\varepsilon_s))$ many states in $\C$ such that $\ket{\psi}$ can be expressed as
\begin{equation}
\ket{\psi} = \ridgeproj_{L,\lambda} \ket{\psi} + \alpha_{\kappa + 1}\ket{\psi_{\kappa+1}}, \text{ where } |\alpha_{\kappa + 1}|^2\cdot \calF_{\calC}(\ket{\psi_{\kappa + 1}}) \leq (3/2) \varepsilon_s + \eta^{-1}(3 \eta(\varepsilon_s)/2).    
\end{equation}
Let us denote $\varepsilon_p = (3/2) \varepsilon_s + \eta^{-1}(3 \eta(\varepsilon_s)/2)$ and $\uptau = \varepsilon_s \eta(\varepsilon_s)/(4(1+\lambda))$. We also have the promise that $\norm{(G + \lambda \id)\Phi^\dagger \ket{\psi}}^2 \geq \kappa \uptau/(\kappa + \lambda)$ (Claim~\ref{claim:norm_projected_state}). At this point, we use Algorithm~\ref{algo:parameter_learning} of Theorem~\ref{thm:parameter_learning} with $\varepsilon_p$ as defined above and $\gamma^2 = \kappa \uptau/(\kappa + \lambda)$. We then determine a list of coefficients $\{\widehat{\beta}_i\}_{i \in [\kappa]}$ corresponding to the states $\{\ket{\phi_i}\}_{i \in [\kappa]}$ in the list $L$ such that $\ket{\psi}$ can be written (up to a global phase) as
$$
\ket{\psi} = \sum_{i = 1}^\kappa \widehat{\beta}_i \ket{\phi_i} + \alpha \ket{\psi_{\kappa + 1}}
$$
such that $|\alpha|^2 \cdot \calF_{\calC}(\ket{\psi_{\kappa + 1}} \leq 2 \varepsilon_p$. Setting $\lambda=1/2$ and $\varepsilon_p = \varepsilon/2$ where $\varepsilon$ give us the desired error of Theorem~\ref{thm:learning_decompositions}. We thus need to pick $\varepsilon_s$ such that $\varepsilon_s \leq \min\{\varepsilon/3, \eta^{-1}(2/3 \eta(\varepsilon/2))\}$. This will depend on the function $\eta(\cdot)$ itself.

The sample and time complexity of using Algorithm~\ref{algo:structure_learning} is then $\widetilde{O}\Big(S_{\wal} \cdot \poly(1/\varepsilon,1/\eta(\varepsilon),\log 1/\delta) \Big)$ and $\widetilde{O}\Big(T_{\wal} \cdot \poly( 1/\varepsilon,1/\eta(\varepsilon),\log 1/\delta)) +  \poly(G_{\prep},T_{\prep}, 1/\varepsilon,1/\eta(\varepsilon)\Big)$ from Theorem~\ref{thm:structure_learning}, respectively. The sample and time complexity of using Algorithm~\ref{algo:parameter_learning} is $\widetilde{O}\Big(\kappa^4/(\varepsilon^2 \eta(\varepsilon/3)) \log(1/\delta) \Big)$ and $\widetilde{O}\Big(\kappa^6 G_{\prep}/(\varepsilon^3 \eta(\varepsilon/3)) \log(1/\delta) \Big)$ from Theorem~\ref{thm:parameter_learning}. The overall sample and time complexity is then as stated. This completes the proof.
\end{proof}

\begin{remark}\label{remark:approx_state_prep}
We have given the proof of Theorem~\ref{thm:learning_decompositions} assuming that for every $\ket{\phi} \in \C$, there is a classical algorithm that outputs a circuit $U$ that prepares $\ket{\phi}$ exactly as $\ket{\phi} = U \ket{0^n}$. This is used as part of subroutines to prepare the residual state (Claim~\ref{claim:prep_ridgeproj_residual_state}) for Algorithm~\ref{algo:structure_learning}), the coefficient state and estimation of the norm of the coefficient vector (Claim~\ref{claim:prep_coeff_state} for Algorithm~\ref{algo:parameter_learning}). We remark that both  claims can actually tolerate approximate preparation of states in $\C$. Claim~\ref{claim:circ_UPhi} will be a circuit that prepares states in $\C$ up to some $\varepsilon_{\prep}$ error (i.e., $\norm{\ket{\phi} - U \ket{0^n}} \leq \varepsilon_{\prep}$). This will result in the Gram matrix $G$ having an approximate block-encoding. Since we already have an approximate block-encoding of $(G + \lambda \id)^{-1}$ (due to approximation of the inverse function by a polynomial), we can tolerate an approximation error in $G$ if we use Lemma~\ref{lem:inverse_BE2}. We need to ensure that $\varepsilon_{\prep}$ is chosen such that the overall error in the block-encoding $(G + \lambda \id)^{-1}$ is lower than those stated in Claims~\ref{claim:prep_coeff_state} and \ref{claim:prep_ridgeproj_residual_state}.
\end{remark}

\section{Utility of structured decompositions}
We have so far discussed how to learn a structured decomposition of any $n$-qubit quantum state $\ket{\psi}$ with respect to a model class $\calC$ (Theorem~\ref{thm:learning_decompositions}, Section~\ref{sec:algo_decomposition_thm}). There are two natural questions, what is the utility of such a decomposition and secondly, what if $\ket{\psi}$ itself was \emph{structured}, could we say anything about the properties of the learned decomposition? We answer both these questions below by giving a few applications of our main result.
\begin{enumerate}
    \item We first show that the inner product of an arbitrary state $\ket{\psi}$ with any state with low extent (Eq.~\eqref{def:extent_C}) with respect to model class $\calC$ (Definition~\ref{def:model_class}) can be approximated with using a ``mimicking state" corresponding to the learned decomposition. 
    \item Given a weak agnostic learner $\calA_{\wal}$ for model class $\calC$, we show that Theorem~\ref{thm:learning_decompositions} boosts this to a strong agnostic learner and hence, the learned decomposition can be used to accomplished improper agnostic learning of $\calC$.
    \item We give an algorithm to learn states with low extent with respect to model class $\calC$.
\end{enumerate}
We state these results explicitly in the subsections below and give their proofs. 

\subsection{Mimicking state for low-extent fidelities}
\begin{lemma}
\label{cor:low_extent_inner_prod}
Let $\calC$ be a model class of $n$-qubit quantum states satisfying Definition~\ref{def:model_class}. Let $\xi > 0, \varepsilon \in (0,1)$, $\uptau = \eta(\varepsilon^2/\xi^2)$, and $\calC(\xi)$ be the set of $n$-qubit states with extent at most $\xi$ with respect to model class $\calC$. For every $n$-qubit  state $\ket{\psi}$, we can find a state~$\ket{\phi}$ which is expressible as a superposition of $O(\xi^2/(\varepsilon^2 \uptau))$ many states in $\C$, using $\widetilde{O}\Big(S_{\wal} \cdot \poly(\xi, 1/\varepsilon, 1/\uptau,\log 1/\delta) \Big)$ copies of $\ket{\psi}$ and in $\widetilde{O}\Big(T_{\wal} \cdot \poly(\xi, 1/\varepsilon,1/\uptau,\log 1/\delta)) +  S_{\wal} \cdot \poly(G_{\prep},T_{\prep}, \xi, 1/\varepsilon, 1/\uptau)\Big)$ time such~that
\begin{itemize}
    \item[$(i)$] $\Big| |\langle \sigma|\psi\rangle| - |\langle \sigma|\phi  \ra| \Big| \leq \varepsilon$ for every  $\ket{\sigma} \in \calC(\xi)$, 
    \item[$(ii)$] $\Big|\calF_{\calC(\xi)}(\ket{\phi}) - \calF_{\calC(\xi)}(\ket{\psi}) \Big| \leq 3 \varepsilon,$
\end{itemize}
where $\calF_{\calC(\xi)}(\ket{\psi}) = \max_{\ket{\sigma} \in \calC(\xi)} |\la \sigma | \psi \ra|^2$.
\end{lemma}
We skip the proof of this lemma and refer a reader to~\cite[Section~7.1]{ad2025structure}.

\subsection{Boosting weak agnostic learners}
In this section, we show that the algorithmic decomposition theorem (Theorem~\ref{thm:learning_decompositions}) can be used to boost weak agnostic learners of $\calC$ into a strong agnostic learner. We thus have the following result regarding improper agnostic learning of states in $\calC$, restated below for convenience.

\agnboosting*

\begin{proof}
We will use the algorithm of Theorem~\ref{thm:learning_decompositions} with a small modification in the end. Towards that, let $\varepsilon_s, \lambda \in (0,1)$ be parameters to be fixed later and run Algorithm~\ref{algo:structure_learning} of Theorem~\ref{thm:structure_learning} with these parameter values. It produces a list $L = \{\ket{\phi_i}\}_{i \in [\kappa]}$ of $\kappa = O(1/(\varepsilon_s \eta(\varepsilon_s))$ many states in $\C$ such that $\ket{\psi}$ can be expressed as
\begin{equation}\label{eq:interim_decomp_psi}
\ket{\psi} = \ridgeproj_{L,\lambda} \ket{\psi} + \alpha_{\kappa + 1}\ket{\psi_{\kappa+1}}, \text{ where } |\alpha_{\kappa + 1}|^2\cdot \calF_{\calC}(\ket{\psi_{\kappa + 1}}) \leq \max\{(3/2) \varepsilon_s, \, \eta^{-1}(3 \eta(\varepsilon_s)/2)\}.    
\end{equation}
We also have the promise that $\norm{\ridgeproj_{L,\lambda} \ket{\psi}} \geq \varepsilon_s \eta(\varepsilon_s)/(4(1+\lambda))$ (Theorem~\ref{thm:structure_learning}) as $\kappa \geq 1$ here as we are guaranteed to run at least one iteration in Algorithm~\ref{algo:structure_learning} here as long as $\varepsilon_s < \opt$. Let us denote $\varepsilon_p = \max\{(3/2) \varepsilon_s, \, \eta^{-1}(3 \eta(\varepsilon_s)/2)\}$ and $\uptau = \varepsilon_s \eta(\varepsilon_s)/(4(1+\lambda))$. 

We now argue that $\ridgeproj_{L,\lambda}\ket{\psi}$ achieves near-maximal fidelity that states in $\C$ would have with $\ket{\psi}$. Let $\ket{\varphi} \in \C$ be the state in $\C$ that achieves maximal fidelity with $\ket{\psi}$ i.e., $|\langle \varphi| \psi \rangle|^2 = \opt = \calF_{\C}(\ket{\psi})$. Starting from Eq.~\eqref{eq:interim_decomp_psi}, we  have
\begin{align*}
    |\langle \varphi| \psi \rangle| \leq |\langle \varphi|(\ridgeproj_{L,\lambda} \ket{\psi})| + |\alpha_{\kappa+1}| \cdot | \langle \varphi | \psi_{\kappa + 1} \rangle| \leq |\langle \varphi|(\ridgeproj_{L,\lambda} \ket{\psi})| + \sqrt{\varepsilon_p} \implies |\langle \varphi| \psi \rangle| - |\langle \varphi|(\ridgeproj_{L,\lambda} \ket{\psi})| \leq \sqrt{\varepsilon_p},
\end{align*}
where we used $|\alpha_{\kappa+1}|\cdot | \langle \varphi | \psi_{\kappa + 1} \rangle| \leq |\alpha_{\kappa+1}| \cdot \sqrt{\calF_{\C}(\ket{\psi_{\kappa+1}})} \leq \sqrt{\varepsilon_p}$. We can then show
\begin{align}
& |\langle \varphi| \psi \rangle|^2 - |\langle \varphi|(\ridgeproj_{L,\lambda} \ket{\psi})|^2 = \Big(|\langle \varphi| \psi \rangle| + |\langle \varphi|(\ridgeproj_{L,\lambda} \ket{\psi})|\Big)\Big(|\langle \varphi| \psi \rangle| - |\langle \varphi|(\ridgeproj_{L,\lambda} \ket{\psi})|\Big) \leq 2 \sqrt{\varepsilon_p} \\
\implies & |\langle \varphi|(\ridgeproj_{L,\lambda} \ket{\psi})|^2 \geq |\langle \varphi| \psi \rangle|^2 - 2 \sqrt{\varepsilon_p} = \opt - 2 \sqrt{\varepsilon_p},
\label{eq:fidelity_promise_projection}
\end{align}
where we have used $|\langle \varphi| \psi \rangle|, |\langle \varphi|(\Lambda_T \ket{\psi})| \leq 1$ and $|\langle \varphi| \psi \rangle|^2 = \opt$ in the final implication. Consequentially, by Cauchy-Scharwz, we have that
\begin{equation}\label{eq:inter_high_norm_proj_state}
\opt - 2 \sqrt{\varepsilon_p} \leq |\langle \varphi|(\ridgeproj_{L,\lambda} \ket{\psi})|^2 \leq \norm{\ket{\phi}}^2 \cdot \norm{\ridgeproj_{L,\lambda} \ket{\psi}}^2 \leq \norm{\ridgeproj_{L,\lambda} \ket{\psi}}^2,
\end{equation}
where we have used that $\ket{\phi}$ is a normalized quantum state.

To solve the task of agnostic learning, we define the quantum state $\ket{\phi}=\ridgeproj_{L,\lambda} \ket{\psi}/\|\ridgeproj_{L,\lambda} \ket{\psi}\|$ and observe that
\begin{equation}\label{eq:phi_solves_agn_learning}
|\langle \psi|\phi\rangle|^2 = \frac{|\la \psi | \ridgeproj_{L,\lambda} | \psi \ra|^2 }{\norm{\ridgeproj_{L,\lambda}\ket{\psi}}^2} = \frac{|\la \psi | \ridgeproj_{L,\lambda} | \psi \ra|^2 }{\la \psi | ({\ridgeproj_{L,\lambda}})^2 | \psi \ra} \geq \la \psi | ({\ridgeproj_{L,\lambda}})^2 | \psi \ra = \norm{\ridgeproj_{L,\lambda}\ket{\psi}}^2 \geq \opt - 2 \sqrt{\varepsilon_p},
\end{equation}
where the first equality used the definition of $\ket{\phi}$, third equality used that $\ridgeproj_{L,\lambda}$ is a contraction (Fact~\ref{fact:ridge_proj_contraction}) and hence $\la \psi | {\ridgeproj_{L,\lambda}}^2 | \psi \ra \leq \la \psi | {\ridgeproj_{L,\lambda}} | \psi \ra$, and also noted that $\ridgeproj_{L,\lambda}$ is a Hermitian semipositive definite operator which gives us $|\la \psi | \ridgeproj_{L,\lambda} | \psi \ra| = \la \psi | \ridgeproj_{L,\lambda} | \psi \ra$. In the last equality, we used Eq.~\eqref{eq:inter_high_norm_proj_state}.

The goal is to then learn $\ket{\phi}$. Recall that so far, we have only learned the list $L$ of states in $\C$ such that $\ket{\phi}$ has the promise of Eq.~\eqref{eq:phi_solves_agn_learning}. Note that $\ridgeproj_{L,\lambda} \ket{\psi} = \sum_{i=1}^\kappa \beta_i \ket{\phi_i}$, where  $\ket{\phi_i}\in L$ and $\{\beta_i\}_{i \in [\kappa]}$ is the solution to the ridge regression problem of Eq.~\eqref{eq:ridge_regression}. From Claim~\ref{claim:norm_projected_state}, we have the promise that $\norm{\vec{\beta}}^2 \geq \kappa \uptau/(\kappa + \lambda)$. Let $\gamma^2 = \kappa \uptau/(\kappa + \lambda)$ and $\varepsilon_1 \in (0,1)$ be an error parameter to be decided later. We use Lemma~\ref{lem:estimate_projection_psi} (which is essentially Algorithm~\ref{algo:parameter_learning} of Theorem~\ref{thm:parameter_learning}) with error parameter $\varepsilon_1$ (i.e., $\varepsilon$ there is instantiated as $\varepsilon_1$) to determine a list of coefficients $\{\widehat{\beta}_i\}_{i \in [\kappa]}$ such that for some $\theta \in (0,\pi]$
\begin{align}
\label{eq:interim_conditionon_phi-psi}
\norm{\sum_{i=1}^\kappa \widehat{\beta_i} \ket{\phi_i} - e^{i\theta} \ridgeproj_{L,\lambda} \ket{\psi}}_2 \leq \varepsilon_1.
\end{align}
It then follows that
\begin{align}
\left|\langle \psi \Big| \Big(\sum_{i=1}^k \widehat{\beta}_i \ket{\phi_i} \Big) \right| 
&= \left|\langle \psi \Big| \Big(e^{i\theta} \ridgeproj_{L,\lambda} \ket{\psi} \Big) + \la \psi | \Big(\sum_{i=1}^k \widehat{\beta}_i \ket{\phi_i} - e^{i\theta} \ridgeproj_{L,\lambda} \ket{\psi} \Big) \right| \nonumber \\   
&\geq \left|\langle \psi \Big| \Big(e^{i\theta} \ridgeproj_{L,\lambda} \ket{\psi} \Big)\right| - \left|\la \psi | \Big(\sum_{i=1}^k \widehat{\beta}_i \ket{\phi_i} - e^{i\theta} \ridgeproj_{L,\lambda} \ket{\psi} \Big) \right| \nonumber \\
&\geq |\langle \psi | \ridgeproj_{L,\lambda} | \psi \rangle| - \varepsilon_1,
\label{eq:interim_fidelity}
\end{align}
where used Eq.~\eqref{eq:interim_conditionon_phi-psi} in the last line. 
Additionally,  we are guaranteed from Eq.~\eqref{eq:interim_conditionon_phi-psi} that
\begin{align}\label{eq:ub_L2_norm_beta}
\alpha:=\norm{\sum_{i=1}^\kappa \widehat{\beta_i} \ket{\phi_i}}_2 \leq \norm{\ridgeproj_{L,\lambda} \ket{\psi}}_2 + \varepsilon_1.
\end{align}
Consider the state $\ket{\widehat{\phi}}$ defined as
\begin{equation*}
\ket{\widehat{\phi}} = \frac{1}{\alpha}\sum_{i=1}^\kappa \widehat{\beta}_i \ket{\phi_i}.
\end{equation*}
Note that $\ket{\widehat{\phi}}$ is a valid normalized state as $\norm{\ket{\widehat{\phi}}} = 1$. The state $\ket{\widehat{\phi}}$ is thus an approximation of $\ket{\phi}$, which we wanted to learn. We now argue that $\ket{\widehat{\phi}}$ solves the problem of agnostic learning. Putting together the lower bound (Eq.~\eqref{eq:interim_fidelity}) and upper bound (Eq.~\eqref{eq:ub_L2_norm_beta}) from above, we have that
\begin{align}
|\langle \psi | \widehat{\phi} \rangle|^2 = \frac{\left|\langle \psi \Big| \Big(\sum_{i=1}^k \widehat{\beta}_i \ket{\phi_i} \Big) \right|^2}{\alpha^2} 
&\geq \frac{|\langle \psi |\ridgeproj_{L,\lambda} |\psi \rangle|^2 - 2\varepsilon_1}{\norm{\ridgeproj_{L,\lambda} \ket{\psi}}^2 + 3\varepsilon_1} \\
&= \frac{|\langle \psi |\ridgeproj_{L,\lambda} |\psi \rangle|^2 - 2 \varepsilon_1}{\norm{\ridgeproj_{L,\lambda} \ket{\psi}}_2^2\Big(1 + \frac{3 \varepsilon_1}{\norm{\ridgeproj_{L,\lambda} \ket{\psi}}_2^2}\Big)} \\
&\geq \frac{|\langle \psi |\ridgeproj_{L,\lambda} |\psi \rangle|^2 - 2 \varepsilon_1}{\norm{\ridgeproj_{L,\lambda} \ket{\psi}}_2^2} \left(1 - \frac{3 \varepsilon_1}{\norm{\ridgeproj_{L,\lambda} \ket{\psi}}_2^2}\right) \\
&\geq \frac{|\langle \psi |\ridgeproj_{L,\lambda} |\psi \rangle|^2}{\norm{\ridgeproj_{L,\lambda} \ket{\psi}}_2^2} - \frac{2 \varepsilon_1}{\norm{\ridgeproj_{L,\lambda} \ket{\psi}}_2^2} - \frac{3 \varepsilon_1 |\langle \psi |\ridgeproj_{L,\lambda} |\psi \rangle|^2}{\norm{\ridgeproj_{L,\lambda} \ket{\psi}}_2^4} \\
&\geq |\la \psi | \phi \ra|^2 - \frac{2 \varepsilon_1}{\uptau^2} - \frac{3 \varepsilon_1 |\la \psi | \phi \ra|^2}{\uptau^2} \\
&\geq \opt - \varepsilon_p - \frac{5 \varepsilon_1}{\uptau^2},
\end{align}
where we used Eq.~\eqref{eq:interim_fidelity} and Eq.~\eqref{eq:ub_L2_norm_beta} in the second inequality in the first line, the fact that $1/(1+x) \geq 1 - x, \,\,\forall x \geq 0$ in the third line, used the definition of $\ket{\phi} := \ridgeproj_{L,\lambda} \ket{\psi}/\norm{\ridgeproj_{L,\lambda} \ket{\psi}}$ and $\norm{\ridgeproj_{L,\lambda} \ket{\psi}} \geq \uptau$ in the fifth line, where $\uptau = \varepsilon_s \eta(\varepsilon_s)/(4(\lambda + 1))$ from the promise of structure learning (in  Theorem~\ref{thm:structure_learning})\footnote{Here, we use the lower bound of $\norm{\ridgeproj_{L,\lambda} \ket{\psi}} \geq \uptau$ instead of that from Eq.~\eqref{eq:inter_high_norm_proj_state} as $\uptau$ is known to us whereas $\opt$~is~unknown.}. Setting $\varepsilon_1 = \varepsilon_p \uptau^2 /5$ and $\varepsilon_p = \varepsilon/2$ gives us the desired result of
$$
|\la \psi | \widehat{\phi} \ra|^2 \geq \opt - \varepsilon.
$$
This in turn requires setting $\varepsilon_s$ such that $\varepsilon_p = \max\{(3/2) \varepsilon_s, \, \eta^{-1}(3 \eta(\varepsilon_s)/2)\} = \varepsilon/2$. We thus need to pick $\varepsilon_s$ such that $\varepsilon_s \leq \min\{\varepsilon/3, \eta^{-1}(2/3 \eta(\varepsilon/2))\}$. This will depend on the function $\eta(\cdot)$ itself.

The sample and time complexity of using Algorithm~\ref{algo:structure_learning} is $\widetilde{O}\Big(S_{\wal} \cdot \poly(1/\varepsilon,1/\eta(\varepsilon),\log 1/\delta) \Big)$ and $\widetilde{O}\Big(T_{\wal} \cdot \poly( 1/\varepsilon,1/\eta(\varepsilon),\log 1/\delta)) +  \poly(G_{\prep},T_{\prep}, 1/\varepsilon,1/\eta(\varepsilon)\Big)$ from Theorem~\ref{thm:structure_learning}, respectively. The sample and time complexity of using Lemma~\ref{lem:estimate_projection_psi} is $\widetilde{O}\left(1/(\varepsilon^{11} \eta(\varepsilon/3)^9) \log(1/\delta) \right)$ and $\widetilde{O}\left(1/(\varepsilon^{19} \eta(\varepsilon/3)^{15}) \log(1/\delta) \right)$, respectively, where the relevant error parameter is $\varepsilon_1 = \varepsilon_p \uptau^2/5 = O(\varepsilon^3 \eta(\varepsilon/3)^2)$ (Eq.~\eqref{eq:interim_conditionon_phi-psi}), $\gamma^2 = \kappa \uptau/(\kappa + \lambda)$ and $\kappa = O(1/\varepsilon\eta(\varepsilon/3))$. The overall sample and time complexity is then as stated. This completes the proof.
\end{proof}

\subsection{Tomography of  states with bounded extent}\label{sec:ST_extent}
In this section, we state our main result of this paper. We give a tomography protocol for quantum states which have bounded extent with respect to model class $\calC$. In particular, we will prove the following main theorem.
\begin{theorem}
\label{thm:tomo_extent}
Let $\C$ be a model class as in Definition~\ref{def:model_class}. Let $\ket{\psi} \in \C(\xi)$. One can learn $\ket{\psi}$ up to trace distance $\varepsilon$ using  $\widetilde{O}\Big(S_{\wal} \cdot \poly(\xi, 1/\varepsilon, 1/\uptau,\log 1/\delta) \Big)$ copies of $\ket{\psi}$ and time complexity
$$
\widetilde{O}\Big(T_{\wal} \cdot \poly(\xi, 1/\varepsilon,1/\uptau,\log 1/\delta)) +  S_{\wal} \cdot \poly(G_{\prep},T_{\prep}, \xi, 1/\varepsilon, 1/\uptau)\Big),
$$
where $S_{\wal}, T_{\wal},T_{\prep},G_{\prep}$ are parameters of the model class $\C$.
\end{theorem}
We give two alternate proofs for the above theorem. The first proof directly follows from Lemma~\ref{cor:low_extent_inner_prod}. The second proof is meant to be more instructive and showcase how the algorithmic decomposition result of Theorem~\ref{thm:learning_decompositions} specializes to when the input unknown state $\ket{\psi}$ is promised to have bounded extent with respect to $\C$.

\paragraph{A simple proof.}
\begin{proof}
We instantiate Lemma~\ref{cor:low_extent_inner_prod} with error parameter $\varepsilon' \in (0,1)$ to be fixed later: this yields an algorithm that given copies of $\ket{\psi}$, outputs an unnormalized state $\ket{\phi}$ such that $$\Big| |\langle \sigma|\psi\rangle| - |\langle \sigma|\phi  \ra| \Big| \leq \varepsilon'$$ for every  $\ket{\sigma} \in \calC(\xi)$. Now, since $\ket{\psi}\in \C(\xi)$, we  have that
\begin{equation}
    \varepsilon' \geq \left|\langle \psi|\psi\rangle - \langle \psi|\phi\ra\right|=\left|1 - \langle \psi|\phi\ra\right| \geq \left| 1 - |\la \psi | \phi \ra| \right| \geq 1 - |\la \psi | \phi \ra|,
\end{equation}
where the second inequality used reverse triangle inequality. The above now implies $ |\la \psi | \phi \ra|^2 \geq (1-\varepsilon')^2 \geq 1 - 2\varepsilon'$. 
Furthermore, this implies that the norm of $\ket{\phi}$ can be lower bounded as
\begin{align}
\label{eq:normofpsitilde}
1 - 2 \varepsilon'\leq |\la \psi | \phi \ra|^2\leq \norm{\ket{\phi}}^2 \cdot \norm{\ket{\psi}}^2  =
\norm{\ket{\phi}}^2,
\end{align}
where the second inequality follows from Cauchy-Schwartz. Also, the norm of $\ket{\phi}$ can be bounded from above as $\|\ket{\phi}\|\leq 1+\varepsilon'$ using Theorem~\ref{thm:parameter_learning}. 
Consider the state $\ket{\phi'} = \ket{\phi}/\norm{\ket{\phi}}$ (which is known to the algorithm since $\ket{\phi}$ is produced by the algorithm explicitly). We then have
$$
|\la \phi' | \psi \ra|^2 = \frac{|\la \phi | \psi \ra|^2}{\norm{\ket{\phi}}^2} \geq \frac{1 - 2\varepsilon'}{1 + \varepsilon'} \geq 1 - 3\varepsilon'.
$$
Setting $\varepsilon' = \varepsilon/3$ gives us the desired result. The sample and time complexity then follows from Lemma~\ref{cor:low_extent_inner_prod}. This completes the proof.
\end{proof}

\paragraph{An alternate proof.}
For the second proof of Theorem~\ref{thm:tomo_extent}, let us inspect how the algorithm of Theorem~\ref{thm:learning_decompositions} operates when the input state is promised to have extent at most $\xi$ with respect to $\C$ (Eq.~\eqref{def:extent_C}). Recall that the first stage of the algorithm involves structure learning (Algorithm~\ref{algo:structure_learning}). We  stop at the end of the $\kappa$th iteration as part of structure learning (see steps \ref{algo_step:est_fidelity} and \ref{algo_step:step_est_alphat}, Algorithm~\ref{algo:structure_learning}) if 
\begin{equation}
\alpha_{\kappa+1}^2 := \norm{(\mathbb{I} - \Pi_{\kappa,\lambda})\ket{\psi}}_2^2 < \varepsilon_s, \enspace \text{or} \enspace \calF_{\calS}(\ket{\psi_{\kappa+1}}) < \varepsilon_s,    
\end{equation}
where the residual state is $\ket{\psi_{\kappa+1}} = (\mathbb{I} - \ridgeproj_{\kappa,\lambda})\ket{\psi}/\alpha_{\kappa+1}$ and $\varepsilon_s$ is an user-defined input error parameter to structure learning (see Theorem~\ref{thm:structure_learning}). 

However, we have not yet exploited the fact that the unknown input state $\ket{\psi}$ is promised to be have extent $\leq \xi$. In particular, we have the following claim.
\begin{claim}\label{claim:stop_cond_extent}
Let $\varepsilon \in (0,1)$. Consider the context of Theorem~\ref{thm:tomo_extent}. Suppose we apply Algorithm~\ref{algo:structure_learning} of Theorem~\ref{thm:structure_learning}) to $\ket{\psi}$ then the following is true. If $|\alpha_t|^2 \geq \varepsilon$, then $\calF_{\calS}(\ket{\psi_t}) \geq \varepsilon/\xi^2$.
\end{claim}
\begin{proof}
Suppose $\ket{\psi}$ has the decomposition $\ket{\psi} = \sum_{i=1}^m c_i \ket{\phi_i}$ for some $m\in \mathbb{N}$ and where $\ket{\phi_i} \in \C, \forall i \in [m]$ with $c_i$ as its corresponding coefficients. We are promised that $\norm{c}_1 \leq \xi$. At the beginning of the $t$th iteration, we have the following decomposition of $\ket{\psi}$:
$$
\ket{\psi} = \ridgeproj_{t-1,\lambda} \ket{\psi} + \alpha_t \ket{\psi_t},
$$
where $\alpha_t = \norm{(\id - \ridgeproj_{t-1,\lambda})\ket{\psi}}$ and $\ket{\psi_t} = (\id - \ridgeproj_{t-1,\lambda})\ket{\psi}/\alpha_t$. From the definition of $\alpha_t$, we have
$$
\alpha_t^2 = \norm{(\id - \ridgeproj_{t-1,\lambda})\ket{\psi}}^2 = \la \psi |(\id - \ridgeproj_{t-1,\lambda})^2| \psi \ra \leq \la \psi | (\id - \ridgeproj_{t-1,\lambda}) | \psi \ra = \alpha_t |\la \psi | \psi_t \ra| \implies \alpha_t \leq |\la \psi| \psi_t \ra|
$$
where we have used that $(\id - \ridgeproj_{t-1,\lambda})$ is Hermitian in the second equality, that $0 \preceq (\id - \ridgeproj_{t-1,\lambda}) \preceq 1$ (for $\lambda \geq 0$) from Fact~\ref{fact:ridge_proj_contraction} in the third inequality and finally the definition of $\ket{\psi_t}$. Using the above bound on $\alpha_t$, we then have that
\begin{align}
& |\alpha_t| \leq |\la \psi_t | \psi \ra| = \Big| \sum_{i=1}^m c_i \la \psi_t | \phi_i \ra \Big| \leq \sum_{i=1}^m |c_i| \cdot |\la \psi_t | \phi_i \ra| \leq \sum_{i=1}^m |c_i| \cdot \sqrt{\calF_{\calC}(\ket{\psi_t}} \leq \xi(\ket{\psi}) \cdot \sqrt{\calF_{\calC}(\ket{\psi_t}} \\
& \implies \calF_{\calC}(\ket{\psi_t} \geq \frac{\alpha_t^2}{\xi(\ket{\psi})^2} \geq \frac{\varepsilon}{\xi^2},
\end{align}
where we have used the decomposition of $\ket{\psi}$ in the second equality, applied triangle inequality in the third inequality, noted that $|\la \psi_t | \phi \ra| \leq \sqrt{\calF_{\calC}(\ket{\psi_t}}, \forall \ket{\phi} \in \C$ by definition in the fourth inequality, and used $\xi(\ket{\psi}) = \sum_{i=1}^m |c_i|$ by definition of stabilizer extent. The implication follows from the promise that $\alpha_t^2 \geq \varepsilon$ and $\xi(\ket{\psi}) \leq \xi$. This completes the proof of the claim.
\end{proof}

We can thus simplify Algorithm~\ref{algo:structure_learning} when the unknown state $\ket{\psi}$ is promised to have extent at most $\xi$ with respect to $\C$. Until we accomplish learning a list $L = \{\ket{\phi_i}\}_{i \in [\kappa]}$ for which $\ridgeproj_{L,\lambda}\ket{\psi}$ solves the task of state tomography, we will learn a state $\ket{\phi} \in \C$ that has high fidelity with the residual state. This is formally described in the proof below.
\begin{proof}[Proof of Theorem~\ref{thm:tomo_extent}]
We will use the algorithm of Theorem~\ref{thm:learning_decompositions} with a small modification in the beginning. Towards that, let $\varepsilon_s, \lambda \in (0,1)$ be parameters to be fixed later and run Algorithm~\ref{algo:structure_learning} of Theorem~\ref{thm:structure_learning} with these parameter values. Suppose we are at the end of the $t$th iteration of Algorithm~\ref{algo:structure_learning}. We observe that by Claim~\ref{claim:stop_cond_extent}, if $\alpha_{t+1}^2 \geq \varepsilon_s$, then $\calF_{\calC}(\ket{\psi_{t+1}}) \geq \varepsilon/\xi^2$. As we estimate $\alpha_{t+1}$ up to error $\varepsilon_s/2$, we are guaranteed to learn a state $\ket{\phi_{t+1}} \in \C$ such that
$$
|\la \phi_{t+1} | \psi_{t+1} \ra|^2 \geq \eta(\varepsilon_s/(2\xi^2)).
$$
Thus, we no longer need to check the fidelity of the residual state $\ket{\psi_{t+1}}$ with $\C$ and step~\ref{algo_step:est_fidelity} can be skipped in Algorithm~\ref{algo:structure_learning}. After $\kappa = O((1+\lambda)/(\varepsilon_s \eta(\varepsilon_s/(2 \xi^2)))$ many iterations (Claim~\ref{claim:ub_kappa}), Algorithm~\ref{algo:structure_learning} produces a list $L = \{\ket{\phi_i}\}_{i \in [\kappa]}$ of states in $\C$ such that $\ket{\psi}$ can be expressed as
\begin{equation}
\ket{\psi} = \ridgeproj_{L,\lambda} \ket{\psi} + \alpha_{\kappa + 1}\ket{\psi_{\kappa+1}}, \text{ where } |\alpha_{\kappa + 1}|^2 \leq (3/2) \varepsilon_s. 
\end{equation}
We then argue that the state $\ket{\phi} = \ridgeproj_{L,\lambda} \ket{\psi}/\norm{\ridgeproj_{L,\lambda} \ket{\psi}}$ accomplishes the task of state tomography. To observe this, we first note that $\norm{\ridgeproj_{L,\lambda} \ket{\psi}} \leq \norm{\ket{\psi}} + \norm{(\id - \ridgeproj_{L,\lambda})\ket{\psi}} \leq 1 +  \sqrt{3\varepsilon_s/2}$. Moreover, 
$$
|\la \psi | \ridgeproj_{L,\lambda} | \psi \ra| = \Big|\la \psi | \psi \ra - \la \psi | (\id - \ridgeproj_{L,\lambda}) | \psi \ra\Big| \geq 1 - \Big|\la \psi | (\id - \ridgeproj_{L,\lambda}) | \psi \ra\Big| \geq 1 - \norm{(\id - \ridgeproj_{L,\lambda})\ket{\psi}} \geq 1 - \sqrt{3 \varepsilon_s/2}.
$$
Combining the above two observations gives us that $\ket{\phi}$ satisfies
$$
|\la \phi | \psi \ra|^2 \geq \left(\frac{1 - \sqrt{3 \varepsilon_s/2}}{1 + \sqrt{3 \varepsilon_s/2}}\right)^2 \geq (1 - \sqrt{6 \varepsilon_s})^2 \geq 1 - 2 \sqrt{6 \varepsilon_s}.
$$
We now proceed as in Theorem~\ref{thm:boostingagnlearning} to obtain a state $\ket{\widehat{\phi}}$ that is close to $\ket{\phi}$ by using Theorem~\ref{thm:parameter_learning} (via Lemma~\ref{lem:estimate_projection_psi}). The sample and time complexity of this algorithm is then as stated. This completes the proof.
\end{proof}

\section{Applications}
In the previous section, we showed how to use our decomposition theorem to boost weak to strong agnostic learners, and also perform tomography when the input state is promised to have structure. In this section we instantiate this tomography result for various specific model classes to obtain learning algorithms. Recall that when instantiating the tomography result for a specific model class, we are concerned with a couple of things (as part of the definition of a model class $\C$ in Definition~\ref{def:model_class}), the complexity of preparing states in $\C$, the complexity of weak agnostic learning~$\C$.

\subsection{Learning states with low stabilizer extent}
\subsubsection{Learning algorithm}
Firstly, $n$-qubit stabilizer states  are describable in time $O(n^2)$ so that satisfies requirement $(i)$ of the definition of model class $\C$. Next, we need an agnostic learner for this class. For the class of stabilizer states (denoted $\Sh$), recently there have been two agnostic learners proposed in the works of~\cite{ad2025structure,chen2024stabilizer}. The first work gives a polynomial-time \emph{weak} agnostic learner, albeit relies on a conjecture in additive combinatorics and the second work gives a quasi-polynomial time \emph{strong} agnostic learner unconditionally. We state their results below which we invoke eventually. 

\begin{theorem}[\cite{chen2024stabilizer}]
\label{thm:stab_bootstrapping}
Fix $\varepsilon\leq \tau \in (0,1)$. 
There is an algorithm that, given~copies of $\ket{\psi}$ with $\Fe_{\Sh}({\ket{\psi}}) \geq \tau$, with high probability outputs a $|\phi\rangle \in \Sh$ with $|\langle \phi | \psi \rangle|^2 \geq \tau - \varepsilon$. The sample~and time complexity of this algorithm is $n\cdot \poly(1/\varepsilon,(1/\tau)^{\log 1/\tau})$ and $\poly(n,1/\varepsilon, (1/\tau)^{\log 1/\tau})$~respectively.
\end{theorem}

\begin{theorem}[\cite{ad2025structure}]
    \label{thm:scforstabilizers}
    Let $\uptau>0$.    Assuming the high-doubling algorithmic $\PFR$ Conjecture,~there is an algorithm that, given~copies of $\ket{\psi}$ with $\Fe_{\Sh}({\ket{\psi}}) \geq \uptau$, with high probability outputs a $|\phi\rangle \in \Sh$ with $|\langle \phi | \psi \rangle|^2 \geq \tau^C$ (for a universal constant $C>1$) using $\poly(n,1/\uptau)$ time and copies of~$\ket{\psi}$.
\end{theorem}
The polynomial Freiman-Ruzsa ($\PFR$) conjecture was recently resolved by~\cite{gowers2023conjecture} where they showed the following:  Suppose $A \subseteq\mathbb{F}_2^{n}$ has doubling constant $K$ (i.e., the set $\{a+a':a,a'\in A\}$ has size at most $K$ times that of $A$), then $A$ is covered by at most $2K^{9}$ cosets of some subgroup $H \subset \textsf{span}(A)$ of size $|H| \leq |A|$. The algorithmic $\PFR$ conjecture that was raised also in~\cite{ad2025structure} conjectures that one can \emph{find} this subspace in polynomial time in the high-doubling regime of $K=\poly(n)$. For a detailed discussion on this, we defer the reader to~\cite{ad2025structure}.\footnote{We remark that a recent work \cite{algopfr} gave a $\poly(n,2^K)$-time algorithm for this task.}

Finally, we need an $\A_{\prep}$ algorithm as part of definition of model class $\C$. To that end,  we use the following result, that given the classical description of a stabilizer state, we can obtain a Clifford circuit preparing it via the following result.

\begin{lemma}[Clifford synthesis~\cite{dehaene2003clifford,patel2003efficient}]\label{lem:clifford_synthesis}
Given the classical description of an $n$-qubit stabilizer state $\ket{\phi}$, there is a quantum algorithm that outputs a Clifford circuit $\textsf{C}$ that prepares $\ket{\phi}$, using $O(n^2)$ many single-qubit and two-qubit Clifford gates.
\end{lemma}

Putting these together the proof of the first part of Theorem~\ref{result:learn_stab_extent} is immediate.

\begin{enumerate}
    \item \emph{Learning stabilizer decompositions.} Applying Theorem~\ref{thm:learning_decompositions} with the above subroutines then gives us a $\poly(n,(1/\varepsilon)^{\log(1/\varepsilon)})$-time algorithm for learning stabilizer decompositions of $\ket{\psi}$ in the sense of Theorem~\ref{thm:learning_decompositions} when using the agnostic learner of Theorem~\ref{thm:stab_bootstrapping} or a $\poly(n,1/\varepsilon)$-time algorithm when using the (conditional) agnostic learner of Theorem~\ref{thm:scforstabilizers}. We remark that this also implies algorithms for learning stabilizer decompositions of an input state $\ket{\psi}$ where the residual state has low Gowers-$3$ norm as a consequence of the polynomial equivalences of Gowers-$3$ norm with stabilizer fidelity (see~\cite{ad2024tolerant,bao2024tolerant,mehraban2024improved}).
    \item \emph{Efficient improper agnostic learning.} Observe that one can use Theorem~\ref{thm:boostingagnlearning} along with Theorem~\ref{thm:scforstabilizers} to get an efficient improper agnostic learner for the model class of stabilizer states, assuming the algorithmic polynomial $\PFR$ conjecture. As far as we know, even conditionally and improperly, we didn't know if stabilizer states are agnostically learnable 
    \item \emph{Tomography of low-stabilizer extent states.} Observe that one can use Theorem~\ref{thm:tomo_extent} along with Theorem~\ref{thm:stab_bootstrapping} to get $\poly(n,(\xi/\varepsilon
    )^{\log \xi/\varepsilon} )$   algorithm for $\varepsilon$-learning states $\ket{\psi}$ which have stabilizer extent $\xi$. If we were to assume the algorithmic $\PFR$ conjecture and used Theorem~\ref{thm:scforstabilizers}, then we'd even a $\poly(n,\xi,1/\varepsilon)$ algorithm (which we show to be optimal up to polynomial factors).
    \item \emph{Tomography of low-stabilizer rank states.} Finally For stabilizer-rank $k$ states, we use the recently established result of $\xi\leq k^k$~\cite{kalra2025stabilizer} to give an unconditional $\poly(n,k^{k^2})$ learning algorithm for stabilizer-rank $k$ states.
\end{enumerate}

\subsubsection{Lower bound}
 We now prove the second part of Theorem~\ref{result:learn_stab_extent}.  Below we  complement our learning result by showing that the $n,\xi$-dependence of our tomography algorithm is optimal up to polynomial~factors.
\begin{lemma}
\label{lem:packinglowerbound}
Let $n\geq 1$, $1\leq \xi\leq 2^{n/4},\varepsilon\leq 1/(2\sqrt{2})$. Let $\C(\xi)=\{\ket{\psi}: \xi(\ket{\psi})\leq \xi\}$. An $0.1$-error tomography protocol for $\C_{\xi}$ requires $\Omega(n+\xi^2)$ copies of the unknown $\ket{\psi}\in \C_\xi$.
\end{lemma}

\begin{proof}
To prove this, we first build an $2\varepsilon$-packing inside $\mathcal C_{\xi}$.  The idea is to take a tensor product of two classes of states, $m$-qubit stabilizer states on $m= n/2$ qubits and a specially constructed subset of states on $n/2$ such that, the newly constructed states have stabilizer extent $\leq \xi$ and the inner product between any two states in this class is at least $\geq \Omega(1)$.

To this end, we first observe the well-known fact that the trace distance between any two $m=n/2$-qubit stabilizer states $\ket{\phi},\ket{\phi'}$ is at least $1/\sqrt{2}$ (i.e., the inner product is at most $1/\sqrt{2}$). Next, we construct a  set of states on $n/2$ qubits. Let $D=2^{n/2}$ and $r=\xi^2$, so that $r\leq \sqrt{D}$ (since $\xi\leq 2^{n/4}$ by assumption).
One can now construct a a set family $
\Fe \subseteq \binom{[D]}{r}$ that satisfies two properties: $(i)$ $|S\cap T| \le  r/2$ for all $S\neq T$ in $\cal F$ and $(ii)$ $\log |\mathcal F|
=\Omega \big(r\log_2(D/r)\big)$ (we prove this in Fact~\ref{lem:setfamilyconstruction} below). For each $S\in\mathcal F$, define the $k$-qubit state
$
|\eta_S\rangle:= \frac{1}{\sqrt r}\sum_{x\in S} |x\rangle.$
Since each computational basis state $|x\rangle$ is a stabilizer state,
$
\xi(\eta_S)\le \sqrt r\le \xi.
$
Moreover, for distinct $S,T\in\mathcal F$, $
\langle \eta_S|\eta_T\rangle
=
{|S\cap T|}/{r}
\le 1/2
$
by construction of $\cal F$. Hence the trace distance between $\ket{\eta_S}$ and $\ket{\eta_T}$ is at least $\sqrt{1-1/4}=\sqrt{3}/2$. Our eventual class of states is now going to be
$$
\mathcal{P}=\{\ket{\phi}\otimes \ket{\eta_S}: \ket{\phi} \text{ is } m\text{-qubit stabilizer}, \eta_S\in \cal F\}.
$$
Every state in $\mathcal P$ belongs to $\mathcal C_{\xi}$ since the stabilizer extent of every $\ket{\phi}\otimes \ket{\eta_S}$ is at most the stabilizer extent of $\eta_S$ which is at most $\xi$. Further, for two distinct states $\ket{\phi}\otimes \ket{\eta_S}$, $\ket{\phi'}\otimes \ket{\eta_T}$  in $\cal P$, observe that \emph{either} $\ket{\phi}\neq \ket{\phi'}$ or $\ket{\eta_S}\neq \ket{\eta_T}$, hence their inner product is at most $1/\sqrt{2}$, hence their trace distance is at least $1/\sqrt{2}$. Hence this forms a $2\varepsilon$-packing inside $\C_\xi$. Furthermore the size of $\cal P$ is 
$$
\Theta(m^2+r\log D/r)=\Theta(n^2+\xi^2\log 1/\xi+n\xi^2),
$$
where we used that $D=2^{n/2},m=n/2,r\leq \xi^2$ in the equality above. 
 At this point, we get our lower bound as follows: suppose we are promised that the unknown state lies in $\mathcal{P}$, every $\varepsilon$-error tomography protocol needs to \emph{identify} the unknown state, one can simply use~\cite[Eq.~(2)]{harrow2012many} to obtain a lower bound of
$
(\log |\Fe_\varepsilon|)/{n}=n+\xi^2.
$
\end{proof}

\begin{lemma}
\label{lem:setfamilyconstruction}
There exists a    set family $
\Fe \subseteq \binom{[D]}{r}$ that satisfies two properties: $(i)$ $|S\cap T| \le  r/2$ for all $S\neq T$ in $\cal F$ and $(ii)$ $\log |\mathcal F|
=\Omega \big(r\log (D/r)\big)$.
\end{lemma}

\begin{proof}
Identify each set $S \in \binom{[D]}{r}$ with its incidence vector
$
x_S \in \{0,1\}^D
$
of Hamming weight $r$. Then for any $S,T \in \binom{[D]}{r}$,
$
d_H(x_S,x_T)=2\bigl(r-|S\cap T|\bigr),
$ where $d_H$ is the Hamming distance.  Hence the condition
$
|S\cap T|\le r/2
$
is equivalent to
$
d_H(x_S,x_T)\ge r.
$ 
So it suffices to construct a binary constant-weight code of length $D$, weight $r$, and minimum distance at least $r$. Let $A(D,2d,r)$ denote the maximum size of a binary constant-weight code of length $D$, weight $r$, and minimum distance at least $2d$. By the  constant-weight Gilbert--Varshamov bound~\cite{levenshtein1971upper}, we have 
$$
A(D,2d,r)\ge
\frac{\binom{D}{r}}{\sum_{i=0}^{d-1}\binom{r}{i}\binom{D-r}{i}}.
$$
Now, choose $d=\lceil r/2\rceil$. Then $2d\ge r$, so any such code corresponds to a family 
$\mathcal F\subseteq \binom{[D]}{r}$ 
such that $|S\cap T|\le r/2$ for all distinct $S,T\in\mathcal F$. It remains to lower bound the size. Since
$$
\sum_{i=0}^{d-1}\binom{r}{i}\binom{D-r}{i}
\le
r\cdot \max_{0\le i\le r/2}\binom{r}{i}\binom{D-r}{i},
$$
we have
$$
\log \sum_{i=0}^{d-1}\binom{r}{i}\binom{D-r}{i}
\le O(r)+\max_{0\le i\le r/2}\Bigl(\log\binom{r}{i}+\log\binom{D-r}{i}\Bigr).
$$
It is not too hard to show that the upper bound is at most $O(r+r\log D/r)$. 
Therefore
$$
\log |\mathcal F|
\;\ge\;
\log \binom{D}{r}
-
\log \sum_{i=0}^{d-1}\binom{r}{i}\binom{D-r}{i}
\;=\;
\Omega \big(r\log(D/r)\big),
$$
as claimed.
\end{proof}

\subsubsection{Further extensions}
On top of stabilizer states, one can also consider the class of states which are generalizations of stabilizer states (yet succinctly learnable and describable) or subclasses of stabilizer states. We briefly mention a few interesting classes below.

\noindent \textbf{High stabilizer dimension.}
Firstly, we observe that the result of~\cite{chen2024stabilizer} gives the following agnostic learner.\footnote{We remark that in~\cite{ad2025structure} they also gave a weak agnostic learning algorithm for states with high stabilizer dimension in time polynomial in $n,2^t,1/\varepsilon$ \emph{assuming} the algorithmic Freiman-Ruzsa conjecture in the high-doubling~regime. }
\begin{theorem}[\cite{chen2024stabilizer}]
Fix any $1 > \tau > \varepsilon > 0$. There is an algorithm that, given access to copies of  $\ket{\psi}$ with
$
\max_{\ket{\phi'} \in \mathcal{S}(n-t)} |\langle {\phi'}|\psi\rangle|^2 > \tau,
$
 outputs a state $\ket{\phi}\in \mathcal{S}(n-t)$ such that
$
|\bra{\phi}\psi\rangle|^2 > \tau - \varepsilon
$
with high probability. 
The algorithm performs single- and two-copy measurements on at most
$
n({2^t}/{\tau})^{O(\log(1/\varepsilon))}
$
copies of $\ket{\psi}$ and runs in time
$
n^3({2^t}/{\tau})^{O(\log(1/\varepsilon))}.
$
\end{theorem}
Note that these states are describable in time $\poly(n,2^t)$ . Like in the proof above on stabilizer states, we again set $\varepsilon \rightarrow \varepsilon/\xi$ and using this we obtain a tomography protocol in time $\poly(n,(2^t\cdot \xi/\varepsilon
)^{\log \xi/\varepsilon
})$ for the following class
$$
\mathcal{C}'_{t,\xi}=\{\ket{\phi}=\sum_{i=1}^k c_i\ket{\psi_i}: \ket{\psi_i}\in \mathcal{S}(n-t) , \sum_i|c_i|\leq \xi\},
$$
i.e., states $\ket{\psi}$ which are expressible as a linear combination of states in $\mathcal{S}(n-t)$

\noindent \textbf{Parity states.} One can also consider the class of parity states, i.e., $\C_{\textsf{par}}=\{\ket{\chi_S}=\frac{1}{\sqrt{2^n}}\sum_x\chi_S(x)\ket{x}:S\in \{0,1\}^n\}$ where $\chi_S(x)=(-1)^{S\cdot x}$. It is not hard to see that one can efficiently learn an unknown $\ket{\psi}\in \C$ using $O(1)$ copies, agnostically learn it also using $O(1)$ copies, describe and prepare these states using $O(n)$ parameters. Furthermore, it was observed in~\cite{adgo2025boosting} that if $f:\{0,1\}^n\rightarrow \{-1,1\}$ is an $s$-term DNF formula, then the quantum state $\ket{\psi_f}=\frac{1}{\sqrt{2^n}}\sum_x f(x)\ket{x}$ satisfies $\Fe_{\C_{\textsf{par}}}\geq 1/s$. Hence using our Theorem~\ref{thm:boostingagnlearning}, we immediately get our agnostic learner for DNF states, one of the main results in~\cite{adgo2025boosting}.

\noindent \textbf{Graph states.} Generalizing the above, one could also consider the class of graph states (also referred to as degree-$2$ phase states), i.e., $\C_{g}=\{\ket{\psi_A}=\frac{1}{\sqrt{2^n}}\sum_{x}(-1)^{x^t A x}\ket{x}: A\in \mathbb{F}_2^{n\times n}\}$. Noting that the set of graph states is a subset of stabilizer states, we can apply the prior tomography algorithm for learning states with low stabilizer extent on states $\ket{\psi}$ of the form $\ket{\psi}=\sum_i c_i \ket{s_i}$ where $\ket{s_i}$ is a graph state. The overall complexity of the tomography algorithm is $\poly(n,(\xi/\varepsilon)^{\log \xi/\varepsilon})$


\subsection{Learning states with low MPS extent}
So far we considered variants of stabilizer states for as our model class. A natural question is could we also consider a model class $\C$ beyond these? To this end, we consider the class of \emph{matrix product states}. A matrix product state with \emph{bond dimension r}
$$
\ket{\psi}=\sum_{i\in \{0,1\}^n}\Tr(A^1_{i_1}\cdots A^n_{i_n})\ket{i_1,\ldots,i_n},
$$
where each matrix $A^j_i$ is an $r\times r$ complex matrix for $i\in \{0,1\},j\in [n]$. We do not discuss this class in detail here and refer the interested reader to~\cite{cramer2010efficient,soleimanifar2022testing,bakshi2024learning} for more on this class of states. Learning matrix product states in the tomography model was considered in the well-known work of~\cite{cramer2010efficient} and in the agnostic setting a recent work of~\cite{bakshi2024learning} showed the following result. 
\begin{theorem}[\cite{bakshi2024learning}]
\label{thm:bakshimpslearning}
Let $\ket{\psi}$ be an $n$-qubit state, and let $\varepsilon,\delta\in(0,1)$. Fix a bond-dimension parameter $r$, and let
$
\mathsf{MPS}_{n,r}
$
denote the class of $n$-qubit matrix product states and bond dimension $r$. Then there is an algorithm that outputs a classical description of a matrix product state
$ |\phi\rangle$
whose bond dimension is
$
n^2\cdot \textsf{poly}(r,1/\varepsilon)$, 
such that, with probability at least $1-\delta$,
$$
|\langle \phi|\psi\rangle|^2\geq 
\max_{|\phi'\rangle\in \textsf{MPS}_{n,r}}
|\langle \phi'|\psi\rangle|^2
-\varepsilon.
$$
The algorithm uses
$
N=\textsf{poly}\!\left(n,r,1/\varepsilon,\log1/\delta\right)
$
copies of $\ket{\psi}$ and
$
\textsf{poly}(N)
$  time.\footnote{We remark that product states are a subset of MPS with bond dimension $r=1$. In~\cite{bakshi2024learning} they also give a \emph{proper} agnostic learner for product states, whereas the MPS agnostic result is improper. For, our main application of learning ``sum of MPS", it wasn't necessary for the agnostic learner to be proper, hence we use the most general result of theirs.}
\end{theorem}

Below we give the proof of Theorem~\ref{result:learn_MPS_extent}. Observe that one can describe an arbitrary $\ket{\psi}$ using $O(nr^2\log(nr/\varepsilon))$ bits of precision (using $\log 1/\varepsilon$ bits one can describe a single bit in $A^j_{i}$, repeating for all $nr^2$ bits in $\{A^j_i\}$ and using a union bound, we get a classical description to $\ket{\psi}$ up to $\varepsilon$-trace distance).  Furthermore, prior works of~\cite{schon2005sequential,malz2024mps} have shown that given the description of an $n$-qubit $\textsf{MPS}$   with bond dimension $r$, there exists a $\poly(n,r)$ classical algorithm that synthesizes a linear-depth polynomial-sized quantum circuit $V$ which outputs this $\textsf{MPS}$ state.

Our main tomography result is for the following class
$$
\mathcal{S}=\{\ket{\phi}=\sum_{i=1}^k c_i\ket{\psi_i}: \ket{\psi_i}\in \textsf{MPS}_{n,r}, \sum_i|c_i|\leq \xi\}.
$$
Note that $\mathcal{S}\subseteq \textsf{MPS}_{n,kr}$ (by Fact~\ref{lem:sumofMPS} proven below), so technically one could have employed the \textsf{MPS} learning algorithms from~\cite{cramer2010efficient,bakshi2024learning,soleimanifar2022testing}  in which case the running time of the algorithm would be $\poly(n,r,k,1/\varepsilon,\log1/\delta)$.  However, if $k=\omega(n,r)$ but $\xi=\poly(n,r)$, their complexity is superpolynomial in $n,r$ and it is unclear if their algorithms can be improved with the extra assumption that $\xi=\poly(n,r)$. To that end, we can use Theorem~\ref{thm:tomo_extent} along with Theorem~\ref{thm:bakshimpslearning} to get $\poly(n,r,\xi,1/\varepsilon,\log 1/\delta)$   algorithm for $\varepsilon$-learning states $\ket{\psi}\in \mathcal{S}$.

\begin{fact}
\label{lem:sumofMPS}
Let
$
\ket{\phi}=\sum_{s=1}^k c_s \ket{\psi_s},
$
where each $\ket{\psi_s}$ is an $n$-qubit MPS of bond dimension $r$. Then $\ket{\phi}$ is an MPS of bond dimension at most $kr$.
\end{fact}

\begin{proof}
Recall that since each $\ket{\psi_s}$ is an MPS of bond dimension $r$, we can write it as
$$
\ket{\psi_s}
=
\sum_{i\in\{0,1\}^n}
\Tr \big(A^{1,s}_{i_1}\cdots A^{n,s}_{i_n}\big)\ket{i_1,\dots,i_n},
$$
with each $A^{j,s}_b\in \mathbb C^{r\times r}$ for $b\in\{0,1\}$ and $j\in[n]$.
For each site $j\in[n]$ and bit $b\in\{0,1\}$, define the block-diagonal matrix
$$
B^j_b
:=
A^{j,1}_b \oplus A^{j,2}_b \oplus \cdots \oplus A^{j,k}_b
\in \mathbb C^{kr\times kr}.
$$
Also absorb the coefficients $c_s$ into the first tensor by setting
$$
C^1_b
:=
(c_1A^{1,1}_b)\oplus(c_2A^{1,2}_b)\oplus\cdots\oplus(c_kA^{1,k}_b),
$$
and for $j\ge 2$ let $
C^j_b:=B^j_b$. 
Consider the MPS
$
\ket{\phi'}
=
\sum_{i\in\{0,1\}^n}
\Tr \big(C^1_{i_1}C^2_{i_2}\cdots C^n_{i_n}\big)\ket{i_1,\dots,i_n}.
$ 
Since products of block-diagonal matrices remain block-diagonal, we have
$$
C^1_{i_1}C^2_{i_2}\cdots C^n_{i_n}
=
\bigoplus_{s=1}^k
c_s A^{1,s}_{i_1}A^{2,s}_{i_2}\cdots A^{n,s}_{i_n}.
$$
Taking the trace gives
$
\Tr\big(C^1_{i_1}\cdots C^n_{i_n}\big)
=
\sum_{s=1}^k c_s\,
\Tr \big(A^{1,s}_{i_1}\cdots A^{n,s}_{i_n}\big).
$ 
Therefore
\begin{align*}
\ket{\phi'}
&=
\sum_{i\in\{0,1\}^n}
\sum_{s=1}^k c_s\,
\Tr\big(A^{1,s}_{i_1}\cdots A^{n,s}_{i_n}\big)\ket{i_1,\dots,i_n}\\
&=
\sum_{s=1}^k c_s
\sum_{i\in\{0,1\}^n}
\Tr\big(A^{1,s}_{i_1}\cdots A^{n,s}_{i_n}\big)\ket{i_1,\dots,i_n}\\
&=
\sum_{s=1}^k c_s\ket{\psi_s}
=
\ket{\phi}.
\end{align*}
Thus $\ket{\phi}$ is an MPS of bond dimension at most $kr$.
\end{proof}


\bibliographystyle{alpha}
\bibliography{references}

\appendix

\section{Barrier in using the orthogonal projector}\label{appsec:barrier_orthoproj}
In this section of the appendix, we show that the Gram matrix over the learned states of Algorithm~\ref{algo:structure_learning} would be very stiff if one utilized the orthogonal projector instead of the ridge projector. As the preparation of the residual state (Claim~\ref{claim:prep_ridgeproj_residual_state}) depends on the condition number of the Gram matrix when using the orthogonal projector, this would lead to an ineffecient algorithm. To observe this, let us consider iteration $t \leq \kappa$ in Algorithm~\ref{algo:structure_learning}.

\begin{claim}\label{claim:norm_residual}
Let $\varepsilon_s \in (0,1)$. Consider the context of Theorem~\ref{thm:structure_learning} and set $\lambda = 0$ in Algorithm~\ref{algo:structure_learning} which corresponds to utilizing the orthogonal projector. Suppose we stop after $\kappa$ many iterations and we are currently at the beginning of iteration $t \leq \kappa$. Then, $\forall j \in [t-1]$, the following is true
$$
\norm{\ket{\phi_j} - \sum_{i=1}^{j-1} \la \phi_i'| \phi_j \ra \ket{\phi_i'}}^2 \geq \eta(\varepsilon_s)/2,
$$
where $\{\ket{\phi_i'}\}$ are the states obtained after running Gram-Schmidt orthonormalization on the set of states $\{\ket{\phi_i}\}$.
\end{claim}
\begin{proof}
Let us denote $\uptau = \eta(\varepsilon_s)/2$ for brevity. We first note that the norm
\begin{align}
    \norm{\ket{\phi_j} - \sum_{i=1}^{j-1} \la \phi_i'| \phi_j \ra \ket{\phi_i'}}^2 &= \left( \bra{\phi_j} - \sum_{i=1}^{j-1} \la \phi_j| \phi_i' \ra \bra{\phi_i'} \right) \left(\ket{\phi_j} - \sum_{i=1}^{j-1} \la \phi_i'| \phi_j \ra \ket{\phi_i'} \right) \\
    &= 1 - \sum_{i=1}^{j-1} \la \phi_j | \phi_i' \ra \la \phi_i'| \phi_j \ra - \sum_{i=1}^{j-1} \la \phi_i' | \phi_j \ra \la \phi_j| \phi_i' \ra + \sum_{i,\ell=1}^{j-1} \la \phi_j| \phi_i' \ra \la \phi_\ell'| \phi_j \ra \la \phi_i' | \phi_\ell' \ra \\
    &= 1 - \sum_{i=1}^{j-1} |\la \phi_i'| \phi_j \ra|^2,
    \label{eq:gram_schmidt_norm_expression}
\end{align}
where we used $\la \phi_i'|\phi_\ell'\ra=0$ for $i \neq \ell$ in the last line. It is then enough to prove that $1 - \sum_{i=1}^{j-1} |\la \phi_i'| \phi_j \ra|^2 \geq \eta$ for all $j \in [t-1]$ in iteration $t$ to show the above claimed result. We will prove this by induction. 

Consider iteration $t=1$ and we proceed to carrying out weak agnostic learning. Then, we will learn $\ket{\phi_1} \in \C$ such that $|\la \phi_1| \psi \ra|^2 \geq \uptau$ (Claim~\ref{claim:promise_WAL}). At the beginning of iteration $t=2$, the claimed statement holds trivially as $\norm{\ket{\phi_1}} = 1$. Suppose again, we proceed through iteration $t=2$, then we learn the state $\ket{\phi_2} \in \C$ such that $|\la \phi_2 | \psi_2 \ra|^2 \geq \uptau$, where
$$
\ket{\psi_2} = \frac{\ket{\psi} - \la \psi | \phi_1 \ra \ket{\phi_1}}{\norm{\ket{\psi} - \la \psi | \phi_1 \ra \ket{\phi_1}}}
$$
is the residual state at the end of iteration $t=1$. As $\la \phi_1 | \psi_2 \ra = 0$, we can express
$$
\ket{\phi_2} = \gamma_1 \ket{\phi_1} + \gamma_2 \ket{\psi_2} + \gamma_3 \ket{\sigma},
$$
where $\ket{\sigma}$ is some $n$-qubit state orthogonal to both $\ket{\phi_1}$ and $\ket{\psi_2}$. The coefficients are given by $\gamma_1 = \la \phi_1 | \phi_2 \ra$, $\gamma_2 = \la \psi_2 | \phi_2 \ra$ and $\gamma_3 = \la \sigma | \phi_2 \ra$. As $\ket{\phi_2}$ is a valid normalized state, we have
\begin{align}
    |\gamma_1|^2 + |\gamma_2|^2 + |\gamma_3|^2 = 1 \implies 1 - |\gamma_1|^2 & = |\gamma_2|^2 + |\gamma_3|^2 \\
    1 - |\la \phi_1 | \phi_2 \ra|^2 &\geq |\la \psi_2 | \phi_2 \ra|^2 \\
    1 - |\la \phi_1 | \phi_2 \ra|^2 &\geq \uptau,
    \label{eq:interim_high_norm_t2}
\end{align}
where we have used the definitions of $\gamma_1,\gamma_2$ in the second line and $|\la \phi_2 | \psi_2 \ra|^2 \geq \uptau$ in the last line. Noting that $\ket{\phi_1'} = \ket{\phi_1}$ and Eq.~\eqref{eq:gram_schmidt_norm_expression} proves the claimed statement for $t=2$. 

Let us now proceed to end of iteration $t=3$ before we complete the proof by induction. At the end of iteration $t=3$, we would have learned $\ket{\phi_3} \in \C$ such that $|\la \phi_3 | \psi_3 \ra|^2 \geq \uptau$. We can now express $\ket{\phi_3}$ as
$$
\ket{\phi_3} = \gamma_1 \ket{\phi_1'} + \gamma_2 \ket{\phi_2'} + \gamma_3 \ket{\psi_3} + \gamma_4 \ket{\sigma},
$$
where $\ket{\sigma}$ is some $n$-qubit state orthogonal to $\ket{\phi_1'},\ket{\phi_2'},\ket{\psi_3}$. Note that $\ket{\phi_2'}$ is the state obtained from applying Gram-Schmidt orthonormalization to $\ket{\phi_2}$ considering the other basis member $\ket{\phi_1'}=\ket{\phi_1}$. We now proceed as before. As $\ket{\phi_3}$ is a valid normalized state, we have
\begin{align*}
    \sum_{i \in [4]} |\gamma_i|^2 = 1 \implies 1 - |\gamma_1|^2 - |\gamma_2|^2 &\geq |\gamma_3|^2 + |\gamma_4|^2 \\
    1 - |\la \phi_1' | \phi_3 \ra|^2 - |\la \phi_2' | \phi_3 \ra|^2 &\geq |\la \psi_3 | \phi_3 \ra|^2 \\
    1 - |\la \phi_1' | \phi_3 \ra|^2 - |\la \phi_2' | \phi_3 \ra|^2 &\geq \uptau,
\end{align*}
where we have used the definitions of $\gamma_i$ for $i \in [4]$ in the second line and $|\la \phi_3 | \psi_3 \ra|^2 \geq \uptau$ in the last line. Applying Eq.~\eqref{eq:gram_schmidt_norm_expression} proves the claimed statement for $t=2$. 

Suppose now, we are at the end of iteration $t$ and we have proceeded through weak agnostic learning. We would have learned $|\la \phi_t | \psi_t \ra|^2 \geq \uptau$. The claim regarding $\norm{\ket{\phi_j} - \sum_{i=1}^{j-1} \la \phi_i'| \phi_j \ra \ket{\phi_i'}}^2 \geq \uptau$ follows from previous iterations up to end of $t-1$ and noting that $\ket{\phi_i'}$ for $i \in [t-1]$ does not change as we move to iteration $t$. So, we just need to show the result for $\norm{\ket{\phi_t} - \sum_{i=1}^{t-1} \la \phi_i'| \phi_j \ra \ket{\phi_i'}}^2$. Towards that, we express $\ket{\phi_t}$ as
$$
\ket{\phi_t} = \Big(\sum_{i=1}^{t-1} \gamma_i \ket{\phi_i'}\Big) + \gamma_t \ket{\psi_t} + \gamma_{t+1} \ket{\sigma},
$$
where $\ket{\sigma}$ is some $n$-qubit state orthogonal to $\ket{\phi_i'}$ for all $i \in [t-1]$ and $\ket{\psi_t}$. The coefficients $\gamma_i = \la \phi_i'|\phi_t\ra$ for $i \in [t-1]$, $\gamma_t = \la \psi_t | \phi_t \ra$ and $\gamma_{t+1} = \la \sigma | \phi_t \ra$. As $\ket{\phi_t}$ is a valid normalized state, we have
\begin{align*}
    \sum_{i \in [t+1]} |\gamma_i|^2 = 1 \implies 1 - \sum_{i=1}^{t-1} |\gamma_i|^2 &= |\gamma_t|^2 + |\gamma_{t+1}|^2 \\
    1 - \sum_{i=1}^{t-1} |\la \phi_i' | \phi_t \ra|^2 &\geq |\la \psi_t | \phi_t \ra|^2 \\
    1 - \sum_{i=1}^{t-1} |\la \phi_i' | \phi_t \ra|^2 &\geq \uptau,
\end{align*}
where we have used the definitions of $\gamma_i$ for $i \in [t+1]$ in the second line and $|\la \phi_t | \psi_t \ra|^2 \geq \uptau$ in the last line. Applying Eq.~\eqref{eq:gram_schmidt_norm_expression} proves the claimed statement for arbitrary $t$. This completes the proof.
\end{proof}

We now show that the Gram matrix $G$ could be very stiff.
\begin{claim}
Let $\varepsilon_s \in (0,1)$. Consider the context of Theorem~\ref{thm:structure_learning} and set $\lambda = 0$ in Algorithm~\ref{algo:structure_learning} which corresponds to utilizing the orthogonal projector. Then, the condition number of the Gram matrix at the end of the algorithm is $\Big(1/(\varepsilon_s \eta(\varepsilon_s)^2)\Big)^{O(1/(\varepsilon_s \eta(\varepsilon_s)))}$.
\end{claim}
\begin{proof}
Let $\{\ket{\phi_i'}\}$ be the states obtained after running Gram-Schmidt orthonormalization on the list $\{\ket{\phi_i}\}$. We now note that the determinant of the Gram matrix over $L$ which we denote by $G$ satisfies
$$
\det(G) = \prod_{j=1}^\kappa \norm{(\id - \Pi_{j-1}) \ket{\psi}}^2 \geq \uptau^\kappa,
$$
where $\Pi_{j-1} = \sum_{i=1}^{j-1} \ket{\phi_i'}\bra{\phi_i'}$ and the second inequality follows from Claim~\ref{claim:norm_residual} with $\uptau = \eta(\varepsilon_s)/2$. Noting that $\det(G)$ is the product over all the eigenvalues of $G$ and that each eigenvalue is at most $\Tr(G) = \kappa$, we also have
$$
\det(G) \leq \lambda_\min(G) \cdot \kappa^{\kappa - 1}.
$$
Combining the two inequalities above gives us that 
$$
\lambda_{\min}(G) \geq \frac{\uptau^\kappa}{\kappa^{\kappa - 1}} \implies \frac{\lambda_{\max}(G)}{\lambda_{\min(G)}} \leq \left(\frac{\kappa}{\uptau}\right)^\kappa = \Big(1/(\varepsilon_s \eta(\varepsilon_s)^2)\Big)^{O(1/(\varepsilon_s \eta(\varepsilon_s)))},
$$
where we have noted that $\kappa = O(1/(\varepsilon_s \eta(\varepsilon_s))$ and the definition of $\uptau$. This completes the proof.
\end{proof}
\end{document}